\documentclass{tac}

\usepackage{graphicx}
\usepackage{amsmath,amssymb,url}
\usepackage{tikz,pgfplots}
\usepackage{tikz-cd}
\usepackage{float,subcaption}
\usetikzlibrary{arrows,calc,decorations,decorations.markings,fadings,positioning,patterns,shapes
	,decorations.pathmorphing,backgrounds,fit,petri,mindmap,shadows,positioning,matrix}
\usepackage{multirow,tabularx,makecell}
\usepackage [all,2cell]{xy}
\usepackage{pigpen}
\usepackage{caption}
\newcommand{\xyline}[2][]{\ensuremath{\smash{\xymatrix@1#1{#2}}}}
\makeatother
\newdir{ >}{{}*!/-7.5pt/@{>}}
\makeatletter
\xyoption{2cell}
\UseAllTwocells

\newcommand{\po}{\ar@{}[dr]|{\text{\pigpenfont R}}}
\newcommand{\pb}{\ar@{}[dr]|{\text{\pigpenfont J}}}

\newtheorem{theorem}{Theorem}
\newtheorem{corollary}{Corollary}
\newtheorem{definition}[theorem]{Definition}
\newtheorem{proposition}[theorem]{Proposition}

\newtheorem{remark}[theorem]{Remark}
\newtheorem{lemma}[theorem]{lemma}

\makeatletter
\def\eqalign#1{\null\,\vcenter{\openup\jot\m@th
		\ialign{\strut\hfil$\displaystyle{##}$&$\displaystyle{{}##}$\hfil
			\crcr#1\crcr}}\,}
\makeatother

\def\Tau{{\text{T}}}
\def\dom{\mathop{\text{dom}}} 
\def\Cal#1{{\cal#1}}
\def\family#1#2#3{\langle#3\rangle_{#1\in#2}}
\def\familyi#1#2{\family{i}{#1}{#2}}
\def\familyiI#1{\familyi{I}{#1}}
\def\sequence#1#2{\family{#1}{\Bbb N}{#2}}
\def\sequencen#1{\sequence{n}{#1}}
\def\eae{=_{\text{a.e.}}}
\def\leae{\le_{\text{a.e.}}}
\def\geae{\ge_{\text{a.e.}}}
\def\restr{\mathchoice
	{\hbox{$\restriction$}\mskip 1mu}
	{\hbox{$\restriction$}\mskip 1mu}
	{\hbox{$\scriptstyle\restriction$}\mskip 1mu}
	{\hbox{$\scriptscriptstyle\restriction$}\mskip 1mu}}
\def\ssbullet{{\scriptscriptstyle\bullet}}
\def\symmdiff{\triangle}

\def\coint#1{\left[#1\right[}

\def\ssbullet{{\scriptscriptstyle\bullet}}

\long\def\inset#1{{\narrower\narrower{\vskip 1pt plus 1pt minus 0pt}
		#1
		{\vskip 1pt plus 1pt minus 0pt}}}

\def\Bvalue#1{\mathchoice
	{\hbox{$[\![#1]\!]$}}
	{\hbox{$[\![#1]\!]$}}
	{\hbox{$\scriptstyle[\![#1]\!]$}}
	{\hbox{$\scriptscriptstyle[\![#1]\!]$}}}

\tikzset{->-/.style={decoration={
			markings,
			mark=at position .5 with {\arrow{>}}},postaction={decorate}}}
\makeatletter
\tikzoption{canvas is plane}[]{\@setOxy#1}
\def\@setOxy O(#1,#2,#3)x(#4,#5,#6)y(#7,#8,#9)%
{\def\tikz@plane@origin{\pgfpointxyz{#1}{#2}{#3}}%
	\def\tikz@plane@x{\pgfpointxyz{#4}{#5}{#6}}%
	\def\tikz@plane@y{\pgfpointxyz{#7}{#8}{#9}}%
	\tikz@canvas@is@plane
}
\makeatother 
\tikzstyle{block} = [draw, rectangle, text width=2.0cm, text centered, minimum height=1.2cm, node distance=4cm,fill=white]
\tikzstyle{container} = [draw, rectangle, inner sep=0.5cm, fill=gray,minimum height=3cm]
\def\bottom#1#2{\hbox{\vbox to #1{\vfill\hbox{#2}}}}
\tikzset{
	mybackground/.style={execute at end picture={
			\begin{scope}[on background layer]
				\node[] at (current bounding box.north){\bottom{1cm} #1};
			\end{scope}
	}},
}

\definecolor{pixel 0}{HTML}{54FF00}
\definecolor{pixel 1}{HTML}{FFFFFF}
\definecolor{pixel 2}{HTML}{FF0000}
\definecolor{pixel 3}{HTML}{0048FF}
\definecolor{pixel 4}{HTML}{000000}

\address{Department of Electrical Engineering,\\
	Indian Institute of Technology Delhi, \\
	New Delhi-110016, INDIA.\\[5pt]
	Department of Electrical Engineering,\\
	Indian Institute of Technology Delhi, \\
	New Delhi-110016, INDIA.\\
}
\eaddress{salil.samant@ee.iitd.ac.in\CR sdjoshi@ee.iitd.ac.in}

\title{Unified Functorial Signal Representation III: Foundations, Redundancy, $L^0$ and $L^2$ functors}
\author{Salil Samant and Shiv Dutt Joshi}

\begin{document}

\maketitle
\begin{abstract}
	In this paper we propose and lay the foundations of a functorial framework for representing signals. By incorporating additional category-theoretic relative and generative perspective alongside the classic set-theoretic measure theory the fundamental concepts of redundancy, compression are formulated in a novel authentic arrow-theoretic way. The existing classic framework representing a signal as a vector of appropriate linear space is shown as a special case of the proposed framework.
	Next in the context of signal-spaces as a categories we study the various covariant and contravariant forms of $L^0$ and $L^2$ functors using categories of measurable or measure spaces and their opposites involving Boolean and measure algebras along with partial extension. Finally we contribute a novel definition of intra-signal redundancy using general concept of isomorphism arrow in a category covering the translation case and others as special cases. Through category-theory we provide a simple yet precise explanation for the well-known heuristic of lossless differential encoding standards yielding better compressions in image types such as line drawings, iconic image, text etc; as compared to classic representation techniques such as JPEG which choose bases or frames in a global Hilbert space.
	
\end{abstract}

\section{Introduction}
\label{intro}
Signal representation lies at the heart of how information is represented in signals making it fundamental to many broad range of applications \cite{kay} such as high fidelity music reproduction, communications, medical imaging, speech processing, radar and sonar, and oil prospecting. This also makes it intertwined with modern representation learning \cite{Bengio} which deals with the problem of effective data representation. Classically a signal is viewed as an entity varying naturally in time or space and modeled as an element of a linear function space. Fixed mathematical structures on complete domain ($\Bbb R^n$) such as measure and topology along with symmetry are exploited by invoking a suitable group action on the signal space; thereby utilizing the techniques from functional and harmonic analysis developed in the last century; refer \cite{heilwalnut} and references therein.

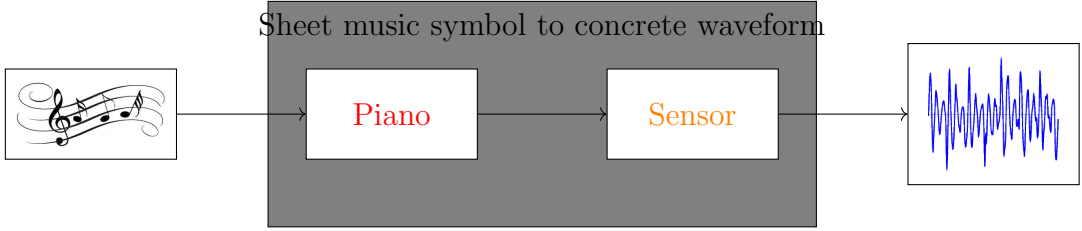
\begin{figure}[ht]
	\centering
	
	\begin{tikzpicture}[mybackground={Sheet music symbol to concrete waveform}]
	
	
	\node [block, name=music] {\includegraphics[width=.99\textwidth]{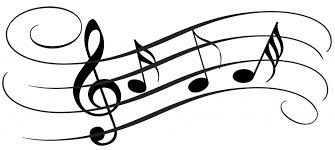}};
	\node [block, right of=music] (GI) {\color{red} Piano};
	\node [block, right of=GI] (SI) {\color{orange} Sensor};
	\node[block, right of=SI] (S) {\includegraphics[width=.90\textwidth]{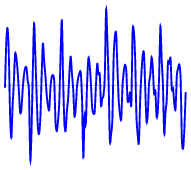}};
	\begin{scope}[on background layer]
	\node [container,fit=(GI) (SI)] (container) {};
	\end{scope}
	\draw [->] (music) -- (GI);
	\draw [->] (GI) -- node {} (SI);
	\draw [->] (SI) -- node {} (S);
	
	\end{tikzpicture}
	\caption{Functorial system transforming abstract structured melodies to concrete music signals.}
	\label{fig:generate_music}
\end{figure}

\begin{figure}[ht]
	\centering
	
	\begin{tikzpicture}[mybackground={Reflected light to concrete waveform}]
	
	
	\node [block, name=object] {Physical object};
	\node [block, right of=object] (GI) {\color{red} Reflected light from object};
	\node [block, right of=GI] (SI) {\color{orange} Sensor};
	\node[block, right of=SI] (S) {\includegraphics[width=.90\textwidth]{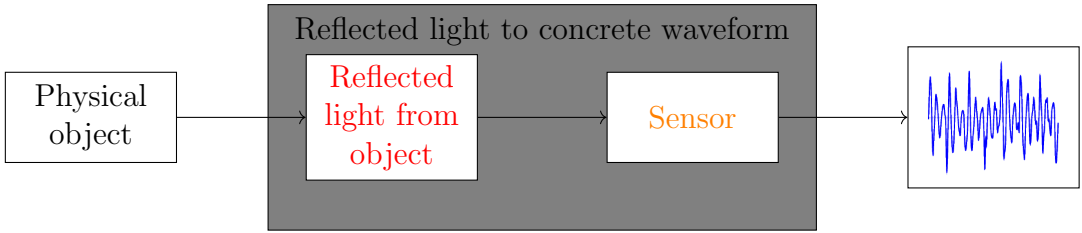}};
	\begin{scope}[on background layer]
	\node [container,fit=(GI) (SI)] (container) {};
	\end{scope}
	\draw [->] (music) -- (GI);
	\draw [->] (GI) -- node {} (SI);
	\draw [->] (SI) -- node {} (S);
	
	\end{tikzpicture}
	\caption{Functorial system transforming reflected light with object structure to concrete image signals.}
	\label{fig:generate_image}
\end{figure}

\begin{figure}[ht]
	\centering
	
	\begin{tikzpicture}[mybackground={Vocal formulation to Concrete waveform}]
	
	
	\node [block, name=speech] {Linguistic word};
	\node [block, right of=speech] (GI) {\color{red} Speech formulation};
	\node [block, right of=GI] (SI) {\color{orange} Sensor};
	\node[block, right of=SI] (S) {\includegraphics[width=.90\textwidth]{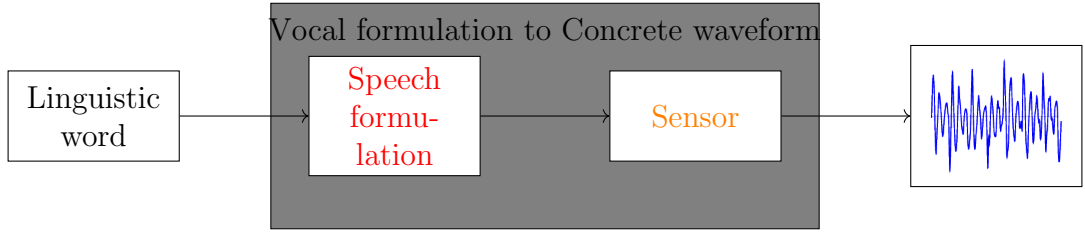}};
	\begin{scope}[on background layer]
	\node [container,fit=(GI) (SI)] (container) {};
	\end{scope}
	\draw [->] (music) -- (GI);
	\draw [->] (GI) -- node {} (SI);
	\draw [->] (SI) -- node {} (S);
	
	\end{tikzpicture}
	\caption{Functorial system transforming words to concrete speech signals.}
	\label{fig:generate_speech}
\end{figure}

The generative theory of Leyton~\cite{Leyton01} argues, that in human psychology various shapes and objects such as sheet music symbols or their aggregates forming melodies are highly structured with maximal transfer of previously occurring objects. The theory also intuitively demonstrates that in general various naturally generated shapes all have these kind of structures. Motivated by this generative intuition, in this sequel we utilize the mathematical structure of a groupoid\footnote{An abstract category $\mathbf{C}$ consists of a collection $X,Y,Z ...$ of objects denoted by $\mbox{Ob}(\mathbf{C})$, for every pair $X,Y \in \mbox{Ob}(\mathbf{C})$, a collection $\mathbf{C}(X,Y)=\{f:X\to Y\mid X,Y\in\mbox{Ob}(\mathbf{C})\}$ of morphisms, for each $X \in \mbox{Ob}(\mathbf{C})$, an identity morphism $\mathbf{id}_X\ \mbox{or}\ \mathbf{1}_X\ :X\to X$ and composition $\mathbf{C}(Y,Z)\ \times\ \mathbf{C}(X,Y)\mapsto\mathbf{C}(X,Z)$, i.e $(g,f)\mapsto g\circ f$, satisfying the unit law for a morphism $f:X \rightarrow Y$, $\mathbf{id}_Y\circ f=f=f\circ \mathbf{id}_X$ and usual associativity for $X \xrightarrow[]{f} Y \xrightarrow[]{g} Z \xrightarrow[]{h} W$, $h\circ (g\circ f)=(h\circ g)\circ f$. A groupoid is a simply a category in which every morphism is invertible.} to model objects carrying certain structures or properties and isomorphisms as their transfer. A heuristic discussion demonstrating advantages of groupoid and fibration over wreath products for such generative structure can be found in~\cite{salilp2} of sequel.

There are two equivalent ways of incorporating a generative groupoid model in signal representation. The first easier way for a reader with just signal theory and functional analysis background is using a simple functor; while the second more categorical way for a reader reasonably well acquainted with (higher) category theory is using a functor-category model where generators are directly modeled as functors and transfers are captured through natural transformations; see Section~\ref{meas2}.

In the easier way, referring Figures~\ref{fig:generate_music},~\ref{fig:generate_image} and~\ref{fig:generate_speech} we think of the shaded block comprising of generating and sensing system as a functor. It transforms a groupoid (or more generally a category) generative structure into a concrete signal as image subcategory. This model works ideally for the generators which generate waveforms well-separated in space or time where we can expect the functor to be faithful and injective on objects for aggregate signal. In general this functor is not faithful due to superposition of individual waveforms and various other reasons. Nevertheless the concept of arrow is a natural setting to model redundancy of signals which makes the functorial model indispensable. Moreover one can utilize the usual set-theoretic measure theory and functional analysis along with category-theoretic generative structure through concept of trivial categorification as studied in~\cite{salilp1}. Now for certain subcategories the functor remains faithful which is discussed later in Section~\ref{limitation}. Using these categories we model it as faithful isomorphism-preserving functor $F:\mathbf{C} \to \mathbf{D}$. As an example if we denote sheet music symbols (or their aggregates forming melodies) as objects $G_1, G_2, G_3,...$ with transfers as $a_1, a_2,...$ forming category $\mathbf{C}$ then signal generation mechanism becomes functorial as shown in Equation~\ref{eq:generation_functor} where the individual output waveforms are represented by $FG_1, FG_2, FG_3,...$ while their preserved relationships are captured  by arrows $Fa_1, Fa_2,...$ forming category $F\mathbf{C}$.    

\begin{equation}
\label{eq:generation_functor}
\xymatrix @R=0.4in @C=0.8in{
	G_1 \ar@(dl,ul)[]|{id_{G_1}} \ar[r]^{a_1} \ar[d]_{a_3 \cdot a_2 \cdot a_1} \ar[dr]_{a_2 \cdot a_1} & G_2 \ar@(dr,ur)[]|{id_{G_2}}  \ar[d]^{a_2} \\
	G_4 \ar@(dl,ul)[]|{id_{G_4}} & G_3 \ar@(dr,ur)[]|{id_{G_3}} \ar[l]_{a_3}
}
\xymatrix @R=0.4in @C=0.8in{
	FG_1 \ar@(dl,ul)[]|{id_{FG_1}} \ar[r]^{Fa_1} \ar[d]_{Fa_3 \circ Fa_2 \circ Fa_1} \ar[dr]_{Fa_2 \circ Fa_1} & FG_2 \ar@(dr,ur)[]|{id_{FG_2}}  \ar[d]^{Fa_2} \\
	FG_4 \ar@(dl,ul)[]|{id_{FG_4}} & FG_3 \ar@(dr,ur)[]|{id_{FG_3}} \ar[l]_{Fa_3}
}
\end{equation}  

Practically the natural generators are very complex structured abstract objects having multiple structures. Depending on a particular structure of interest within the observed signal, a concrete codomain category is chosen. Then the functor $F$ is constrained to be faithful and isomorphism-preserving. Being faithful and isomorphism-preserving ensures that using isomorphisms in the category $F\mathbf{C}$ we can uniquely infer transfers (isomorphisms) between generators (objects) thereby directly gaining some insight into generative mechanism of source. In this work, we shall be particularly interested  in measurable and measure-preserving structures motivated by translational, scaling, amplitude redundancies very common in classic image, audio and speech signals. This interested is also attributed to connections with classical representation techniques in spaces of measurable functions such as $L^2(\Bbb R^n)$. As mentioned earlier unfortunately in real world scenarios the equipments are never ideal and further the superposition of waveforms in observed signal makes the functor non-faithful. This puts certain limitations on groupoid in the domain category that could be inferred from the observed signal which we study in Section~\ref{limitation} along with possible solutions and work-around. Interestingly by using the classic set-theoretic measure theory in addition to pure category theory certain limitations could be effectively overcome by using novel concept of trivial categorification; introduced in~\cite{salilp1} of the sequel. Loosely speaking this treats objects as trivial categories and therefore we can utilize additional set-theoretic properties of objects independently along with treating them simply as objects of enclosing category. As an example, in the context of this paper a measurable function $f:(I,\Sigma_{I}) \to (\Bbb R,\Sigma_{\Cal B})$ which is an object of $\mathbf{Meas}^{\rightarrow}$ is also a trivial category and therefore by considering additional property of $\Bbb R$ being a field, it is can be point-wise added to and multiplied by any other measurable function $g:(I,\Sigma_{I}) \to (\Bbb R,\Sigma_{\Cal B})$. In other words, $f$ is simultaneously an element of Riesz space $\mathfrak{L^0}_{I}$ and this property is independent of it being an object of $\mathbf{Meas}^{\rightarrow}$ which recognizes only measurable structure on $\Bbb R$. This novel concept of using set-theory alongside category-theory simultaneously for signal representation is explored in this paper.

\subsection{Motivation}
\label{sec:motiv}
In this section, we first visually motivate the functorial viewpoint of signal through Figures~\ref{fig:visual_depiction_signal} and~\ref{fig:space_overview} and then we summarize the major differences of proposed functorial framework from classic representation in Table~\ref{table:summary_sig_rep} which stand as prime motivations for a functorial framework . These are explained later in Section~\ref{sec:funcsigspace}. The fundamental concept of treating a signal as functor is not new and is first explored in the context of topology in~\cite{robinson1}. However there are some subtle differences in this work as discussed in Section~\ref{sigfunc}. Nevertheless it stands as third motivation in addition to category-theoretic relative perspective of~\cite{SGA1} and generative intuition of~\cite{Leyton01} towards developing this framework in sequel.

Figure~\ref{fig:visual_depiction_signal} gives a visual depiction of proposed functorial signal representation framework. We enumerate the salient features of this framework as follows; 

\begin{enumerate}
	\item Complete signal (a measurable function) $f = ... \amalg (f\restr I) \amalg (f\restr K) \amalg (f\restr J) \amalg ... $ is naturally the coproduct in category $\mathbf{Meas}^{\rightarrow}$ or functor category $\mathbf{Meas}^{\mathbf{2}}$, where subobjects $(f\restr I),(f\restr K),(f\restr J),...$ are local real-valued partial functions on disjoint half-open intervals $I$, $J$, $K$, ....
	\item The underlying generators(capturing the intuition of generative theory in~\cite{Leyton01}) of a signal are directly modeled either as functors $G_1$,$G_2$,$G_3$,... or else as objects of base category $\mathbf{C}$. The transfers (or isomorphisms) between the generators are automatically captured via natural transformations (or natural isomorphisms) or base category arrows $a_1,a_2,...$.
	\item Then functorial representation models the signal either as a functor $F: \mathbf{C} \rightarrow \mathbf{D}$ or as a subcategory of the usual functor category $\mathbf{Meas}^{\mathbf{2}}$   
	\item Whenever $G_1$ and $G_2$ are isomorphic, the corresponding subobjects $(f\restr I)$, $(f\restr J)$ also become isomorphic via $(h,\phi)$ whenever the functor preserves this isomorphism. Then $(f\restr J)$ is naturally viewed as redundant relative to $(f\restr I)$.
	\item By considering a field property on the underlying set of $\Bbb R$, the objects of $\mathbf{Meas}^{\rightarrow}$ such as $(f\restr I)$ are also set-theoretically elements of Riesz spaces $\mathfrak{L}^0_I$ additionally bringing usual set-theoretic measure-theory alongside category theory for signal representation. Properties of null-ideal, measure on $(I,\Sigma_I)$ allows introducing equivalence classes $(f\restr I)^{\ssbullet} \in L^0(I,\Sigma_I,\Cal N(\mu_I))$, $(f\restr I)^{\ssbullet} \in L^2(I,\Sigma_{I},\mu_I)$ through set-theoretic measure-theory.
	\item The arrow $(h,\phi)$ in certain cases induces $T_{(h,\phi^{-1})}:L^2(I,\Sigma_{I},\mu_I) \to L^2(J,\Sigma_{J},\mu_J)$ and signal space becomes a subcategory in $\mathbf{Hilb}$ or $\mathbf{Riesz}$ matched to its generative structure as shown in Figure~\ref{fig:space_overview}.
\end{enumerate}

\begin{figure}
	\resizebox*{\textwidth}{!}{
		\begin{tikzpicture}
		\draw[domain=-1.57:1.57,samples=100] plot(\x,{sin(2*\x r)});
		\draw[domain=7.85:11,samples=100] plot(\x,{sin(2*\x r)});
		\draw plot [smooth] coordinates {(1.57,0) (3,-0.5) (4,0.75) (5,1) (6,-0.4) (7,1) (7.85,0)};
		\draw[<->] (-4.5,0) -- (13,0) node [pos=1,below] {$\mathbb{R}$};
		\draw plot [smooth] coordinates {(-1.57,0) (-2.5,1) (-2.8,0) (-3.5,-0.5) (-3.7,-1) (-4,0)};
		\draw plot [smooth] coordinates {(11,0) (11.5,-0.5) (12,1) (12.5,0.5)};
		\draw[->] (-3.5,-1) --(-3.5,3.5) node [pos=1,left] {$\mathbb{R}$};
		\draw[-,dashed] (-1.57,2.5) -- (-1.57,-1.5);
		\draw[-,dashed] (1.57,2.5) -- (1.57,-1.5);
		
		\draw[-,dashed] (7.85,2.5) -- (7.85,-1.5);
		\draw[-,dashed] (11,2.5) -- (11,-1.5);
		\draw[<->] (-1.57,-1.3) --(1.57,-1.3) node[midway,fill=white] {$I$};
		\draw[<->] (1.57,-1.3) --(7.85,-1.3) node[midway,fill=white] {$K$};
		\draw[<->] (7.85,-1.3) --(11,-1.3) node[midway,fill=white] {$J$};
		\node at (0,1.3) [draw=none] {$f\restr I = F(G_1)$};
		\node at (0,2.3) [draw=none] {Subspace $\mathfrak{L}^0_I$};
		\node at (5,1.3) [draw=none] {$f\restr K = F(G_3)$};
		\node at (5,2.3) [draw=none] {Subspace $\mathfrak{L}^0_K$};
		\node at (9.4,1.3) [draw=none] {$f\restr J = F(G_2)$};
		\node at (9.4,2.3) [draw=none] {Subspace $\mathfrak{L}^0_J$};
		\node at (5,3.3) [draw=none] {Global space $\mathfrak{L}^0_{\Bbb R}$};
		\node at (-2.5,0) [draw=none,above left] {$(0,0)$};
		\node at (12.5,0) [draw=none,above right] {$t$};
		\node at (-3.5,1.5) [draw=none,above left] {$f(t)$};
		\end{tikzpicture}
	}
	\caption{Functorial signal representation: Global Signal $f$ is chopped into naturally related local sub-signals ..$(f\restr I)$,$(f\restr K)$,.. thought of lying in a category.}
	\label{fig:visual_depiction_signal}
\end{figure}
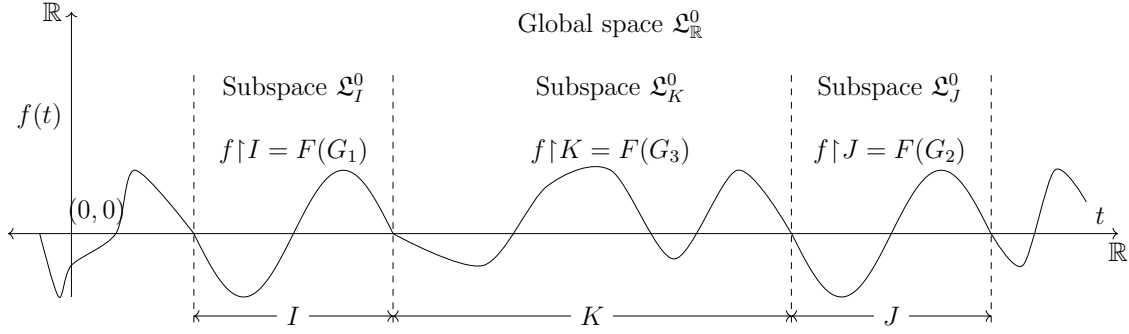

The proposed mathematical model of a functor for natural generating mechanism of a signal offers various additional benefits when compared to classic representation techniques:  

\begin{itemize}
	\item {\bf Signal source as base category}: In classic case, the mathematical structure of a source generating a signal carrying information of interest is usually not taken into consideration. From the perspective of information theory it is treated as either memoryless or with memory. By modeling a source with memory as a groupoid in tune with generative intuition we seek to capture isomorphic relationships between waveforms generated by the source directly impacting the amount of perceived information in signal. The memoryless source as a special case having no interdependencies of successive messages is modeled as a discrete category.   
	
	\item {\bf Redundancy using relative perspective:} Classic redundancy is defined as ratio of actual rate of source to its absolute rate and indicates the maximum possible data compression ratio by which bits in its efficient representation can be decreased. In functorial model, the relative perspective offered by category theory provides an authentic tool to model interdependence between messages or sub-signals. This leads to arrow-theoretic structural definition of redundancy. It also becomes possible to understand compression in a natural category-theoretic way.
	
	\item {\bf Signal space as a category}: The generic signal or message spaces such as $L^2(\Bbb R^n)$ cater mainly to memoryless sources since the linearly independent basis best model sub-signals or messages which are independent. In functorial model observed signal becomes an image subcategory inducing additional category structure on domain such as $\Bbb R^n$. The Signal space in certain cases can be modeled as a category matched to the generative structure of the signal and therefore unique to every signal.    
\end{itemize}

\begin{figure}
	\resizebox*{\textwidth}{!}{
		\begin{tikzpicture}
		[object/.style={circle,draw=blue!80,fill=blue!20,thick,
			inner sep=0pt,minimum size=6mm},
		object2/.style={circle,draw=red!80,fill=red!20,thick,
			inner sep=0pt,minimum size=4mm}]
		
		\tiny
		
		\draw [->] (-3,0) -- (-2,0);
		\draw [->]  (-3,0) -- (-3,1,0);
		\draw [->] (-3,0) -- (-3,0,1);
		\node at (-3,1.2) {$L^2(I,\Sigma_{I},\mu_I)$};
		\node[below] at (-3,0) {$0$};
		\draw[->,color=red] (-3,0) -- (-2.3,0.5) node[above] {$(f\restr I)^{\ssbullet}$};
		
		\draw [->] (0,0) -- (1,0);
		\draw [->]  (0,0) -- (0,1,0);
		\draw [->] (0,0) -- (0,0,1);
		\node at (0,1.2) {$L^2(K,\Sigma_{K},\mu_K)$};
		\node[below] at (0,0) {$0$};
		\draw[->,color=black] (0,0) -- (-0.5,0.5) node[above] {$(f\restr K)^{\ssbullet}$};
		
		\draw [->] (3,0) -- (4,0);
		\draw [->]  (3,0) -- (3,1,0);
		\draw [->] (3,0) -- (3,0,1);
		\node at (3,1.2) {$L^2(J,\Sigma_{J},\mu_J)$};
		\node[below] at (3,0) {$0$};
		\draw[->,color=red] (3,0) -- (3.7,0.5) node[above] {$(f\restr J)^{\ssbullet}$};
		
		\node[above] (left) at (-2.3,1.3) {};
		\node[above] (right) at (2.3,1.2) {};
		\draw [->] (left) to [bend left=45] (right);
		\node at (0,2) [draw=none] {$T_{(h,\phi^{-1})}$};
		\end{tikzpicture}
	}
	\caption{Signal space as a subcategory of $\mathbf{Hilb}$ or $\mathbf{Riesz}$}
	\label{fig:space_overview}
\end{figure}
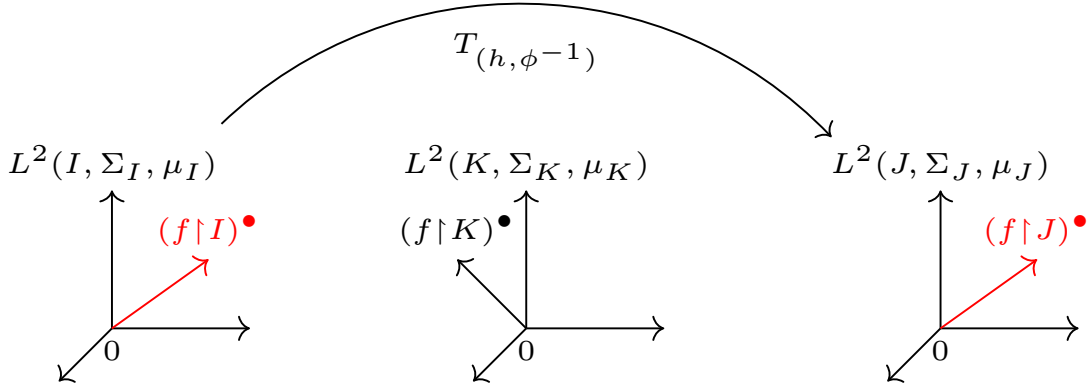

\begin{center}
	\begin{table}[ht]
		\centering
		\resizebox{\linewidth}{!}{
			\begin{tabular}{|c|c|}\hline
				{\bf Contemporary function based model} & {\bf Proposed functorial model} \\ \hline
				\makecell{Signal as entity varying in time or space\\modeled mathematically as a function} & 
				\makecell{Signal as a functor from a category (modeling\\generative structure)to a category (modeling observed waveforms)}\\ \hline
				\makecell{Simple 1D time signal as measurable\\function $f:(\Bbb R,\hat\Sigma_{\Cal B}) \to (\Bbb R,\Sigma_{\Cal B})$} & 
				\makecell{Simple 1D time signal as $F\mathbf{C}$\\ using a functor $F:\mathbf{C} \to \mathbf{Meas}^{\rightarrow}$} \\ \hline
				\makecell{Signal $f:(\Bbb R,\hat\Sigma_{\Cal B}) \to (\Bbb R,\Sigma_{\Cal B})$ viewed as a special case\\$F:\mathbf{C} \to \mathbf{Meas}^{\rightarrow}$; $\mathbf{C}$ being some discrete category.} & 
				\makecell{Signal as a family of objects $F(G_1) = f\restr I, F(G_2) =f\restr J,...$\\ along with non-trivial arrows $F(a_1),F(a_2),...$}\\ \hline
				\makecell{Space $\mathfrak{L^0}_{\Bbb R}(\mu)$ canonically isomorphic\\to $\prod_{i\in (I,J,..)}\mathfrak{L^0}_{i}(\mu_i)$ via $f\mapsto(f\restr I,f\restr J,...)$} & 
				\makecell{$f\restr J = Fa(FG_1) = f\restr I$-valued point of $f\restr J + \Delta_J$ \\$(f\restr I,f\restr J,...)$ as family of generalized elements and $\Delta$s}\\ \hline
				\makecell{Using properties of $\eae$, measure, $|f|^2$ integrable, \\field $\Bbb R$, $f^{\ssbullet} \in L^0(\Bbb R,\hat\Sigma_{\Cal B},\Cal N(\mu))$, $f^{\ssbullet} \in L^2(\Bbb R,\hat\Sigma_{\Cal B},\mu)$} & 
				\makecell{Properties of null-ideal,measure on $(I,\Sigma_I)$, $\Bbb R$ as field\\permits $(f\restr I)^{\ssbullet} \in L^0(I,\Sigma_I,\Cal N(\mu_I))$, $(f\restr I)^{\ssbullet} \in L^2(I,\Sigma_{I},\mu_I)$}\\ \hline
				\makecell{Signal space is generically fixed \\as $L^0(\Bbb R,\hat\Sigma_{\Cal B},\Cal N(\mu))$ or $L^2(\Bbb R,\hat\Sigma_{\Cal B},\mu)$.} & 
				\makecell{When arrows $(h,\phi),(h',\phi'),...$ uniquely define operators \\ the resulting subcategory in $\mathbf{Hilb}$ or $\mathbf{Riesz}$ is signal matched.} \\ \hline
				\makecell{Signal space is a fixed object $L^2(\Bbb R,\hat\Sigma_{\Cal B},\mu)$ in $\mathbf{Hilb}$\\or $\mathbf{Riesz}$ where $L^2:\mathbf{LocMeasure} \rightarrow \mathbf{Riesz}$} & 
				\makecell{When $\mathbf{C} \subseteq \mathbf{LocMeasure}$ is generative category,\\then generative signal matched space is directly modeled as $L^2(\mathbf{C})$} \\ \hline
				\makecell{When $\mathbf{C}= \mathbf{1}$, time axis thought of as single base \\corresponding to the object $(\Bbb R,\hat\Sigma_{\Cal B})$ of $\mathbf{Meas}$.} & 
				\makecell{Multiple $(I,\Sigma_{I})$,$(J,\Sigma_{J})$,... bases and change of base\\in opposite $\mathbf{Meas}^{op}$ leading to celebrated relative viewpoint.} \\ \hline
				\makecell{Classic signal $f^{\ssbullet} \in L^2(\Bbb R,\Sigma_{Leb},\mu) = {L^2}(\mu)$\\$f^{\ssbullet}\mapsto[(f\restr I)^{\ssbullet},(f\restr J)^{\ssbullet},...]$} & 
				\makecell{When $\mathbf{C} \subseteq \mathbf{LocMeasure}$, signal $f^{\ssbullet}$ = $[(f\restr I)^{\ssbullet},(f\restr J)^{\ssbullet},...]$\\$(f\restr J)^{\ssbullet}$ represented as $\Delta_J +L^2(\phi^{-1})(f\restr I)^{\ssbullet}$} \\ \hline
				\makecell{In summary no mathematical modeling of signal\\source, pure set-theoretic measure theory} & 
				\makecell{In summary mathematical modeling of source as category\\combining both category and set-theoretic measure theory}\\ \hline
			\end{tabular}
		}
		\caption{A summary of major differences between conventional and functorial signal model}
		\label{table:summary_sig_rep}
	\end{table}
\end{center}

In summary, the {\bf category-theoretic perspective} views signal $f = ... \amalg (f\restr I) \amalg (f\restr K) \amalg (f\restr J) \amalg ... $ as the coproduct object in category $\mathbf{Meas}^{\rightarrow}$ where the subobjects are related by special arrows $(h,\phi):(f\restr I) \to (f\restr J)$. The {\bf generative perspective} views objects $(f\restr I) = F(G_1)$ as having correspondence to natural generators $G_1$ such as melodies/physical objects/linguistic words related by arrow $a: G_1 \to G_2$ modeled using a functor $F:\mathbf{C} \to \mathbf{Meas}^{\rightarrow}$. Finally the {\bf Set-theoretic measure-theory perspective} utilizes field property on the underlying set of measurable space $(\Bbb R,\Sigma_{\Cal B})$, treating $(f\restr I)$ also as an element of Riesz space $\mathfrak{L}^0_I$. Moreover the classic signal representation techniques don't take into account any special mathematical modeling of signal source and utilize pure set-theoretic measure theory and functional analysis. In the proposed functorial model the mathematical modeling of source as a category facilitates combining both category-theory (especially for relative viewpoint) and set-theoretic measure theory (for computational viewpoint) simultaneously.

\subsection{Notation and Terminology}

$(F,\mathbf{C},\mathbf{D})$: Abstract graph of a functor, $\mathbf{D}^{\rightarrow}$: Arrow category of category $\mathbf{D}$.
$f\restr I$ : Restriction ($f|_{I}$) of function $f$ to domain $I$, {$\mathbf{Meas}$}  : Category of measurable spaces $(X,\Sigma_{X})$ and measurable maps (functions) $f:X\to Y$, {$\mathbf{Measure}$}  : Category of measure spaces $(X,\Sigma_{X},\mu)$ and inverse-measure-preserving morphisms $f:X\to Y$, {$\mathbf{LocMeas}$}  : Category of localizable measurable spaces $(X,\Sigma_{X},\Cal N(\mu))$ and $f^\ssbullet$ a.e. equivalence classes of non-singular measurable maps $f:X\to Y$, {$\mathbf{LocMeasure}$}  : Category of localizable measure spaces $(X,\Sigma_{X},\mu)$ and $f^\ssbullet$ a.e. equivalence classes of inverse-measure-preserving maps $f:X\to Y$, {$\mathbf{compBoolAlg}$}  : Category of Dedekind $\sigma$-complete Boolean algebras $\frak B$ and sequentially order-continuous Boolean homomorphisms $\pi:\frak B\to\frak A$, {$\mathbf{MeasureAlg}$}  : Category of measure algebras $(\frak B,\bar\mu)$ and SOC measure-preserving Boolean homomorphisms $\pi:\frak B\to\frak A$, {$\mathbf{countMeas}$}  : Subcategory of $\mathbf{Meas}$ (or $\mathbf{LocMeas}$) of measurable spaces $(X,\Cal PX)$ (or $(X,\Cal PX,\phi)$)and measurable maps (functions) $f:X\to Y$, {$\mathbf{countMeasure}$}  : Subcategory of $\mathbf{Measure}$ (or $\mathbf{LocMeasure}$) of measure spaces $(X,\Cal PX,count)$ and inverse-measure-preserving maps $f:X\to Y$, {$\mathbf{Riesz}$} : Category of Riesz spaces and Riesz homomorphisms, {$\mathbf{BanLatt}$}  : Category of Banach lattices and bounded Riesz homomorphisms, {$\mathbf{Hilb}$} : Category of Hilbert spaces and continuous(bounded) linear maps, {$\mathbf{compBoolAlg}$}  : Category of Dedekind $\sigma$-complete Boolean algebras $\frak B$ and sequentially order-continuous Boolean homomorphisms $\pi:\frak B\to\frak A$, {$\mathbf{compBoolAlg}$}  : Category of Dedekind $\sigma$-complete Boolean algebras $\frak B$ and
sequentially order-continuous Boolean homomorphisms $\pi:\frak B\to\frak A$, $\mathfrak{L}^0(X,\Sigma_{X})$ : (or $\mathfrak{L}^0= \mathfrak{L}^0(\mu)$) Space of virtually measurable real-valued functions $f$ defined on conegligible subsets of $X$, $\mathfrak{L}^0_X$ : Space of measurable real-valued functions $f$ defined on $X$, $\mathfrak{L}^p(X,\Sigma_{X})$ : (or $\mathfrak{L}^p= \mathfrak{L}^p(\mu)$, $p\in(1,\infty)$) Set of functions $f\in\mathfrak{L}^0=\mathfrak{L}^0(\mu)$ such that $|f|^p$ is integrable, ${L^0}(X,\Sigma_{X},\mu)$ : (or ${L}^0={L}^0(\mu)={L^0}(X,\Sigma_{X},\Cal N(\mu))$) Set of equivalence classes $f^{\ssbullet}$ in $\mathfrak{L}^0(\mu)$ under $\eae$, ${L^p}(X,\Sigma_{X},\mu)$ : (or ${L}^p={L}^p(\mu)$, $p\in(1,\infty)$) Set of functions $\{f^{\ssbullet}:f\in\mathfrak{L}^p\}\subseteq L^0=L^0(\mu)$ in $\mathfrak{L}^0(\mu)$ under $\eae$.
\section{On $L^0$ and $L^2$ functors}
\label{sec:funcsigspace}
In this section we study in detail all possible cases of $L^0$ and $L^2$ functors including the special cases of $l^0$, $l^2$ and their extension to partial categories. These are inevitably required in the context of proposed functorial signal representation model when the measure structure of generative base category is of interest for various reasons as motivated in Section~\ref{sec:motiv}.  

\subsection{Base or Domain categories $\mathbf{C},\mathbf{C}^{op}$}\label{sec:base}

All set-theoretic $L^p$ constructions are inherently functors referring~\cite{fremlinmt3}. We will consider two major categories $\mathbf{LocMeas}$ a category of {\bf localizable measurable spaces} refer~\cite{locale},~\cite{thesis} $\mathbf{LocMeasure}$ a category of {\bf localizable measure spaces} and their opposites. These categories are useful in context of $L^0$ and $L^2$ (or $L^p$, $1 \le p \le \infty$). Next we consider special subcategories $\mathbf{countMeas}$,$\mathbf{countMeasure}$ and partial monic derived category Par$(\mathbf{LocMeas},\mathcal{M})$, $\mathbf{PInj}$ of the earlier categories as summarized in Table~\ref{table:base_cat} to cover almost all the cases of $l^0$ and $l^2$ (or $l^p$, $1 \le p \le \infty$).

Initially we attempted generalizing the example of $l^2$ functor on the category $\mathbf{PInj}$ as studied in~\cite{Heunen2013} through the concept of restriction and inverse categories using underlying category of localizable measurable spaces with finite products. However later referring to the work~\cite{fremlinmt3} and references therein we discovered that the functoriality of function spaces is well-studied in the context of Boolean and measure algebras. This made it easier to apply directly in signal representation the perspectives on a functor developed in~\cite{salilp1} to shed light on fundamental concepts of redundancy and compression.

\begin{center}
	\begin{table}[ht]
		\centering
		\resizebox{\linewidth}{!}{
			\begin{tabular}{|c|c|}\hline
				Base category $\mathbf{C}$ & Opposite category $\mathbf{C}^{op}$ \\ \hline
				$\mathbf{LocMeas}$ & $\mathbf{LocMeas}^{op}$ \\ \hline
				$\mathbf{LocMeasure}$ & $\mathbf{LocMeasure}^{op}$ \\ \hline
				$\mathbf{countMeas}$ or $\mathbf{countLocMeas}$ & $\mathbf{countMeas}^{op}$ or $\mathbf{countLocMeas}^{op}$ \\ \hline
				$\mathbf{countMeasure}$ or $\mathbf{countLocMeasure}$ & $\mathbf{countMeasure}^{op}$ or $\mathbf{countLocMeasure}^{op}$ \\ \hline
				$\mathbf{PInj}$, Par$(\mathbf{LocMeas},\mathcal{M})$ & $\mathbf{PInj}$, Par$(\mathbf{LocMeas},\mathcal{M})$ \\ \hline
				
			\end{tabular}
		}
		\caption{Various base categories covering all cases of functors $L^0,L^2,l^0,l^2$.}
		\label{table:base_cat}
	\end{table}
\end{center}

\subsection{Objects of $\mathbf{C},\mathbf{C}^{op}$}

Theorem~\ref{meas_alg} essentially relates the objects of categories $\mathbf{LocMeas}$, $\mathbf{LocMeas}^{op}$, $\mathbf{LocMeasure}$, $\mathbf{LocMeasure}^{op}$.

\begin{definition} 
	\label{def:con_meas_alg}
	For a measure space $(X,\Sigma,\mu)$, $(\frak B,\bar\mu)$, as constructed in Theorem~\ref{meas_alg}, is called the {\bf measure algebra} of $(X,\Sigma,\mu)$.
\end{definition} 

\begin{theorem} 
	\label{meas_alg}	
	Let $(X,\Sigma_X,\mu)$ be a measure space, and $\Cal N$ be the null ideal of $\mu$. Corresponding to the measurable space $(X,\Sigma_X,\Cal N)$ we have $\frak B$ as the Boolean algebra quotient $\Sigma_X/\Sigma_X\cap\Cal N$. Then we can associate a functional $\bar\mu:\frak B\to[0,\infty]$ defined by setting 
	
	\centerline{$\bar\mu E^{\ssbullet}=\mu E$ for every $E\in\Sigma_X$,} 
	
	\noindent and $(\frak B,\bar\mu)$ is a measure algebra corresponding to $(X,\Sigma_X,\mu)$.  The canonical 
	map $E\mapsto E^{\ssbullet}:\Sigma_X\to\frak B$ is sequentially order-continuous. 
\end{theorem}

\begin{proof}
	{\bf (a) }From Proposition~\ref{quotient_ba_com} and its corollary~\ref{cor:quotient_ba_com}, it follows that $\frak B$ is a Dedekind $\sigma$-complete 
	Boolean algebra. By Corollary~\ref{cor:seq_order_con}, $E\mapsto E^{\ssbullet}$ is sequentially order-continuous, since $\Sigma_X \cap\Cal N$ is a $\sigma$-ideal of $\Sigma_X$. 
	
	\medskip 
	
	{\bf (b)} For $E$, $F\in\Sigma_X$ and $E^{\ssbullet}=F^{\ssbullet}$ in $\frak B$, we have $E\symmdiff F\in\Cal N$, hence using $\mu(E\setminus F) = \mu(F\setminus E) = 0$ (refer Corollary~\ref{cor:eq_rel_ideal}) we get {$\mu F \le\mu E+\mu(F\setminus E)=\mu E\le\mu F+\mu(E\setminus F)=\mu F$} or $\mu F=\mu E$. Consequently the defined formula does indeed define a function $\bar\mu:\frak B\to[0,\infty]$. 
	
	\medskip 
	
	{\bf (c)} We verify the axioms of Definition~\ref{def:meas_alg}. First for the equivalence class of empty set we have {$\bar\mu 0=\bar\mu\emptyset^{\ssbullet}=\mu\emptyset=0$.} Second let $\sequencen{a_n}$ be a disjoint sequence in $\frak B$. Now choose for each $n\in\Bbb N$ an $E_n\in\Sigma_X$ such that 
	$E_n^{\ssbullet}=a_n$. Note that $\sequencen{E_n}$ need not be disjoint although $E_n^{\ssbullet}$ are equivalence classes and therefore partitions of $\Sigma_X$. In order to evaluate $\sum_{n=0}^{\infty}\bar\mu a_n$ we form a disjoint sequence $\sequencen{F_n}$ by setting $F_n=E_n\setminus\bigcup_{i<n}E_i$;  then {$F_n^{\ssbullet} = E_n^{\ssbullet}\setminus\bigcup_{i<n}E_i^{\ssbullet}= E_n^{\ssbullet}\setminus\sup_{i<n}E_i^{\ssbullet}$ using fact that $\frak B$ is a Dedekind $\sigma$-complete Boolean algebra. Thus $F_n^{\ssbullet} =a_n\setminus\sup_{i<n}a_i=a_n$} since $\sequencen{a_n}$ is disjoint. Then $\bar\mu a_n=\bar\mu F_n^{\ssbullet}=\mu F_n$ for each $n$. Now let $E=\bigcup_{n\in\Bbb N}F_n=\bigcup_{n\in\Bbb N}E_n$;  then $E^{\ssbullet}=\bigcup_{n\in\Bbb N}F_n^{\ssbullet}=\sup_{n\in\Bbb N}F_n^{\ssbullet}=\sup_{n\in\Bbb N}a_n$. 
	Hence, {$\bar\mu E^{\ssbullet} = \bar\mu(\sup_{n\in\Bbb N}a_n) =\mu E =\sum_{n=0}^{\infty}\mu F_n =\sum_{n=0}^{\infty}\bar\mu a_n$.} 
	For the last axiom, if $a\ne 0$, choose an $E\in\Sigma_X$ such that $E^{\ssbullet}=a$. But $E\notin\Cal N$ as it belongs to different equivalence class than that of empty set, so $\bar\mu a=\mu E>0$. Indeed we have verified that $(\frak B,\bar\mu)$ is a measure algebra. 
\end{proof}

\begin{center}
	\begin{table}[ht]
		\centering
		\resizebox{\linewidth}{!}{
			\begin{tabular}{|c|c|}\hline
				Objects of category $\mathbf{C}$ & Objects of category $\mathbf{C}^{op}$ \\ \hline
				$(X,\Sigma_{X},\Cal N)$ (localizable measurable space) & $\frak B$ (or complete Boolean algebra $\Sigma_X/\Sigma_X\cap\Cal N$) \\ \hline
				$(X,\Sigma_{X},\mu)$ (localizable measure space) & $(\frak A,\bar\mu)$ (or measure 
				algebra of $(X,\Sigma_X,\mu)$) \\ \hline
				$(X,\Cal PX,\Phi)$ ($\Phi$ null-ideal of counting measure) & $\frak B$ (or atomic complete Boolean algebra $\Sigma_X$) \\ \hline
				$(X,\Cal PX,counting)$ (localizable counting measure space) & $(\frak B,\bar{counting})$ (or measure 
				algebra of counting measure on $X$) \\ \hline
				Objects of $\mathbf{PInj}$ or Par$(\mathbf{LocMeas},\mathcal{M})$ & $\mathbf{C}^{op}$ via inverse map on arrows if it exists. \\ \hline
				
			\end{tabular}
		}
		\caption{The Objects of various base categories from which appropriate functors $L^0,L^2,l^0,l^2$ are defined.}
		\label{table:base_obj}
	\end{table}
\end{center}

Table~\ref{table:base_obj} summarizes all the objects of base categories. The last case of Par$(\mathbf{LocMeas},\mathcal{M})$ is essential for the extension of $L^0$ and Par$(\mathbf{B},\mathcal{M})$ where $\mathbf{B} \subseteq \mathbf{LocMeasure}$ consists of objects with $\sigma$-finite measures caters to the extension of $L^p$. Then category $\mathbf{PInj}$ is same as Par$(\mathbf{countMeas},\mathcal{M})$ for the extended case of $l^0$ but reduces to special case of Par$(\mathbf{B},\mathcal{M})$ where $\mathbf{B}$ consists of objects $(X,\Cal PX,counting)$ with countable $X$ (since counting measure on a set $X$ is $\sigma$-finite iff $X$ is countable) for the extended case of $l^2$. These extensions were motivated from~\cite{Heunen2013} and makes use of the fact that these categories are inverse restriction categories (or roughly speaking partial groupoids) and therefore the arrows are uniquely reversed to form an opposite category from the objects and arrows of original category. This is possible on account of good categorical properties of $\mathbf{LocMeas}$ especially the finite products which are required for deriving the partial categories.

\subsection{Arrows of $\mathbf{C},\mathbf{C}^{op}$}

Now we discuss the appropriate structure preserving functions or homomorphisms of the base categories. Recall that if $X$ and $Y$ are any sets and $\Sigma_Y$ a $\sigma$-algebra of subsets of $Y$; then for an arbitrary function $f:X\to Y$ between sets; $\{f^{-1}[F]:F\in\Sigma_Y\}$ is always a $\sigma$-algebra of subsets of $X$. Now only those functions $f$ are called measurable when $\{f^{-1}[F]:F\in\Sigma_Y\} \subseteq \Sigma_X$. It is well known that $\mathbf{Meas}$ is a category with measurable spaces such as $(X,\Sigma_{X}),(Y,\Sigma_{Y})$ as its objects and these measurable functions $f$ as its morphisms. Now the objects of $\mathbf{LocMeas}$ and $\mathbf{LocMeasure}$ could be viewed as measurable spaces with some added structure such as null ideal or actual measure or dually as appropriate Boolean algebras. Thus the proper structure preserving maps on these objects are measurable functions which preserve this additional structure such as measurable non-singular (measure zero-reflecting) morphisms and Inverse-measure-preserving maps and dually the corresponding Boolean homomorphisms which are reviewed here.

The propositions~\ref{arrows_measurable},~\ref{arrows_measure} essentially relates the arrows of categories $\mathbf{LocMeas}$, $\mathbf{LocMeas}^{op}$, $\mathbf{LocMeasure}$, $\mathbf{LocMeasure}^{op}$.

\begin{proposition} 
	\label{arrows_measurable}
	Let $(X,\Sigma_X,\mu)$ and $(Y,\Sigma_Y,\nu)$
	be measure spaces, and correspondingly $(\frak B,\bar\mu)$, $(\frak A,\bar\nu)$ their
	measure algebras. If $f:X\to Y$ is a non-singular measurable function such that
	$f^{-1}[F]\in\Sigma_X$ for every $F\in\Sigma_Y$ and $\mu f^{-1}[F]=0$
	whenever $\nu F=0$. Then there exists a sequentially order-continuous
	Boolean homomorphism $\phi_f:\frak A\to\frak B$ defined by the setting
	
	\centerline{$\phi_f F^{\ssbullet}=f^{-1}[F]^{\ssbullet}$ for every
		$F\in\Sigma_Y$.}
	
\end{proposition}

\begin{proof}
	Let $\Cal N(\mu)$ and  $\Cal N(\nu)$ be null ideals of the measures $\mu$ and $\nu$ respectively. Let
	$\Cal I=\Sigma_X\cap\Cal N(\mu)$ and $\Cal J=\Sigma_Y\cap\Cal N(\nu)$, then from Theorem~\ref{meas_alg}
	we have $\frak B$ and $\frak A$ as the Boolean algebra quotients $\Sigma_X/\Cal I$ and $\Sigma_Y/\Cal J$
	respectively. By definition of non-singular measurable function, $\{f^{-1}[N]:N\in\Cal N(\nu)\} \subseteq \Cal N(\mu)$ hence  
	$f^{-1}[F]\in\Cal I$ for every $F\in\Cal J$.
	Now consider $F_1$, $F_2\in\Sigma_Y$ such that $F_1^{\ssbullet}=F_2^{\ssbullet}$, then $F_1\symmdiff F_2\in\Cal J$ consequently
	$f^{-1}[F_1]\symmdiff f^{-1}[F_2]=f^{-1}[F_1\symmdiff F_2] \in \Cal I$ and $f^{-1}[F_1]^{\ssbullet}=f^{-1}[F_2]^{\ssbullet}$.
	So the setting indeed defines a map $\phi_f:\frak A\to\frak B$. 
	
	Next we prove it is a Boolean homomorphism. For $F_1$, $F_2\in\Sigma_Y$, we have 
	$\phi_f F_1^{\ssbullet}\symmdiff \phi_f F_2^{\ssbullet}
	=f^{-1}[F_1]^{\ssbullet}\symmdiff f^{-1}[F_2]^{\ssbullet}
	=(f^{-1}[F_1]\symmdiff f^{-1}[F_2])^{\ssbullet}
	=f^{-1}[F_1\symmdiff F_2]^{\ssbullet}
	=\phi_f(F_1\symmdiff F_2)^{\ssbullet}
	=\phi_f(F_1^{\ssbullet}\symmdiff F_2^{\ssbullet})$.
	
	Hence $\phi_f(a_1\symmdiff a_2)=\phi_f a_1\symmdiff \phi_f a_2$ for all $a_1$,
	$a_2\in\frak A$. Similarly it is easily verified that, $\phi_f(a_1\cap a_2)=\phi_f a_1\cap a_2$ for all $a_1$,
	$a_2\in\frak A$, and finally $\phi_f 1_{\frak A}=\phi_f Y^{\ssbullet}
	=f^{-1}[Y]^{\ssbullet}=X^{\ssbullet}=1_{\frak B}$.
	
	Finally it remains to prove that $\phi_f$ is sequentially order-continuous.
	To prove this, let $\sequencen{a_n}$ be a sequence in $\frak A$.  For each $n$
	choose an $F_n\in\Sigma_Y$ such that $F_n^{\ssbullet}=a_n$, and let
	$F=\bigcup_{n\in\Bbb N}F_n$.   As the map
	$E\mapsto E^{\ssbullet}:\Sigma_Y\to\frak A$ is sequentially order-continuous
	either from Theorem~\ref{def:meas_alg} or from Corollary~\ref{cor:seq_order_con}, 
	$F^{\ssbullet}=\sup_{n\in\Bbb N}a_n$ in $\frak A$. Hence
	
	$\phi_f(\sup_{n\in\Bbb N}a_n) =\phi_f F^{\ssbullet}
	=f^{-1}[F]^{\ssbullet}
	=(\bigcup_{n\in\Bbb N}f^{-1}[F_n])^{\ssbullet}
	=\sup_{n\in\Bbb N}f^{-1}[F_n]^{\ssbullet}
	=\sup_{n\in\Bbb N}\phi_f F_n^{\ssbullet}
	=\sup_{n\in\Bbb N}\phi_f a_n$, proving $\phi_f$ is sequentially order-continuous.
\end{proof}

\begin{proposition}
	\label{arrows_measure}
	Let $(X,\Sigma_X,\mu)$ and $(Y,\Sigma_Y,\nu)$ be
	measure spaces, with measure algebras $(\frak B,\bar\mu)$ and
	$(\frak A,\bar\nu)$. Let $f:X\to Y$ be inverse-measure preserving.   
	Then we have a sequentially order-continuous measure-preserving Boolean
	homomorphism $\phi_f:\frak A\to\frak B$ defined by setting
	$\phi_f F^{\ssbullet}=f^{-1}[F]^{\ssbullet}$ for every $F\in\Sigma_Y$.
\end{proposition}

\begin{proof}
	Since $f:X\to Y$ is inverse-measure-preserving, $f^{-1}[F]\in\Sigma_X$ and $\mu(f^{-1}[F])=\nu F$ for every
	$F\in\Sigma_Y$. Following Theorem~\ref{arrows_measurable} it is immediate that $\phi_f$ sequentially order-continuous Boolean homomorphism. 
	Using $\mu(f^{-1}[F])=\nu F$ for every $F\in\Sigma_Y$, we get $\bar\mu(f^{-1}[F]^{\ssbullet})=\bar\mu(\phi_f F^{\ssbullet})=\bar\nu F^{\ssbullet}$ from Definition~\ref{def:con_meas_alg}. It is proved that $\phi_f:\frak A\to\frak B$ is sequentially order-continuous measure-preserving Boolean homomorphism.
\end{proof}

Table~\ref{table_base_mor} summarizes all the arrows of base categories.

\begin{center}
	\begin{table}[ht]
		\centering
		\resizebox{\linewidth}{!}{
			\begin{tabular}{|c|c|}\hline
				Arrows of category $\mathbf{C}$ & Arrows of category $\mathbf{C}^{op}$ \\ \hline
				$f^\ssbullet$ a.e. class of Non-singular(M0R) measurable $f:X\to Y$ & Sequentially order-continuous
				Boolean homomorphism $\pi_f:\frak B\to\frak A$ \\ \hline
				$f^\ssbullet$ a.e. class of Inverse-measure-preserving $f:X\to Y$& SOC 
				measure preserving Boolean homomorphism $\pi_f:\frak B\to\frak A$ \\ \hline
				Non-singular(M0R) measurable $f:X\to Y$ & Sequentially order-continuous
				Boolean homomorphism $\pi_f:\frak B\to\frak A$ \\ \hline
				Inverse-measure-preserving $f:X\to Y$ & SOC 
				measure preserving Boolean homomorphism $\pi_f:\frak B\to\frak A$ \\ \hline
				Partial Injections or Partial measurable maps & Same as arrows of Inverse Category $\mathbf{C}$ \\ \hline
				
			\end{tabular}
		}
		\caption{The arrows of various base categories from which appropriate functors $L^0,L^2,l^0,l^2$ are defined.}
		\label{table_base_mor}
	\end{table}
\end{center} 

\subsection{Codomain Categories}\label{sec:codom}
All $L^p$ function spaces are Riesz spaces. This suggests that the generic codomain category for $L^0$ and $L^2$ (or $L^p$, $1 \le p \le \infty$) functors is the category $\mathbf{Riesz}$ with objects as Riesz spaces and arrows as Riesz homomorphisms. However 
depending upon the kind of structure one is interested on signal space, we briefly tabulate different codomain categories with objects carrying those structures and corresponding structure-preserving morphisms as summarized in Table~\ref{table:codomain_structures}.
\begin{center}
	\begin{table}[ht]
		\centering
		\resizebox{\linewidth}{!}{
			\begin{tabular}{|c|c|}\hline
				Structure & Structure preserving morphism \\ \hline
				Linear space & Linear map \\ \hline
				Partial order & Monotone (order-preserving) map \\ \hline
				Lattice & Lattice (order-continuous) homomorphism \\ \hline
				Partial ordered Linear Space & Positive (order-preserving) Linear map \\ \hline
				Riesz Space & Riesz homomorphism  \\ \hline
				Topological Linear Space & Continuous linear map  \\ \hline
				Banach Lattice & Bounded Riesz homomorphism  \\ \hline
			\end{tabular}
		}
		\caption{The combination of structures in codomain categories for functors $L^0,L^2,l^0,l^2$.}
		\label{table:codomain_structures}
	\end{table}
\end{center}

The standard categories of interest for $L^0$ is the category of Riesz spaces and Riesz homorphisms (multiplicative sequentially order-continuous Riesz homomorphism); while for $L^p$ is the category of normed Riesz spaces and norm-presereving Riesz homorphisms (multiplicative sequentially order-continuous Riesz homomorphism). For $L^2$ the category of Hilbert spaces with continuous linear maps is often most important from application perspective.

\subsection{Functors $L^0$,$L^2$}
Having studied the domain and codomain categories, in this section we formulate various possible major cases of $L^0$ and $L^2$ functors. The main reference for these cases are volumes~\cite{fremlinmt2},~\cite{fremlinmt3} of the treatise on unified measure theory.

\subsection{Results on Functors $L^0$}
In this section, we specifically discuss all results pertaining to the usual $L^0$ construction of measure theory which can be viewed appropriately as a functor. The discussion and proof for next result is adapted from Exercise $241X(g)$ of~\cite{fremlinmt2}.

\begin{proposition}~\cite{fremlinmt2} 
	\label{prop:L0_func_space}	
	Let $(X,\Sigma_{X},\Cal N(\mu))$ and $(Y,\Sigma_{Y},\Cal N(\nu))$ be localizable measurable spaces, and $\phi:X\to Y$ a non-singular measurable function. Then (a)For every real valued 
	$\Sigma_Y$-measurable function $g$ defined on $Y$,
	$g\phi$ is $\Sigma_X$-measurable. (b) The map
	$g\mapsto g\phi:\mathfrak{L}^0(\nu)\to\mathfrak{L}^0(\mu)$ induces 
	$T:L^0(\nu)\to L^0(\mu)$ defined by setting
	$Tg^{\ssbullet}=(g\phi)^{\ssbullet}$ for every $g \in \mathfrak{L}^0$.
	(c) $T$ is linear operator, that preserves multiplicative structure $T(v\times w)=Tv\times Tw$ for all $v$,
	$w\in L^0(\nu)$, and also lattice structure, that is $T(\sup_{n\in\Bbb N}v_n)=\sup_{n\in\Bbb N}Tv_n$
	given $\sequencen{v_n}$ is a sequence in $L^0(\nu)$ with a supremum in $L^0(\nu)$.
\end{proposition} 
\begin{proof}
	\quad(a) Noting that every inverse-measure-preserving function by its definition is always measurable and
	using Lemma~\ref{thm:composite_meas_func}, it is immediate that $g\phi$ is $\Sigma_X$-measurable. Briefly, Let $a\in\Bbb R$. There exists an $F\in\Sigma_Y$ such that $\{y:g(y)\le a\}=F$. Again there exists an $E\in\Sigma_X$ such that $\phi^{-1}[F]=E$. Hence $\{x:g\phi(x)\le a\}= E\in\Sigma_X$. Since $a$ is arbitrary, $g\phi$ is $\Sigma_X$-measurable.
	
	\quad(b) By definition, $\mathfrak{L}^0(\nu)$ is the set of real-valued functions $f$
	defined on conegligible subsets of $Y$ which are virtually
	measurable, that is, such that $f\restr C$ is measurable for
	some conegligible set $C\subseteq Y$. In other words, $\mathfrak{L}^0(\nu)$ is 
	just the set of real-valued functions $f$, defined on subsets of $Y$, which are equal almost everywhere to some $\Sigma_Y$-measurable function $g$ from $Y$ to $\Bbb R$. More precisely, if $g:Y\to\Bbb R$ is $\Sigma_Y$-measurable and
	$f\eae g$, then $C=\{y:y\in\dom f,\,f(y)=g(y)\}$ is conegligible and
	$f\restr C=g\restr C$ is measurable, so $f\in\mathfrak{L}^0(\nu)$ by part(h) of~\ref{thm:prop_meas_func}. Now by Proposition~\ref{prop:equivalence_coneg_func} $\eae$ is an equivalence relation on $\mathfrak{L}^0(\nu)$ and the space $L^0(\nu)$ is the set of equivalence classes in $\mathfrak{L}^0(\nu)$ under
	$\eae$. For $f\in\mathfrak{L}^0(\nu)$, we denote $f^{\ssbullet}$ for its equivalence class in $L^0(\nu)$. Consider $f\in\mathfrak{L}^0(\nu)$ and $g\in\mathfrak{L}^0(\nu)$ such that $f^{\ssbullet} = g^{\ssbullet}$ in $L^0(\nu)$. Then $C=\{y:y\in\dom f,\,f(y)=g(y)\}$ is conegligible and $f\restr C=g\restr C$ is measurable. Again repeating arguments of Lemma~\ref{thm:composite_meas_func} we may conclude $(f\phi)^{\ssbullet} = (g\phi)^{\ssbullet}$. Explicitly let $A=\phi^{-1}[C]=\dom (g\restr C)\phi$, which is surely conegligible (since $\phi$ is non-singular or zero-reflecting) and $a\in\Bbb R$. Since $f\restr C=g\restr C$ is $\Sigma_Y$-measurable, by Definition~\ref{def:subspace_sigma_meas_func} there exists an $F\in\Sigma_Y$ such that $\{y:(g\restr C)(y)\le a\}=\{y:(f\restr C)(y)\le a\}= F\cap C$. On the other hand, there exists an $E\in\Sigma_X$ such that $\phi^{-1}[F]=E$. Hence $\{x:(g\restr C)\phi(x)\le a\}=\{x:(f\restr C)\phi(x)\le a\}=A\cap E\in\Sigma_A$. Since $a$ is arbitrary, $(g\restr C)\phi = (f\restr C)\phi$ is $\Sigma_X$-measurable. But $g\phi : X \rightarrow \Bbb R$ is $\Sigma_X$-measurable and extends $(g\restr C)\phi = g\phi \restr A$ on a conegligible $A \subseteq X$. Consequently $g\phi \eae (g\restr C)\phi \eae (f\restr C)\phi$ and $(f\phi)^{\ssbullet} = (g\phi)^{\ssbullet}$ in $L^0(\mu)$. Indeed we have well-defined $T:L^0(\nu)\to L^0(\mu)$ where $Tg^{\ssbullet}=(g\phi)^{\ssbullet}$ for every $g \in \mathfrak{L}^0(\nu)$.    
	
	\quad(c)(i) By Theorem~\ref{thm:prop_meas_func} part (b), $g + g'$ is $\Sigma_Y$-measurable,  where $(g + g')(y) = g(y) + g'(y) $ for $y \in Y$.
	Therefore there exists an $F\in\Sigma_Y$ such that $\{y:(g+g')(y)\le a\}=\{y:g(y) + g'(y)\le a\}=F$. But there exists an $E\in\Sigma_X$ such that $\phi^{-1}[F]=E$. Hence $\{x:(g+g')\phi(x)\le a\}=\{x:g\phi(x)+g'\phi(x)\le a\}= E\in\Sigma_X$. Since $a$ is arbitrary, $g\phi + g'\phi = (g + g')\phi$ is $\Sigma_X$-measurable. Now from Proposition~\ref{prop:properties_coneg_func} part (b) $(g + g')^{\ssbullet} = g^{\ssbullet} + g'^{\ssbullet}$ in $L^0(\nu)$ and $(g\phi)^{\ssbullet} + (g'\phi)^{\ssbullet} = ((g + g')\phi)^{\ssbullet}$ in $L^0(\mu)$. Hence $T(g + g')^{\ssbullet}=((g + g')\phi)^{\ssbullet}=(g\phi)^{\ssbullet} + (g'\phi)^{\ssbullet} = Tg^{\ssbullet} + Tg'^{\ssbullet} $. Similarly using Theorem~\ref{thm:prop_meas_func} and Proposition~\ref{prop:properties_coneg_func} part (c) it follows in same manner for any scalar $c \in \Bbb R$, that $T(c.g)^{\ssbullet}=((c.g)\phi)^{\ssbullet}= c.(g\phi)^{\ssbullet} = c. Tg^{\ssbullet} $. Hence $T$ is a linear operator.
	
	\quad(c)(ii) Using Theorem~\ref{thm:prop_meas_func} and Proposition~\ref{prop:properties_coneg_func} part (d) it follows , that $T(g \times g')^{\ssbullet}=((g \times g')\phi)^{\ssbullet}= (g\phi)^{\ssbullet} \times (g'\phi)^{\ssbullet} = Tg^{\ssbullet} \times Tg'^{\ssbullet} $. Hence $T$ preserves multiplication.
	
	\quad(c)(iii) Using Theorem~\ref{thm:prop_meas_func} and Proposition~\ref{prop:properties_coneg_func} part (e) it follows , that $T(\sup_{n\in\Bbb N}g_n)^{\ssbullet}=((\sup_{n\in\Bbb N}g_n)\phi)^{\ssbullet}= \sup_{n\in\Bbb N}(g\phi)^{\ssbullet} = \sup_{n\in\Bbb N}Tg^{\ssbullet} $. Hence $T$ preserves lattice structure.	
	
\end{proof} 

Proposition~\ref{prop:L0_func_space} proves that $L^0$ construction is truly functorial; it explicitly sets up a recipe for defining $L^0$ functor. Also note that we are not really making any use of second clause in the definition of inverse-measure preserving function,
that is $\mu(\phi^{-1}[F])=\nu F$ for every $F\in\Sigma_Y$. This is intimately related to the fact that $L^0$ theory
truly involves only Dedekind $\sigma$-complete Boolean algebras rather than measure algebras as Theorem~\ref{thm:dual_L0_func_space} makes this precise. In other words, the measures enter theory only via their ideals of negligible sets. This directly influences the definition of $L^0$ functor and stands as motivation behind the category of localizable measurable spaces and zero-reflecting maps apart from its good properties and generalization to partial monic derivatives.

\begin{definition}
	\label{def:L0_func_contra} 
	Let $\mathbf{LocMeas}$, $\mathbf{Riesz}$ be the categories as defined earlier. The (contravariant) functor $L^0:$ $\mathbf{LocMeas}$ $\rightarrow \mathbf{Riesz}$ is defined as a mapping that associates to each measurable space $(X,\Sigma_{X},\Cal N(\mu))$ an object $L^0(X,\Sigma_{X},\Cal N(\mu)):=\{f^{\ssbullet} \ | f \in \mathfrak{L}^0(\mu)\}$. The assignment of morphism $\phi^{\ssbullet}:(X,\Sigma_{X},\Cal N(\mu)) \rightarrow (Y,\Sigma_{Y},\Cal N(\nu)) $ is given by
	\begin{equation}
	L^0(\phi^{\ssbullet})(g^{\ssbullet}) = (g\phi)^{\ssbullet}  
	\end{equation}
	where $L^0(\phi^{\ssbullet}): L^0(Y,\Sigma_{Y},\Cal N(\nu)) \rightarrow L^0(X,\Sigma_{X},\Cal N(\mu))$ is the morphism in $\mathbf{Riesz}$ between Riesz spaces with multiplication, which is also multiplicative sequentially order-continuous.
\end{definition}

We verify the functoriality of $L^0$ now.

The object $L^0(X,\Sigma_{X},\Cal N(\mu)) = L^0(\mu)$ is a well-defined Archimedean and Dedekind $\sigma$-complete Riesz space; refer Sections~\ref{sec:linear},~\ref{sec:partialorder}. Thus it is a well-defined object of category $\mathbf{Riesz}$. Next we check if $L^0(\phi^{\ssbullet})$ is a well-defined morphism of $\mathbf{Riesz}$. From Proposition~\ref{prop:L0_func_space}, every $g^{\ssbullet}$ is sent to $L^0(\phi^{\ssbullet})(g^{\ssbullet}) = (g\phi)^{\ssbullet}$ which is a proper element of $L^0(Y,\Sigma_{Y},\Cal N(\nu)) = L^0(\nu)$ proving that $L^0(\phi^{\ssbullet})$ is a well-defined morphism in $\mathbf{Riesz}$. Now we check if composition is preserved with the help of commutative composition diagram,

\begin{equation}
\xymatrix{
	(\Bbb R,\Sigma_{\Cal B}) \ar[r]^{id} & (\Bbb R,\Sigma_{\Cal B}) \ar[r]^{id} & (\Bbb R,\Sigma_{\Cal B})  \\
	(X,\Sigma_{X}) \ar[r]_{\phi} \ar@{.>}[ru]^{g\phi} \ar@{->}[u]^{g\phi} & (Y,\Sigma_{Y}) \ar[r]_{\psi} \ar@{->}[u]^{f\psi = g} \ar@{.>}[ru]^{f\psi} & (Z,\Sigma_{Z}) \ar@{->}[u]_{f} 
}
\end{equation}

\begin{equation}
(L^0(\phi^{\ssbullet}) \circ L^0(\psi^{\ssbullet}))(f^{\ssbullet}) = L^0(\phi^{\ssbullet}) [L^0(\psi^{\ssbullet})(f^{\ssbullet})] = L^0(\phi^{\ssbullet}) (f\psi)^{\ssbullet} = (f\psi\phi)^{\ssbullet} 
\end{equation}

\begin{equation}
= L^0(\psi\phi)^{\ssbullet}(f^{\ssbullet}) =(L^0(\psi^{\ssbullet} \cdot \phi^{\ssbullet}))(f^{\ssbullet})
\end{equation}
Finally for the identity, $L^0(id^{\ssbullet}_{(Y,\Sigma_{Y},\Cal N(\nu))})(g^{\ssbullet})=(g \cdot id)^{\ssbullet} = g^{\ssbullet}$ implying $L^0(id^{\ssbullet}_{(Y,\Sigma_{Y},\Cal N(\nu))})=id_{L^0(Y,\Sigma_{Y},\Cal N(\nu))}$. Indeed functor $L^0$ is a well-defined contravariant functor.

Ofcourse depending upon the particular structures or their combinations such as vector, partial order, lattice and multiplication one can easily form variations of Definition~\ref{def:L0_func_contra} which is a sort of prototype, using different codomain categories as shown in Table~\ref{table:codomain_structures}. One could also generalize to complex $L^0$ using spaces based on complex-valued functions instead of real-valued functions bearing in mind that the usual order structure does not hold. Here one can make use of $\mathfrak{L^0}_{\Bbb C} =\mathfrak{L^0}_{\Bbb C}(\nu)$ for the space of complex-valued functions $g$ such that $\dom g$ is a conegligible subset of $Y$ and there is a conegligible subset $C\subseteq Y$ such that $g\restr C$ is measurable; which means the real and imaginary parts of $g$ both belong to $\mathfrak{L^0}(\nu)$.

In the special case of $\mu$ and $\nu$ being counting measures (and therefore point-supported by Definition~\ref{def:pointsupp}) on sets $X$ and $Y$ the Riesz spaces $L^0(\mu)$ becomes $\mathfrak{L^0}(\mu)=\Bbb R^X$ itself, since $\Sigma_{X} = \Cal PX$ and every set except the empty set has a non-zero measure or in other words the null ideal is trivial. Hence we restrict to a subcategory $\mathbf{countLocMeas}$ of $\mathbf{LocMeas}$ with objects of type $(X,\Cal PX,\Phi)$ where $\Phi$ is trivial null ideal with single element the usual empty set. Also note that here equivalence class $\phi^{\ssbullet}$ coincides with the measurable map $\phi$ itself, since we quotient by trivial null ideal.

\begin{definition}
	\label{def:l0_func_contra} 
	Let $\mathbf{countLocMeas}$, $\mathbf{Riesz}$ be the categories as defined earlier. The (contravariant) functor $l^0:$ $\mathbf{countLocMeas}$ $\rightarrow \mathbf{Riesz}$ is defined as a mapping that associates to each measurable space $(X,\Cal PX,\Phi)$ an object $l^0(X,\Cal PX,\Phi):=\{f \in \mathfrak{L}^0(\mu)\}$. The assignment of morphism $\phi:(X,\Cal PX,\Phi) \rightarrow (Y,\Cal PY,\Phi) $ is given by
	\begin{equation}
	l^0(\phi)(g) = g\phi 
	\end{equation}
	where $l^0(\phi): l^0(Y,\Cal PY,\Phi) \rightarrow l^0(X,\Cal PX,\Phi)$ is the morphism in $\mathbf{Riesz}$ between Riesz spaces with multiplication, which is also multiplicative sequentially order-continuous.
\end{definition}

We now state the dual of Proposition~\ref{prop:L0_func_space} which leads to a covariant form of $L^0$ from opposite of $\mathbf{LocMeas}$ which we denote as $\mathbf{compBoolAlg}$. Note that in general these categories are not equivalent.  
\begin{theorem}
	\label{thm:dual_L0_func_space}
	If $\frak B$ and $\frak A$ are Dedekind $\sigma$-complete Boolean algebras,
	and $\pi:\frak B\to\frak A$ is a sequentially order-continuous Boolean
	homomorphism, then
	
	(a) We have a multiplicative sequentially order-continuous Riesz 
	homomorphism $T_{\pi}:L^0(\frak B)\to L^0(\frak A)$ defined by the
	formula
	
	\centerline{$\Bvalue{T_{\pi}u>a}=\pi\Bvalue{u>a}$}
	
	\noindent whenever $a\in\Bbb R$ and $u\in L^0(\frak B)$.
	
	(b) $T_{\pi}$ is order-continuous iff $\pi$ is order-continuous,
	injective iff $\pi$ is injective, surjective iff $\pi$ is surjective.
	
	(c) If $\frak C$ is another Dedekind $\sigma$-complete Boolean algebra
	and $\theta:\frak A\to\frak C$ another sequentially order-continuous
	Boolean homomorphism then $T_{\theta\pi}=T_{\theta}T_{\pi}:L^0(\frak
	B)\to L^0(\frak C)$.
\end{theorem}

The proof of this theorem appears in Fremlin Vol 3. If $\frak B$ is the 
Boolean algebra quotient $\Sigma_X/\Sigma_X\cap\Cal N(\mu)$ then
$L^0(\mu)$ gets identified with $L^0(\frak B)$. Correspondingly each $g^{\ssbullet} \in {L}^0(\mu)$
is identified with $u \in {L}^0(\frak B)$ which is
$a\mapsto\Bvalue{g^{\ssbullet}>a}:\Bbb R\to\frak B$, for every $a\in\Bbb R$.
Now $\Bvalue{g^{\ssbullet}\in E}$ means the region where $g^{\ssbullet}$ takes values in $E$ or
$\Bvalue{g^{\ssbullet}\in E}=g^{-1}[E]^{\ssbullet}$.
Hence $\Bvalue{g^{\ssbullet}>a}=\Bvalue{g^{\ssbullet}\in (a,\infty)\,}$ or 
$\xi^\ssbullet : E^{\ssbullet} \mapsto g^{-1} [E]^{\ssbullet}$ for any Borel set $E \subseteq \Bbb R$.
In conclusion $u$ or $g^\ssbullet$ corresponds to Boolean homomorphism $\xi: (\Sigma_{\Cal B}/(\Sigma_{B} \cap \Cal N)) \rightarrow \frak B$. This theorem therefore gives rise to an explicit Definition~\ref{def:L0_func_cov}.

\begin{definition}
	\label{def:L0_func_cov} 
	Let $\mathbf{compBoolAlg}$, $\mathbf{Riesz}$ be the categories as defined earlier. The (covariant) functor $L^0:$ $\mathbf{compBoolAlg}$ $\rightarrow \mathbf{Riesz}$ is defined as a mapping that associates to each Dedekind $\sigma$-complete Boolean algebra $\frak B$ an object $L^0(\frak B):=\{u \ | a\mapsto\Bvalue{u>a} \text{for every}\ a\in\Bbb R \}$. The assignment of morphism $\pi:\frak B \rightarrow \frak A $ is given by $L^0(\pi)(u)$ defined by
	\begin{equation}
	\Bvalue{L^0(\pi)u>a}=\pi\Bvalue{u>a}  
	\end{equation}
	where $L^0(\pi): L^0(\frak B) \rightarrow L^0(\frak A)$ is the morphism in $\mathbf{Riesz}$ between Riesz spaces with multiplication, which is also multiplicative sequentially order-continuous.
\end{definition}

\subsection{Results on Functors $L^2$}
\label{sec:L2_func_space}

Analogous to Proposition~\ref{prop:L0_func_space}, Proposition~\ref{prop:L2_func_space} proves that $L^2$ construction is really functorial; it explicitly sets up a recipe for defining $L^2$ functor. The primary difference is that we now make use of second clause in the definition of inverse-measure preserving function, that is $\mu(\phi^{-1}[F])=\nu F$ for every $F\in\Sigma_Y$. This is again intimately related to the fact that $L^2$ theory truly involves measure spaces or dually measure algebras and normed Riesz spaces as Proposition~\ref{prop:L2_func_space} makes this precise.

\begin{proposition} 
	\label{prop:L2_func_space}	
	Let $(X,\Sigma,\mu)$ and $(Y,\Tau,\nu)$ be measure
	spaces, and $\phi:X\to Y$ an inverse-measure-preserving function. Then,   
	(a)For every $[-\infty,\infty]$-real valued 
	$\Sigma_Y$-measurable function $g$ defined on $Y$,
	$g\phi$ is $\Sigma_X$-measurable. (b) The map
	$g\mapsto g\phi:\mathfrak{L}^2(\nu)\to\mathfrak{L}^2(\mu)$ induces 
	$T:L^2(\nu)\to L^2(\mu)$ defined by setting
	$Tg^{\ssbullet}=(g\phi)^{\ssbullet}$ for every $g \in \mathfrak{L}^2$ such that $|g|^2$ is integrable (often termed square integrable functions) (c) $T$ is norm-preserving linear operator that is, $\|Tg^{\ssbullet}\|_2=\|g^{\ssbullet}\|_2$ for every $g^{\ssbullet}\in L^2(\nu)$, which preserves multiplicative structure $T(v\times w)=Tv\times Tw$ for all $v$,
	$w\in L^2(\nu)$, and also lattice structure, that is $T(\sup_{n\in\Bbb N}v_n)=\sup_{n\in\Bbb N}Tv_n$
	given $\sequencen{v_n}$ is a sequence in $L^2(\nu)$ with a supremum in $L^2(\nu)$.
\end{proposition} 
\begin{proof}
	By definition, $\mathfrak{L}^2=\mathfrak{L}^2(\mu)$ is simply the set of functions
	$g\in\mathfrak{L}^0=\mathfrak{L}^0(\mu)$ such that $|g|^2$ is integrable,
	and $L^2=L^2(\mu)$ is defined as the set of functions
	$\{g^{\ssbullet}:g\in\mathfrak{L}^2\}\subseteq L^0=L^0(\mu)$.
	
	Also $L^2(\mu)$ is a Riesz subspace of $L^0(\mu)$. In 
	Proposition~\ref{prop:L0_func_space}, we have already proved that $T_{\phi}:L^0(\nu)\to L^0(\mu)$
	is a Riesz homomorphism, that is, a linear operator which preserves both multiplicative and lattice structure.
	Since here $T:L^2(\nu)\to L^2(\mu)$ is really a restriction of $T_{\phi}$ operating on equivalence classes $g^{\ssbullet}$ with $g\in\mathfrak{L}^0=\mathfrak{L}^0(\mu)$ such that $|g|^2$ is integrable. Hence $T:L^2(\nu)\to L^2(\mu)$ is also a Riesz homomorphism
	and therefore a linear operator which preserves both multiplicative and lattice structure. All now remains to be proven is that, it is
	also norm preserving.
	
	For this recall that $\|\,\|_2:L^2=L^2(\mu)\to\coint{0,\infty}$ is actually a norm defined by writing
	$\|g^{\ssbullet}\|_2=\|g\|_2$ for every $g\in\mathfrak{L}^2$ where $\|g\|_2=(\int|g|^2)^{1/2}$. Now
	a real-valued $|g|^2 = f$ is integrable means it is expressible as
	$f'-f''$ with integral as {$\int f=\int f'-\int f''$.} where $f'$ and $f''$ are such that: 
	(i) the domain of $f'$ is a conegligible subset of $X$ and
	$f'(x)\in\coint{0,\infty}$ for each $x\in\dom f'$ (and same for $f''$),
	(ii) there exists a non-decreasing sequence $\sequencen{f'_n}$ of
	non-negative simple functions such that $\sup_{n\in\Bbb N}\int f'_n<\infty$ and
	$\lim_{n\to\infty}f'_n(x)\penalty-100=f'(x)$ for almost
	every $x\in X$ (and same for $f''$). Hence $\int f'=\lim_{n\to\infty}\int f'_n$ whenever $\sequencen{f'_n}$ is a
	non-decreasing sequence of simple functions converging to $f'$ almost everywhere. Now for simplicity we assume 
	$|g|^2 = f$ is a simple function and prove for this case, but general case is not much difficult since it simply involves taking supremum of non-decreasing sequence of simple functions as described earlier. In this special case, $\int |g|^2=\int f=\sum_{i=0}^m a_i\nu F_i$ where $f=\sum_{i=0}^m a_i\chi F_i$ and every $F_i \in \Sigma_Y$ is a measurable set of finite measure. But $|g|^2$ being simple means $g$ is also simple and $a_i = |b_i|^2$. Hence $\int g = \sum_{i=0}^m b_i\nu F_i$ where $g=\sum_{i=0}^m b_i\chi F_i$. Now $g\phi =(\sum_{i=0}^m b_i\chi F_i)\phi$, but $\chi F_i\phi = \chi E_i$ since there exist $E_i\in\Sigma_X$ such that $\phi^{-1}[F_i]=E_i$. Hence $g\phi =\sum_{i=0}^m b_i\chi E_i$ and $|g\phi|^2 =\sum_{i=0}^m a_i\chi E_i$. Now using crucial property of inverse-measure-preserving, that is $\mu E_i=\mu(\phi^{-1}[F_i])=\nu F_i$, $\int |g\phi|^2 =\sum_{i=0}^m a_i\mu E_i =\sum_{i=0}^m a_i\nu F_i =\int |g|^2$. Therefore $\|g\phi\|_2= \|g\|_2$ implying $\|Tg^{\ssbullet}\|_2=\|g^{\ssbullet}\|_2$. Indeed generalizing to non-simple integrable $g$, it is proved that $T:L^2(\nu)\to L^2(\mu)$ is a norm preserving Riesz homomorphism.
	
\end{proof} 

The spaces $L^p(X,\Sigma_{X},\mu)$ for any $p \in [1,\infty]$ are Banach Lattices, hence instead of considering simply $\mathbf{Riesz}$ the category of Riesz spaces and Riesz homomorphisms, one can consider a codomain category $\mathbf{BanLatt}$ the category of Banach Lattices (complete normed Riesz spaces) and Riesz homomorphisms (order-continuous Positive operators) which are contractions. The objects additionally also have the property of multiplication and therefore homomorphisms become multiplicative when they preserve this property.    

\begin{definition}
	\label{def:L2_func_contra} 
	Let $\mathbf{LocMeasure}$, $\mathbf{BanLatt}$ be the categories as defined earlier. The (contravariant) functor $L^2:$ $\mathbf{LocMeasure}$ $\rightarrow \mathbf{BanLatt}$ is defined as a mapping that associates to each measure space $(X,\Sigma_{X},\mu)$ an object $L^2(X,\Sigma_{X},\mu):=\{f^{\ssbullet} \ | f \in \mathfrak{L}^0(\mu) ,|f|^2 \text{ is integrable}\}$. The assignment of inverse-measure -preserving morphism $\phi:(X,\Sigma_{X},\mu) \rightarrow (Y,\Sigma_{Y},\nu) $ is given by
	\begin{equation}
	L^2(\phi)(g^{\ssbullet}) = (g\phi)^{\ssbullet}  
	\end{equation}
	where $L^2(\phi): L^2(Y,\Sigma_{Y},\nu) \rightarrow L^2(X,\Sigma_{X},\mu)$ is the morphism in $\mathbf{BanLatt}$ between Banach lattices with multiplication, which is also multiplicative norm-preserving, that is $\|L^2(\phi)g^{\ssbullet}\|_2=\|g^{\ssbullet}\|_2$.
\end{definition}

The functoriality of $L^2$ easily gets verified in the same manner as $L^0$.

Once again depending upon the particular structures or their combinations such as vector, partial order, lattice, multiplication and complete norm one can easily form variations of Definition~\ref{def:L2_func_contra} which is a sort of prototype, using different codomain categories. 

From viewpoint of signal representation, the most common variation is to consider structure of completed normed vector space or Banach space and generalize to complex $L^2$ using spaces based on complex-valued functions. Again we need to consider
$\mathfrak{L^0}_{\Bbb C} =\mathfrak{L^0}_{\Bbb C}(\nu)$ for the space of complex-valued functions $g$ that are virtually measurable. In other words, $\mathfrak{L^0}_{\Bbb C}(\nu)$ is just the set of complex-valued functions $f$, defined on subsets of $X$, which are equal almost everywhere to some $\Sigma_X$-measurable function $h$ from $X$ to $\Bbb C$.

\begin{definition}
	\label{def:L2_func_contra_Hilb} 
	Let $\mathbf{Measure}$, $\mathbf{Hilb}$ be the categories as defined earlier. The (contravariant) functor $L^2:$ $\mathbf{Measure}$ $\rightarrow \mathbf{Hilb}$ is defined as a mapping that associates to each measure space $(X,\Sigma_{X},\mu)$ an object $L^2(X,\Sigma_{X},\mu):=\{f^{\ssbullet} \ | f \in \mathfrak{L}^0_{\Bbb C}(\mu) ,|f|^2 \text{ is integrable}\} :=\{h : X \rightarrow \mathbb{C} \ | \int |h(x)|^2 d\mu < \infty\}$. The assignment of inverse-measure -preserving morphism $\phi:(X,\Sigma_{X},\mu) \rightarrow (Y,\Sigma_{Y},\nu) $ is given by
	\begin{equation}
	L^2(\phi)(g^{\ssbullet}) = (g\phi)^{\ssbullet}  
	\end{equation}
	where $L^2(\phi): L^2(Y,\Sigma_{Y},\nu) \rightarrow L^2(X,\Sigma_{X},\mu)$ is the morphism in $\mathbf{Hilb}$ between Hilbert spaces which is also norm-preserving, that is $\|L^2(\phi)g^{\ssbullet}\|_2=\|g^{\ssbullet}\|_2$.
\end{definition}

Again in the special case of counting measures on sets $X$ and $Y$ the Hilbert spaces $L^2(\mu)$ becomes $\mathfrak{L^2}(\mu)$ itself and customarily denoted as $l^2(X)$ since $\Sigma_{X} = \Cal PX$ and every set except the empty set has a non-zero measure or in other words the null ideal is trivial. Hence we restrict to a subcategory $\mathbf{countMeasure}$ of $\mathbf{Measure}$ with objects of type $(X,\Cal PX,counting)$ where $counting$ is point-supported counting measure.

\begin{definition}
	\label{def:l2_func_contra_Hilb} 
	Let $\mathbf{countMeasure}$, $\mathbf{Hilb}$ be the categories as defined earlier. The (contravariant) functor $l^2:$ $\mathbf{countMeasure}$ $\rightarrow \mathbf{Hilb}$ is defined as a mapping that associates to each measure space $(X,\Cal PX,counting)$ an object $l^2(X,\Cal PX,counting):=\{f \in \mathfrak{L}^0_{\Bbb C}(\mu) ,|f|^2 \text{ is integrable}\} :=\{f : X \rightarrow \mathbb{C} \ | \sum_{x\in X}|f(x)|^2  < \infty\}$. The assignment of inverse-measure -preserving morphism $\phi:(X,\Cal PX,counting) \rightarrow (X,\Cal PY,counting) $ is given by
	\begin{equation}
	l^2(\phi)(g) = (g\phi)  
	\end{equation}
	where $l^2(\phi): l^2(Y,\Cal PY,counting) \rightarrow l^2(X,\Cal PX,counting)$ is the morphism in $\mathbf{Hilb}$ between Hilbert spaces which is also norm-preserving, that is $\|l^2(\phi)g\|_2=\|g\|_2$.
\end{definition}

Before we can handle the dual of Proposition~\ref{prop:L2_func_space} which leads to a covariant form of $L^2$ from opposite of $\mathbf{LocMeasure}$ which we denote as $\mathbf{MeasureAlg}$, it is important to note that if $\frak B$ is
any algebra carrying two non-isomorphic measures say $\mu$ and $\mu'$ , the
corresponding $L^2(\mu)$ and $L^2(\mu')$ spaces are still isomorphic. Hence rather than defining the functor from $\mathbf{MeasureAlg}$ often functor $L^P (1 \le p < \infty)$ is determined by Boolean ring of elements of finite measure in a measure algebra than in terms of the whole algebra. Note that the algebra $\frak B$ is uniquely determined in certain cases but the measure $\bar\mu$ is never determined.   If $(\frak B,\bar\mu)$ is a measure algebra, then a Boolean ring is the ideal $\frak B^f=\{b:b\in\frak B,\,\bar\mu b<\infty\}$ whereas a ring homomorphism $\pi:\frak B^f\to\frak A^f$ is termed as
{\bf measure-preserving} if $\bar\nu\pi b=\bar\mu b$ for every $b\in\frak B^f$. Theorem~\ref{thm:ring_algebra_dep} makes this dependency precise.

\begin{theorem}
	\label{thm:ring_algebra_dep}
	If $(\frak B,\bar\mu)$ and $(\frak A,\bar\nu)$
	are measure algebras and $\pi:\frak B^f\to\frak A^f$ a
	measure-preserving ring homomorphism between corresponding rings of elements of finite measure, then
	
	(a) There is a unique order-continuous norm-preserving Riesz
	homomorphism $T_{\pi}:L^2(\frak B,\bar\mu)\to L^2(\frak A,\bar\nu)$ such
	that $\Bvalue{T_{\pi}u>a}=\pi\Bvalue{u>a}$, $\|T_{\pi}u\|_2=\|u\|_2$ for every
	$u\in L^2(\frak B,\bar\mu)$ and $a>0$.
	
	(b) If $(\frak C,\bar\lambda)$ is another measure algebra and
	$\theta:\frak A^f\to\frak C^f$ another measure-preserving ring
	homomorphism, then $T_{\theta\pi}=T_{\theta}T_{\pi}:
	L^2(\frak B,\bar\mu)\to L^2(\frak C,\bar\lambda)$.
\end{theorem}

For proof pertaining to general $p$ see~\cite{fremlinmt3}. This theorem therefore gives rise to an explicit Definition~\ref{def:L2_func_cov}.

\begin{definition}
	\label{def:L2_func_cov} 
	Let $\mathbf{MeasureAlg}$, $\mathbf{BanLatt}$ be the categories as defined earlier. The (covariant) functor $L^2:$ $\mathbf{MeasAlg}$ $\rightarrow \mathbf{BanLatt}$ is defined as a mapping that associates to each measure algebra $(\frak B,\bar\mu)$ an object $L^2(\frak B,\bar\mu):=\{u \ | a\mapsto\Bvalue{u>a} \text{for every}\ a\in\Bbb R \}$. The assignment of morphism $\pi:(\frak B,\bar\mu) \rightarrow (\frak A,\bar\nu) $ is given by $L^2(\pi)(u)$ defined by
	\begin{equation}
	\Bvalue{L^2(\pi)u>a}=\pi'\Bvalue{u>a};  \pi':\frak B^f\to\frak A^f, a>0 
	\end{equation}
	where $L^2(\pi): L^2(\frak B,\bar\mu) \rightarrow L^2(\frak A,\bar\nu)$ is the morphism in $\mathbf{BanLatt}$ between Banach lattices with multiplication, which is also multiplicative norm-preserving, that is $\|L^2(\pi)u\|_2=\|u\|_2$.
\end{definition}
As expected, $L^2(\pi): L^2(\frak B,\bar\mu) \rightarrow L^2(\frak A,\bar\nu)$ corresponds to the map
$g^{\ssbullet}\mapsto(g\phi)^{\ssbullet}:L^1(\nu)\to L^1(\mu)$ of Theorem~\ref{prop:L2_func_space}, when $\phi:X\to Y$ is an inverse-measure-preserving function between $(X,\Sigma,\mu)$ and $(Y,\Tau,\nu)$ with $\pi:(\frak B,\bar\mu)\to(\frak A,\bar\nu)$ the corresponding measure-preserving homomorphism.

Finally Table~\ref{table:duality} quickly summarizes the essential functorial definitions and their dual versions which we considered in detail here.

\begin{center}
	\begin{table}[ht]
		\centering
		\resizebox{\linewidth}{!}{
			\begin{tabular}{|c|c|}\hline
				Contravariant Functorial Model & Dual Covariant Functorial Model \\ \hline
				$\xymatrix{
					(\Bbb R,\Sigma_{\Cal B}) \ar[r]^{id} & (\Bbb R,\Sigma_{\Cal B})  \\
					(I,\Sigma_{I}) \ar[r]_{\phi} \ar@{.>}[ru]^{g\phi} \ar@{->}[u]^{g\phi} & (J,\Sigma_{J}) \ar@{->}[u]_{g}
				}$ & $\xymatrix{
					(\Sigma_{\Cal B},\symmdiff,\cap) \ar@{->}[d]_{\pi \xi} & (\Sigma_{\Cal B},\symmdiff,\cap) \ar[l]^{id} \ar@{.>}[ld]^{\pi \xi} \ar@{->}[d]^{\xi (E \mapsto g^{-1}[E])} \\
					(\Sigma_{I},\symmdiff,\cap)   & (\Sigma_{J},\symmdiff,\cap) \ar[l]^{\pi} 
				}$ \\ \hline
				$\xymatrix{
					L^0(I,\Sigma_{I},\Cal N(\mu)) & L^0(J,\Sigma_{J},\Cal N(\nu)) \ar[l]^{T = L^0(\phi^\ssbullet)}  \\
					(I,\Sigma_{I},\Cal N(\mu)) \ar[r]_{\phi^\ssbullet} & (J,\Sigma_{J},\Cal N(\nu))
				}$ & $\xymatrix{
					L^0(\frak B) & L^0(\frak B) \ar[l]^{L^0(\pi_\phi)}  \\
					(\Sigma_{I}/(\Sigma_{I} \cap \Cal N(\mu)))   & (\Sigma_{J}/(\Sigma_{J} \cap \Cal N(\nu))) \ar[l]^{\pi_\phi} 
				}$ \\ \hline
				$\xymatrix{
					L^2(\mu) & L^2(\nu) \ar[l]^{T = L^2(\phi^\ssbullet)}  \\
					(I,\Sigma_{I},\mu) \ar[r]_{\phi^\ssbullet} & (J,\Sigma_{J},\nu)
				}$ & $\xymatrix{
					L^2(\frak B,\bar\mu) & L^2(\frak A,\bar\nu) \ar[l]^{L^2(\pi_\phi)}  \\
					\frak B^f   & \frak A^f \ar[l]^{\pi_\phi|_f} \\
					(\frak B,\bar\mu)   & (\frak A,\bar\nu) \ar[l]^{\pi_\phi} 
				}$ \\ \hline
				$L^0:$ $\mathbf{LocMeas}$ $\rightarrow \mathbf{Riesz}$ & $L^0:$ $\mathbf{compBoolAlg}$ $\rightarrow \mathbf{Riesz}$. \\ \hline
				$L^2:$ $\mathbf{LocMeasure}$ $\rightarrow \mathbf{BanLatt}$ & $L^2:$ $\mathbf{MeasureAlg}$ $\rightarrow \mathbf{BanLatt}$  \\ \hline
			\end{tabular}
		}
		\caption{Duality between Covariant and Contravariant forms of $L^0$ and $L^2$ Summarized.}
		\label{table:duality}
	\end{table}
\end{center}

\subsection{Generalization to partial monic categories derived from $\mathbf{LocMeas}$}

For a reader purely interested in signal representation and redundacy from a category theory viewpoint this section can be omitted. The reason we are including it here is because these extensions were motivated by~\cite{Heunen2013} and~\cite{thesis} earlier when we constructed a visual proof of concept. However later we realized that unified treatment of measure theory in~\cite{fremlinmt1},~\cite{fremlinmt2},~\cite{fremlinmt3},~\cite{fremlinmt4} has studied the functorial aspects of $L^0$ and $L^p$ constructions especially from the categories of Boolean and measure algebras. This made our task easier while at the same time we could connect these works developing a much broader perspective and emphasizing the functorial aspects from the application viewpoint.

The category $\mathbf{LocMeas}$ being equivalent to the opposite of category of commutative von Neumann algebras~\cite{thesis} has desirable properties with all finite transversal limits which makes it possible to consider its monic and partial derivatives. Being in possession of a category with limits such as equalizers, products and therefore pullbacks, one can immediately form the monic and partial derived categories from the original category \cite{cockettlack}.

First we consider the $\mathcal{M}$-category of $\mathbf{LocMeas}$. 
The $\mathcal{M}$-category $(\mathbf{LocMeas},\mathcal{M})$ of (stable) monics in the category $\mathbf{LocMeas}$ consists of
\begin{itemize}
	\item \textbf{objects} the collection $(X,\Sigma_{X},\Cal N(\mu)),(Y,\Sigma_{Y},\Cal N(\nu)),... \in \mbox{Ob}(\mathbf{LocMeas})$
	\item \textbf{morphisms} for a pair $\mathbb{X},\mathbb{Y} \in \mbox{Ob}(\mathbf{LocMeas})$, $\mathbf{LocMeas}(\mathbb{X},\mathbb{Y})=\{f:\mathbb{X}\to \mathbb{Y}\mid f$ is \textit{monic} $\}$.
\end{itemize}

Since the isomorphisms in parent category are monic they get included in this category. Further pullbacks in $\mathbf{LocMeas}$ of monic along general map $f$ exists and is itself a monic (hence an $\mathcal{M}$-morphism).
\begin{equation*}  
\xymatrix@R=0.4in @C=0.8in{
	\mathbb{X} \times_{\mathbb{Y}} \mathbb{Z} \ar@{>->}[r]^{m'} \ar[d]_{f'} & \mathbb{Z} \ar[d]^{f} \\
	\mathbb{X} \ar@{  >->}[r]_{m} & \mathbb{Y}
}
\end{equation*}

A general $\mathcal{M}$-category $(\mathbf{C},\mathcal{M})$ is also sometimes denoted as $(\mathbf{C},Monic)$. An easier example is $(\mathbf{Set},\mathcal{M})$ or $(\mathbf{Set},Monic)$ which is the category of sets with all injections.

Next consider a partial (restriction) category Par$(\mathbf{LocMeas},\mathcal{M})$ of the above $\mathcal{M}$-category which consists of:
\begin{itemize}
	\item \textbf{objects}: a collection $\mathbb{X},\mathbb{Y},\mathbb{Z} ... \in \mbox{Ob}(\mathbf{LocMeas})$
	\item \textbf{morphisms}: pair $(m,f): \mathbb{X} \rightarrow \mathbb{Y}$ (upto equivalence) where $m \in (\mathbf{LocMeas},\mathcal{M})$ and $f \in \mathbf{LocMeas}$:
	$
	\xymatrix{
		& \mathbb{X'} \ar[rd]^f \ar@{>->}[ld]_m \\
		\mathbb{X}  & &  \mathbb{Y}
	}
	$
	\item \textbf{identity}: $(\mathbf{1}_\mathbb{X},\mathbf{1}_\mathbb{X}): \mathbb{X} \rightarrow \mathbb{X}$ 
	\item \textbf{composition}: $(m',g)(m,f)=(mm'',gf')$:
	$
	\xymatrix{
		& & \mathbb{X''} \ar[rd]^{f'} \ar@{>->}[ld]_{m''} \\ 
		& \mathbb{X'} \ar[rd]^f \ar@{>->}[ld]_{m} & & \mathbb{Y'} \ar@{>->}[ld]_{m'} \ar[rd]^g  \\
		\mathbb{X}  & &  \mathbb{Y} & & \mathbb{Z}
	}
	$
	\item \textbf{Restriction}: for the morphism $(m,f): \mathbb{X} \rightarrow \mathbb{Y}$, $\overline{(m,f)}=(m,m): \mathbb{X} \rightarrow \mathbb{X}$
\end{itemize}
with the usual unit and associative laws.

Up-to equivalence means we factor out by the equivalence relation $(\thicksim)$ where $(m,f) \thicksim (m',f')$ if there exists an isomorphism $g : \mathbb{X'} \rightarrow \mathbb{X''}$ such that $m' \circ g  =m$ and $f' \circ g = f$.

Here a basic example to keep in mind is Par$(\mathbf{Set},\mathcal{M})$ which is the category of sets and partial maps. Similarly Par$((\mathbf{Set},\mathcal{M}),\mathcal{M})$ is the category $\mathbf{PInj}$ of sets and partial injections. Further Par$(\mathbf{LocMeas},\mathcal{M})$ is the category of localizable measurable spaces and partial measurable maps. Par$((\mathbf{LocMeas},\mathcal{M}),\mathcal{M})$ is the category of localizable measurable spaces and partial monic measurable maps

Categories of partial maps were given the characterization of restriction structure and categories by \cite{cockettlack}. They are widely seen as a simple equational axiomatization for categories of partial maps. In-fact every restriction category embeds fully and faithfully into a partial map category. Very brief definitions of the restriction and dagger category are recalled in~\ref{app:parcat}. Using these following proposition~\ref{prop4} is immediate for the category studied in this subsection.

\begin{proposition}
	\label{prop4}
	Every inverse category is a dagger category. In particular Par$((\mathbf{LocMeas},\mathcal{M}),\mathcal{M})$ is a dagger category.
\end{proposition}
\begin{proof}
	From the definitions of inverse and dagger categories we immediately deduce that dagger structure is given by  $\dagger(1_X) = {1_X}$ for all objects $X$ and ${f}^{\dagger} = {f}^{\circ}$ for all morphisms $f$ in the inverse category $\mathbf{C}$. It is clear that Par$((\mathbf{LocMeas},\mathcal{M}),\mathcal{M})$ being an inverse category is also dagger, which implies that every partial monic measurable non-singular morphism can be reversed uniquely which we denote as its adjoint morphism.
\end{proof}

Finally we are ready to extend the basic $L^0$ functor to the domain category constructed here.

\subsection{The $L^0$ Functor extended to partial category Par$(\mathbf{LocMeas},\mathcal{M})$}

In this section we briefly touch upon the well-known case of functor $l^2: \mathbf{PInj}^{op} \rightarrow \mathbf{Hilb}$ which was first observed by \cite{Barr}, developed in \cite{Heunen2013} and studied as an example of unique decomposition categories in [Bob Coecke]. This section is not essential purely from the perspective of signal representation however we studied this case and tried generalizing it to measure spaces from sets in reverse fashion. We show how it can be recast as the special case of Definition~\ref{def:par_L0}.        

\begin{definition}
	\label{def:par_L0} 
	Let Par$(\mathbf{LocMeas},\mathcal{M})$ ,$\mathbf{Riesz}$ be the categories. The functor $L^0:$ Par$(\mathbf{LocMeas},\mathcal{M})$ $\rightarrow \mathbf{Riesz}$ is defined as a mapping that associates to each measurable space $(Y,\Sigma_{Y},\Cal N(\nu))$ an object $L^0(Y,\Sigma_{Y},\Cal N(\nu)):=\{g^{\ssbullet} \ | g \in \mathfrak{L}^0(\nu)\}$. The assignment of morphism $f^{\ssbullet} :\mathbb{X} \rightarrow \mathbb{Y} $ or $(\xyline{\mathbb{X} & \mathbb{W} \ar@{ >->}|-{f_1}[l] \ar@{ ->}|-{f_2}[r] & \mathbb{Y}})$ is given by
	\begin{equation}
	L^0(f^{\ssbullet})(g^{\ssbullet}) = (gf)^{\ssbullet} or L^0(f^{\ssbullet})(g^{\ssbullet})(x)= g^{\ssbullet}(f_1(w)) 
	\end{equation}
	where $L^0(f^{\ssbullet}): L^0(\mathbb{Y}) \rightarrow L^0(\mathbb{X})$ is the morphism in $\mathbf{Riesz}$.
\end{definition}

\subsection{The case of linear Borel measurable function $h: (\Bbb R,\Sigma_{\Cal B}) \rightarrow (\Bbb R,\Sigma_{\Cal B})$ }
The propositions~\ref{prop:L0_func_space} and~\ref{prop:L2_func_space} can be used for non-singular and inverse-measure-preserving maps between local measurable spaces with identity map on $(\Bbb R,\Sigma_{\Cal B})$. However for mathematical modeling of transformations involving change in signal amplitudes we need a general Borel measurable function $h: (\Bbb R,\Sigma_{\Cal B}) \rightarrow (\Bbb R,\Sigma_{\Cal B})$. Although this induces in general a function $\bar{h}:L^0(\mu) \rightarrow L^0(\mu)$ it is not necessarily linear transformation. If we need this endo-transformation to be invertible and linear (which becomes an arrow in category $\mathbf{Hilb}$) we restrict to linear $h$ which covers uniform amplitude scaling most frequent in image signals. The next proposition which we extend from~\cite{fremlinmt2} makes this precise. 

\begin{proposition}
	\label{prop:lin_bor_fun}
	If $(X,\Sigma,\mu)$ is a measure space and $h:\Bbb R\to\Bbb R$ is a Borel
	measurable function. Then $hf\in\mathfrak{L^0}=\mathfrak{L^0}(\mu)$ for
	every $f\in\mathfrak{L^0}$ and $hf\eae hg$ whenever
	$f\eae g$. Define a function $\bar h:L^0\to L^0$ by
	setting $\bar h(f^{\ssbullet})=(hf)^{\ssbullet}$ for every $f\in\mathfrak{L^0}$;
	then $\bar{h}$ becomes a linear operator when $h$ is a linear function.
\end{proposition}
\begin{proof}
	From Theorem~\ref{thm:prop_meas_func} parts (f) and (g) and Proposition~\ref{prop:properties_coneg_func} 
	it follows that $hf\in\mathfrak{L^0}=\mathfrak{L^0}(\mu)$ for every $f\in\mathfrak{L^0}$.
	
	For simplicity we assume that $f:X\to\Bbb R$ is $\Sigma_X$-measurable
	while $h:\Bbb R\to\Bbb R$ is $\Sigma_{\Cal B}$-measurable or more precisely
	$\dom h=\Bbb R$ and $\dom f=X$
	$f\eae g$, then $C=\{x:x\in\dom g,\,f(x)=g(x)\}$ is conegligible and
	$f\restr C=g\restr C$ is measurable. Now if $a\in\Bbb R$, then $\{y:h(y)<a\} = E$, where $E$ is a Borel subset of $\Bbb R$ and
	$(f\restr C)^{-1}[E]$ is of the form $F\cap\dom (f\restr C)$, where $F\in\Sigma_X$, then $\{x:(h(f\restr C))(x)<a\}=\{x:(h(g\restr C))(x)<a\}=F\cap C\in\Sigma_{C}$. As $a$ is arbitrary, $h(f\restr C)=h(g\restr C)$ is $\Sigma_{X}$-measurable. But $hf : X \rightarrow \Bbb R$ is $\Sigma_X$-measurable and extends $h(f\restr C) = hf \restr C$ on a conegligible $C \subseteq X$. Consequently $hf \eae h(f\restr C) \eae h(g\restr C)$ and $(hf)^{\ssbullet} = (hg)^{\ssbullet}$ in $L^0(\mu)$. Indeed we have well-defined $\bar{h}:L^0(\mu)\to L^0(\mu)$ where $\bar{h}f^{\ssbullet}=(hf)^{\ssbullet}$ for every $f \in \mathfrak{L}^0(\mu)$.
	
	Next By Theorem~\ref{thm:prop_meas_func} part (b), $f + f'$ is $\Sigma_X$-measurable,  where $(f + f')(x) = g(x) + g'(x) $ for $x \in X$. Therefore there exists an $F\in\Sigma_X$ such that $\{x:(f+f')(x)< a\}=\{x:f(x) + f'(x)< a\}=F$. But there exists a Borel $B = (\infty,a'] \in\Sigma_{\Cal B}$ such that $h^{-1}[B]=E=(\infty,a]$ and $f^{-1}[E]=F \in\Sigma_{X}$. Hence $\{x:[h(f+f')](x)\le a'\}=\{x:h(f(x)+f'(x))\le a'\}= F\in\Sigma_X$. Since $a'$ is arbitrary, $hf + hf' = h(f + f')$ is $\Sigma_X$-measurable. Now from Proposition~\ref{prop:properties_coneg_func} part (b) $(f + f')^{\ssbullet} = f^{\ssbullet} + f'^{\ssbullet}$ in $L^0(\mu)$ and using the fact $h(y+y')=hy+hy'$ for a linear $h$ we have $(h(f + f'))^{\ssbullet}= (hf + hf')^{\ssbullet} = (hf)^{\ssbullet} + (hf')^{\ssbullet} $ in $L^0(\mu)$. Hence $\bar{h}(f + f')^{\ssbullet}=(h(f + f'))^{\ssbullet}=(hf)^{\ssbullet} + (hf')^{\ssbullet} = \bar{h}f^{\ssbullet} + \bar{h}f'^{\ssbullet} $. Similarly using Theorem~\ref{thm:prop_meas_func} and Proposition~\ref{prop:properties_coneg_func} part (c) it follows in same manner for any scalar $c \in \Bbb R$, that $\bar{h}(c.f)^{\ssbullet}=(h(c.f))^{\ssbullet}= c.(hf)^{\ssbullet} = c. \bar{h}f^{\ssbullet} $. Hence $\bar{h}$ is a linear operator.	
\end{proof}

\section{Signal as functor}
\label{sigfunc}
In Section~\ref{sec:motiv} we summarized in Table~\ref{table:summary_sig_rep} the natural differences which are our prime motivations for proposing the concept of functorial signal representation. In this section we discuss these differences in detail using Figure~\ref{fig:prototype_thesis_1} and Equation~\ref{eq:sig_func}. Note that one can formulate all these differences alternatively using Equation~\ref{eq:sig_func_cat_1} instead of Equation~\ref{eq:sig_func} by substituting $F(G_1) = G_1l, F(G_2) = G_2l, F(a) = (a_M,a_o) = (h,\phi)$ using the fact that $\mathbf{Meas}^{2}$ is same as $\mathbf{Meas}^{\rightarrow}$ appropriately as discussed in Section~\ref{meas2}.

\begin{figure}
	\resizebox*{\textwidth}{!}{
		\begin{tikzpicture}
		\draw[domain=-1.57:1.57,samples=100] plot(\x,{sin(2*\x r)});
		\draw[domain=7.85:11,samples=100] plot(\x,{sin(2*\x r)});
		\draw plot [smooth] coordinates {(1.57,0) (3,-0.5) (4,0.75) (5,1) (6,-0.4) (7,1) (7.85,0)};
		\draw[<->] (-4.5,0) -- (13,0) node [pos=1,below] {$\mathbb{R}$};
		\draw plot [smooth] coordinates {(-1.57,0) (-2.5,1) (-2.8,0) (-3.5,-0.5) (-3.7,-1) (-4,0)};
		\draw plot [smooth] coordinates {(11,0) (11.5,-0.5) (12,1) (12.5,0.5)};
		\draw[->] (-3.5,-1) --(-3.5,2.0) node [pos=1,left] {$\mathbb{R}$};
		\draw[-,dashed] (-1.57,1.5) -- (-1.57,-1.5);
		\draw[-,dashed] (1.57,1.5) -- (1.57,-1.5);
		
		\draw[-,dashed] (7.85,1.5) -- (7.85,-1.5);
		\draw[-,dashed] (11,1.5) -- (11,-1.5);
		\draw[<->] (-1.57,-1.3) --(1.57,-1.3) node[midway,fill=white] {$I$};
		\draw[<->] (7.85,-1.3) --(11,-1.3) node[midway,fill=white] {$J$};
		\node at (0,1.3) [draw=none] {$f\restr I = FG_1$};
		\node at (9.4,1.3) [draw=none] {$f\restr J = FG_2$};
		\node at (-2.5,0) [draw=none,above left] {$(0,0)$};
		\node at (12.5,0) [draw=none,above right] {$time$};
		\node at (-3.5,1.5) [draw=none,above right] {$amplitude$};
		\end{tikzpicture}
	}
	\caption{In complete time signal $f$ (global on its entire duration), the restrictions $f\restr I$ and $f\restr J$ (or local sub-signals) on half-open intervals $I$, $J$ are generated by objects $G_1$, $G_2$ of base category isomorphic to each other through arrow $a$.}
	\label{fig:prototype_thesis_1}
\end{figure}
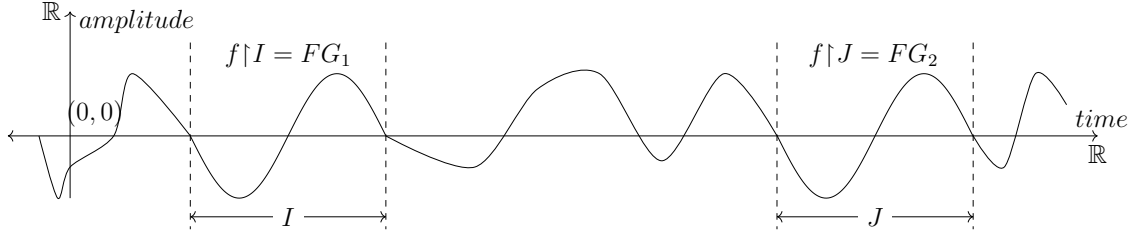

\begin{equation}
\label{eq:sig_func}
\xymatrix{
	G_1 \ar[r]^{a} & G_2
} 
\quad \\
\xrightarrow{F}
\quad
\xymatrix{
	(\Bbb R,\Sigma_{\Cal B}) \ar[r]^{h} & (\Bbb R,\Sigma_{\Cal B})  \\
	(I,\Sigma_{I}) \ar[r]_{\phi} \ar@{->}[u]^{f\restr I = F(G_1)} & (J,\Sigma_{J}) \ar@{->}[u]_{F(G_2) \cong F(a)[F(G_1)]}
}
\end{equation}
\begin{equation}
\label{eq:sig_func_cat_1}
\xymatrixcolsep{1.5cm}
\xymatrix{
	\mathbf{2}
	\ar@/^2pc/[rr]_{\quad}^{G_1}="1"
	\ar@/_2pc/[rr]_{G_2}="2"
	&& \mathbf{Meas}
	\ar@{}"1";"2"|(.2){\,}="7"
	\ar@{}"1";"2"|(.8){\,}="8"
	\ar@{=>}"7" ;"8"^{a} } \quad 
\quad
\xymatrix{
	G_1M=(\Bbb R,\Sigma_{\Cal B}) \ar@{<->}[r]^{a_{M}= h} 	& Sc_2=(\Bbb R,\Sigma_{\Cal B}) \\
	G_1O=(I,\Sigma_{I})	\ar@{<->}[r]_{a_{O}=\phi}	\ar@{->}[u]^{G_1l=f\restr I} & Sc_1=(J,\Sigma_{J}) 	\ar@{->}[u]^{G_2l \cong (h,\phi)f\restr I} \\
}
\end{equation}

\subsection{Underlying Intuition}

The classic model represents a signal as some real $\mathbb{R}$ or complex $\mathbb{C}$-valued function\footnote{function $f:\mathcal{D} \rightarrow \mathcal{R}$ is a mapping from one set of values the domain $\mathcal{D}$, to another set the codomain $\mathcal{R}$, assigning each member $x \in \mathcal{D}$, the value $f(x) \in \mathcal{R}$.} of time or space denoted as $f:\Bbb R \rightarrow \Bbb R$. Depending on the mathematical structures such as measure, topology, group on the domain $\Bbb R$, tools from measure theory, functional or harmonic analysis are utilized in analyzing the signal. The proposed model represents same signal appropriately as a functor\footnote{functor is a mapping from $\mathbf{C}$ the domain category, to $\mathbf{D}$ the codomain category, assigning each object $O \in \mathbf{C}$, an object $F(O) \in \mathbf{D}$, and each arrow $a:O \rightarrow M \in \mathbf{C}$, an arrow $F(a):F(O) \rightarrow F(M) \in \mathbf{D}$, preserving the composition of arrows $F(a \cdot b) = F(a) \circ F(b)$ and the identity arrows $F(id_O) = id_{F(O)}$.} $F:\mathbf{C} \to \mathbf{D}$ from a base category to a codomain category additionally bringing the powerful relative perspective and other tools from category theory. The base category mathematically models the natural generative structure~\cite{Leyton01} of a signal. Its objects correspond to true generators while the arrows capture relationship between them. The codomain category is at times arrow category of another category. This is often the case since signals classically considered as functions are most appropriately modeled as arrows of a category. The study of functor as a structure-preserving-morphism and trivial categorification can be found in~\cite{salilp1}. The base category interpreted as category of generators and transfers serves to capture intuition of Leyton's generative theory as in~\cite{Leyton86a},~\cite{Leyton01} and will often be a groupoid. 

\begin{center}
	\begin{table}[ht]
		\centering
		\resizebox{\linewidth}{!}{
			\begin{tabular}{|c|c|}\hline
				{\bf Classical function based model} & {\bf Proposed functorial model} \\ \hline
				\makecell{Signal as real (or complex) valued entity varying \\in time or space domains modeled mathematically 
					\\ as a real (or complex) valued function of time or space.} & 
				\makecell{Signal as a functor from a category (modeling\\generative structure)to a category (modeling observed waveforms)}\\ \hline
			\end{tabular}
		}
		\caption{Difference in underlying intuitions between conventional and proposed signal representation model}
		\label{table:diff1_sig_rep}
	\end{table}
\end{center}

\subsection{Case of basic measurable structure}

Considering a measurable structure on the domain $\Bbb R$ while taking the signal to be real-valued the classical model is precisely a Lebesgue\footnote{$\hat\Sigma_{\Cal B}$ is $\sigma$-algebra of Lebesgue measurable sets, the completion of $\Sigma_{\Cal B}$ the $\sigma$-algebra of Borel subsets of $\Bbb R$.}-measurable function\footnote{A ($\Sigma_{X}$-) measurable function $f:(X,\Sigma_{X}) \to (Y,\Sigma_{Y})$ is a function $f: X \to Y$ such that $\{f^{-1}[F]:F\in\Sigma_{Y}\} \subseteq \Sigma_X$.} $f:(\Bbb R,\hat\Sigma_{\Cal B}) \to (\Bbb R,\Sigma_{\Cal B})$. The proposed model represents same signal as a functor $F:\mathbf{C} \to \mathbf{Meas}^{\rightarrow}$\footnote{{$\mathbf{Meas}$} is the category of measurable spaces $(X,\Sigma_{X})$ and measurable functions $f:(X,\Sigma_{X}) \to (Y,\Sigma_{Y})$. $\mathbf{Meas}^{\rightarrow}$ is the arrow category of $\mathbf{Meas}$ with measurable functions $f:(X,\Sigma_{X}) \to (Y,\Sigma_{Y})$ as objects and pairs of measurable functions $(h,\phi): f \to g$ as arrows such that $h \circ f = g \circ \phi$.} denoted simply as the image subcategory $F\mathbf{C}$. The work in~\cite{robinson1} is the study of signal sheaf $F: \mathfrak{Face}_X \rightarrow \mathbf{Vect}$ using face category of simplical complex $X$ based on presheaf model of a signal $F:\mathfrak{Opn}_X \rightarrow \mathbf{Vect}$. These are based on topological structure of signal domain appropriately modeled as base category. In our proposed model $F:\mathbf{C} \to \mathbf{D}$, the base category $\mathbf{C}$ is expected to capture the natural generative structure of signal to be represented; the intuition of which is based on work in~\cite{Leyton86a},~\cite{Leyton01}. The major difference in our approach is that the mathematical structure to be considered on time or space domain is a groupoid determined by generative category $\mathbf{C}$ specific to the signal and depending on choice of $\mathbf{D}$ can vary over multiple structures such as measure(which will be explored here choosing $\mathbf{D} = \mathbf{Meas}^{\rightarrow}$), topology and graph which are present in specific classes of signals. Considering a signal generally as $F\mathbf{C}$, leads us towards general arrow-based redundancy, Grothendieck's relative viewpoint in signal analysis and understanding how compression basically occurs.

\begin{center}
	\begin{table}[ht]
		\centering
		\begin{tabular}{|c|c|}\hline
			{\bf Classical function based model} & {\bf Proposed functorial model} \\ \hline
			\makecell{simple 1D time signal as measurable\\ function $f:(\Bbb R,\hat\Sigma_{\Cal B}) \to (\Bbb R,\Sigma_{\Cal B})$} & 
			\makecell{Simple 1D time signal as $F\mathbf{C}$\\ using a functor $F:\mathbf{C} \to \mathbf{Meas}^{\rightarrow}$}\\ \hline
		\end{tabular}
		\caption{Difference in mathematical expressions using most basic measurable structure}
		\label{table:diff2_sig_rep}
	\end{table}
\end{center}

\subsection{Trivial arrows versus Non-trivial arrows}

The classic model $f:(\Bbb R,\hat\Sigma_{\Cal B}) \to (\Bbb R,\Sigma_{\Cal B})$ can be considered as a special case of proposed model $F:\mathbf{1}$\footnote{$\mathbf{1}$ is the trivial category with a single (trivial) object $1$ and just identity (trivial) arrow $id_1$.} $\to \mathbf{Meas}^{\rightarrow}$ (or $F:\mathbf{C} \to \mathbf{Meas}^{\rightarrow}$ with $\mathbf{C}$ as a discrete category) where $F(\mathbf{1})$ is a category with object $F(1) = f: (\Bbb R,\hat\Sigma_{\Cal B}) \to (\Bbb R,\Sigma_{\Cal B})$ along with identity $F(id_1) = id_f$ (or discrete image subcategory consisting of objects $f\restr I,f\restr J,...$). Thus the classic model from the viewpoint of proposed model is essentially a single object $f$ of $\mathbf{Meas}^{\rightarrow}$ with trivial identity or a discrete subcategory of $\mathbf{Meas}^{\rightarrow}$. The proposed model is then viewed as generalization to arbitrary base category. If we denote objects of $\mathbf{C}$ as $G_1, G_2,...$ and the arrows as $a, a', a'',...$. Then signal is a collection or family of objects $F(G_1) = f\restr I$\footnote{$f\restr I$ is the restriction ($f|_{I}$) of function $f$ to domain $I$.} $,F(G_2) = f\restr J,...$ and the arrows $F(a) = (h,\phi), F(b) = (h',\phi'),...$. The multiple non-trivial arrows across objects in proposed model truly bring the relative perspective of whole category theory in signal analysis. In arrow category, $(h,\phi) : f \to g$ where $h \circ f = g \circ \phi$; the pair $(h,\phi)$ is not unique but general and in particular specifying $f \to g$ makes it unique in that context. Hence we can perfectly have $(h,\phi) : f' \to g'$ where $h \circ f' = g' \circ \phi$ and therefore it is always understood with particular domain and codomain.

\begin{center}
	\begin{table}[ht]
		\centering
		\resizebox{\linewidth}{!}{
			\begin{tabular}{|c|c|}\hline
				{\bf Classical function based model} & {\bf Proposed functorial model} \\ \hline
				\makecell{Signal $f:(\Bbb R,\hat\Sigma_{\Cal B}) \to (\Bbb R,\Sigma_{\Cal B})$ viewed as a special case\\$F:\mathbf{C} \to \mathbf{Meas}^{\rightarrow}$; $\mathbf{C}$ being some discrete category.} & 
				\makecell{Signal as a family of objects $F(G_1) = f\restr I, F(G_2) =f\restr J,...$\\ along with non-trivial arrows $F(a_1),F(a_2),...$}\\ \hline
			\end{tabular}
		}
		\caption{Difference of trivial arrows versus additional non-trivial arrows}
		\label{table:diff3_sig_rep}
	\end{table}
\end{center}

\subsection{Collection of independent elements versus arrow based generalized elements and differentials.}

Now referring Figure~\ref{fig:visual_depiction_signal}, in set-theoretic measure theory the signal $f \in \mathfrak{L^0}_{\Bbb R}$\footnote{$\mathfrak{L^0}_{X}$ is the Riesz space of all measurable real-valued functions $f$ defined on $X$. $\mathfrak{L^0}_{X} \subseteq \mathfrak{L}^0(X,\Sigma_{X}) = \mathfrak{L}^0= \mathfrak{L}^0(\mu)$ where $\mathfrak{L}^0(\mu)$ is the set of virtually measurable real-valued functions $f$ defined on conegligible subsets of $X$.} is simply an element of a global space $\mathfrak{L^0}_{\Bbb R}$ canonically isomorphic to $\prod_{i\in (I,J,..)}\mathfrak{L^0}_{i}$ via $f\mapsto(f\restr I,f\restr J,...)$. Thus the collection $(f\restr I,f\restr J,...)$ in classic case has local components such as $f\restr I$ which are always linearly independent. This independence in classic model does not authentically reflect the natural generative relationship between these components which make them interdependent typical to sources with memory from information theory viewpoint. The proposed model precisely utilizes the arrows and fundamental concept of generalized elements to represent the components differentially. Referring Equation~\ref{eq:prac_sig_func}, $F(G_1)=f\restr I$, $F(G_2)=f\restr J$ and $F(a)=(h,\phi)$ this implies $f\restr J$ is related to $f\restr I$ via arrow $F(a)$ therefore we express $f\restr J = f\restr I$-valued point of $f\restr J + \Delta_J$ where $\Delta_J$ is differential error in actual measured signal. Thus signal becomes a collection $(f\restr I,F(a)f\restr I + \Delta_J,...)$ where generalized elements\footnote{A generalized element of N (also called M-valued point of N) is just an arrow $a : M \rightarrow N$ in a category. Hence an element of a set $S$ in set-theory is simply the arrow $e : T \rightarrow S$ in the category $\mathbf{Set}$ where $T$ is a terminal object.} and $\Delta$s together faithfully model the interdependencies between these components. Differential $\Delta_J = F(G_2) - F(a)F(G_1)$ in Riesz space $\mathfrak{L^0}_{J}$ roughly indicates the linear deviation of $G_2$ from $G_1$ and is relatively small when arrow $a$ is (total) isomorphism. Note that vector addition in $F(a)F(G_1) + \Delta_J$ is not a coproduct operation in $\mathbf{Meas}^{\rightarrow}$ but an operation in $\mathfrak{L^0}_{J}$ putting certain limitations on faithfulness of functor $F$ which we discuss in Section~\ref*{limitation}. Note that $f\restr J:(J,\Sigma_{J}) \to (\Bbb R,\Sigma_{\Cal B})$ is a measurable function therefore an object of $\mathbf{Meas}^{\rightarrow}$. But it is also a trivial category and therefore by considering additional property of $\Bbb R$ being a field, it is can be point-wise added to and multiplied by any other measurable function $g\restr J:(J,\Sigma_{J}) \to (\Bbb R,\Sigma_{\Cal B})$. In other words, $f\restr J$ is simultaneously an element of Riesz space $\mathfrak{L^0}_{J}$ and an object of $\mathbf{Meas}^{\rightarrow}$ which recognizes only measurable structure on $\Bbb R$. This is a novel concept of using set-theory alongside category-theory simultaneously for signal representation application.

\begin{equation}
\label{eq:prac_sig_func}
\xymatrix{
	G_1 \ar[r]^{a} & G_2
} 
\quad 
\xrightarrow{F}
\quad
\xymatrix{
	& (\Bbb R,\Sigma_{\Cal B}) \\ 
	(\Bbb R,\Sigma_{\Cal B}) \ar[r]^{h} &   \ar@{->}[u]_{\Delta_J}\\
	(I,\Sigma_{I}) \ar[r]_{\phi} \ar@{->}[u]^{f\restr I = F(G_1)} & (J,\Sigma_{J}) \ar@{->}[u]_{F(a)[F(G_1)]}
}
\end{equation}      

\begin{center}
	\begin{table}[ht]
		\centering
		\resizebox{\linewidth}{!}{
			\begin{tabular}{|c|c|}\hline
				{\bf Classical function based model} & {\bf Proposed functorial model} \\ \hline
				\makecell{Space $\mathfrak{L^0}_{\Bbb R}(\mu)$ canonically isomorphic\\to $\prod_{i\in (I,J,..)}\mathfrak{L^0}_{i}(\mu_i)$ via $f\mapsto(f\restr I,f\restr J,...)$} & 
				\makecell{$f\restr J = Fa(FG_1) = f\restr I$-valued point of $f\restr J + \Delta_J$ \\$(f\restr I,f\restr J,...)$ as family of generalized elements and $\Delta$s}\\ \hline
			\end{tabular}
		}
		\caption{Difference of independent elements versus arrow based generalized elements and differentials.}
		\label{table:diff4_sig_rep}
	\end{table}
\end{center}

\subsection{Signal space as a signal matched category}
\label{sec:sig_matched_cat_1}

Using additional properties of equivalence relation $\eae$, Lebesgue measure $(\mu)$ and $|f|^2$ being integrable, the classic signal $f$ is usually not differentiated from almost equal everywhere measurable functions on $\Bbb R$. Thus $f^{\ssbullet} \in L^0(\Bbb R,\hat\Sigma_{\Cal B},\Cal N(\mu))$\footnote{${L^0}(X,\Sigma_{X},\mu) = {L}^0={L}^0(\mu)={L^0}(X,\Sigma_{X},\Cal N(\mu))$ is set of equivalence classes $f^{\ssbullet}$ in $\mathfrak{L}^0(\mu)$ or $\mathfrak{L}^0_X$  under $\eae$} and using measure along with property $|f|^2$ is integrable, $f^{\ssbullet} \in L^2(\Bbb R,\hat\Sigma_{\Cal B},\mu)$\footnote{${L^p}(X,\Sigma_{X},\mu) = {L}^p={L}^p(\mu)$, $p\in(1,\infty)$ is set of functions $\{f^{\ssbullet}:f\in\mathfrak{L}^p\}\subseteq L^0=L^0(\mu)$ in $\mathfrak{L}^0(\mu)$ under $\eae$}. Thus classic signal $f^{\ssbullet}$ is a vector in global linear spaces $L^0,L^2$. Similarly using properties of equivalence relation $\eae$ and subspace measure on $(I,\Sigma_I)$ we have $(f\restr I)^{\ssbullet} \in L^0(I,\Sigma_I,\Cal N(\mu_I))$, $(f\restr I)^{\ssbullet} \in L^2(I,\Sigma_{I},\mu_I)$ while in certain cases (see Propositions~\ref{prop:L0_func_space},~\ref{prop:L2_func_space},~\ref{prop:lin_bor_fun} for such cases) arrows $(h,\phi),(h',\phi'),...$ uniquely define linear operators $(f\restr I)^{\ssbullet}\mapsto (h(f\restr I)\phi^{-1})^{\ssbullet}$ from $L^2(I,\Sigma_{I},\mu_I)$ to $L^2(J,\Sigma_{J},\mu_J)$. Then in proposed model signal is thought of lying inside a {\bf signal matched category} whose objects are $L^0(I,\Sigma_I,\Cal N(\mu_I)),L^0(J,\Sigma_J,\Cal N(\mu_J))$$,...$ or $L^2(I,\Sigma_{I},\mu_I),L^2(J,\Sigma_{J},\mu_J),...$ and non-trivial morphisms are operators defined by $(h,\phi),(h',\phi')$. Set-theoretically the signal is a family of elements (vectors) comprising one element from each object of the signal matched category. Further these elements are related by arrows which are nothing but restrictions of the operators within category to individual elements leading to representation with generalized elements and differentials.

\begin{center}
	\begin{table}[ht]
		\centering
		\resizebox{\linewidth}{!}{
			\begin{tabular}{|c|c|}\hline
				{\bf Classical function based model} & {\bf Proposed functorial model} \\ \hline
				\makecell{Using $\eae$, measure, $|f|^2$ integrable\\$f^{\ssbullet} \in L^0(\Bbb R,\hat\Sigma_{\Cal B},\Cal N(\mu))$, $f^{\ssbullet} \in L^2(\Bbb R,\hat\Sigma_{\Cal B},\mu)$} & 
				\makecell{Additional properties of null-ideal,measure on $(I,\Sigma_I)$\\permits $(f\restr I)^{\ssbullet} \in L^0(I,\Sigma_I,\Cal N(\mu_I))$, $(f\restr I)^{\ssbullet} \in L^2(I,\Sigma_{I},\mu_I)$}\\ \hline
				\makecell{Signal space is generically fixed \\as $L^0(\Bbb R,\hat\Sigma_{\Cal B},\Cal N(\mu))$ or $L^2(\Bbb R,\hat\Sigma_{\Cal B},\mu)$.} & 
				\makecell{When arrows $(h,\phi),(h',\phi'),...$ uniquely define operators \\ the resulting subcategory in $\mathbf{Hilb}$ or $\mathbf{Riesz}$ is signal matched.} \\ \hline
			\end{tabular}
		}
		\caption{Difference of generic function space versus subcategory in $\mathbf{Hilb}$ or $\mathbf{Riesz}$.}
		\label{table:diff5_sig_rep}
	\end{table}
\end{center}

\subsection{Invoking functorial $L^2$ construction for special class of inverse-measure preserving maps on measure spaces}
\label{sec:sig_matched_cat_2}

In special class of signals where the relationships between local sub-signals are completely captured by the inverse-measure-preserving maps\footnote{A function $\phi:(X,\Sigma_X,\mu)\to (Y,\Sigma_Y,\nu)$ is {\bf inverse-measure-preserving} if $\phi^{-1}[F]\in\Sigma_X$ and $\mu(\phi^{-1}[F])=\nu F$ for every
	$F\in\Sigma_Y$.} denoted by $\phi, \phi',...$, the functorial structure of $L^2$ construction, $L^2:\mathbf{LocMeasure}$\footnote{{$\mathbf{LocMeasure}$} is the category of localizable measure spaces $(X,\Sigma_{X},\mu)$ and $f^\ssbullet$ a.e. equivalence classes of inverse-measure-preserving maps $f:X\to Y$} $\rightarrow \mathbf{Riesz}$\footnote{$\mathbf{Riesz}$ is the category of Riesz spaces and Riesz homomorphisms.} or $L^2:\mathbf{LocMeasure} \rightarrow \mathbf{Hilb}$\footnote{$\mathbf{Hilb}$ is the category of Hilbert spaces and continuous(bounded) linear maps} can be invoked leading to {\bf signal matched category} as $L^2(\mathbf{C})$ where $\mathbf{C} \subseteq \mathbf{LocMeasure}$. This captures generative structure where generators are related by common translations of signal domain. Ofcourse, the classic signal space is fixed global space $L^2(\Bbb R,\hat\Sigma_{\Cal B},\mu)$, a single object in $\mathbf{Riesz}$ or $\mathbf{Hilb}$. See real-world examples in Section~\ref{sec:redundacy_examples}.

\begin{center}
	\begin{table}[ht]
		\centering
		\resizebox{\linewidth}{!}{
			\begin{tabular}{|c|c|}\hline
				{\bf Classical function based model} & {\bf Proposed functorial model} \\ \hline
				\makecell{Signal space fixed $L^2(\Bbb R,\hat\Sigma_{\Cal B},\mu)$\\where $L^2:\mathbf{LocMeasure} \rightarrow \mathbf{Riesz}$} & 
				\makecell{When $\mathbf{C} \subseteq \mathbf{LocMeasure}$ for affine transformations\\of time,then signal matched subcategory is $L^2(\mathbf{C})$} \\ \hline
			\end{tabular}
		}
		\caption{Difference of generic function space versus subcategory in $\mathbf{Hilb}$ or $\mathbf{Riesz}$.}
		\label{table:diff6_sig_rep}
	\end{table}
\end{center}
\subsection{Change of base and Grothendieck's relative point of view}	

By using the opposite category of (complete) Boolean Algebras, it is possible to think of local sub-signals as determined by pullbacks which are unique upto isomorphism, in cases where Boolean homomorphism $\phi^{-1}$ determined by $\phi$ on measurable spaces is an isomorphism. In such cases the pullback object $\Sigma_{\Cal B} \times_{\Sigma_{I}} \Sigma_{J}$ is isomorphic to $\Sigma_{\Cal B}$ via some isomorphism $h^{-1}$ as opposite arrow of some Borel measurable function $h :(\Bbb R,\Sigma_{\Cal B}) \rightarrow (\Bbb R,\Sigma_{\Cal B})$. Thus the opposite arrow $g^{-1}$ of observed measurable function $g$ is pulled back map of $\phi^{-1} \circ g^{-1}$ determined upto isomorphism using $g^{-1} = h^{-1} \circ ((\phi^{-1} \circ g^{-1})*)$. The change of base from $(\Sigma_{I},\symmdiff,\cap)$ to $(\Sigma_{J},\symmdiff,\cap)$ along $\phi^{-1}$ leads to a relative perspective namely that $(\Sigma_{\Cal B},\symmdiff,\cap)$  fibred on $(\Sigma_{I},\symmdiff,\cap)$ via $\phi^{-1} \circ g^{-1}$ produces an object $\Sigma_{\Cal B} \times_{\Sigma_{I}} \Sigma_{J}$ over $(\Sigma_{J},\symmdiff,\cap)$ which is isomorphic to $(\Sigma_{\Cal B},\symmdiff,\cap) $ fibered on $(\Sigma_{J},\symmdiff,\cap)$ via $g^{-1}$.       

\begin{equation}
\label{eq:gro_rel}
\xymatrix{
	(\Bbb R,\Sigma_{\Cal B}) \ar[r]^{id} & (\Bbb R,\Sigma_{\Cal B})  \\
	(I,\Sigma_{I}) \ar[r]_{\phi} \ar@{.>}[ru]^{g\phi} \ar@{->}[u]^{g\phi} & (J,\Sigma_{J}) \ar@{->}[u]_{g}
}
\quad
\xymatrix{
	(\Sigma_{\Cal B},\symmdiff,\cap) \ar@{->}[d]_{\phi^{-1} \circ g^{-1}} & (\Sigma_{\Cal B},\symmdiff,\cap) \ar[l]^{id} \ar@{.>}[ld]^{\phi^{-1} \circ g^{-1}} \ar@{->}[d]^{g^{-1} (E \mapsto g^{-1}[E])} \\
	(\Sigma_{I},\symmdiff,\cap)   & (\Sigma_{J},\symmdiff,\cap) \ar[l]^{\phi^{-1}} 
}
\end{equation}
\begin{equation}
\xymatrix{
	\Sigma_{\Cal B} \ar@/_/[ddr]_{g^{-1}} \ar@/^/[drr]^{id}
	\ar@{<.>}[dr]|-{h^{-1}}            \\
	& \Sigma_{\Cal B} \times_{\Sigma_{I}} \Sigma_{J} \ar[d]^{(\phi^{-1} \circ g^{-1})*} \ar@{<->}[r]^{(\phi^{-1})*}
	& \Sigma_{\Cal B} \ar[d]^{\phi^{-1} \circ g^{-1}}       \\
	& \Sigma_{J} \ar@{<->}[r]_{\phi^{-1}}   & \Sigma_{I} }
\label{eq:pullback}
\end{equation}

\begin{center}
	\begin{table}[ht]
		\centering
		\resizebox{\linewidth}{!}{
			\begin{tabular}{|c|c|}\hline
				{\bf Classical function based model} & {\bf Proposed functorial model} \\ \hline
				\makecell{When $\mathbf{C}= \mathbf{1}$, time axis thought of as single base \\corresponding to the object $(\Bbb R,\hat\Sigma_{\Cal B})$ of $\mathbf{Meas}$.} & 
				\makecell{Multiple $(I,\Sigma_{I})$,$(J,\Sigma_{J})$,... bases and change of base~\cite{BJ}\\in opposite $\mathbf{Meas}^{op}$ leading to celebrated relative viewpoint.} \\ \hline
			\end{tabular}
		}
		\caption{Difference of fixed base versus multiple bases leading to application of Grothendieck's relative point of view in proposed model.}
		\label{table:diff7_sig_rep}
	\end{table}
\end{center}

\subsection{Arrow-theoretic understanding of redundancy and compression}

In classic signal representation since the signal space is a fixed object such as $L^2(\Bbb R,\hat\Sigma_{\Cal B},\mu)$ the compression is mathematically explained via choosing a certain basis (or frame) in this space resulting into sparse representation. However the functorial viewpoint gives the freedom of choosing spaces and their transformations respecting or matched with generative structure. The model as a category naturally leads to interpreting redundancy through arrows across objects or sub-signals. As an illustration, referring Figure~\ref{fig:visual_depiction_signal} we consider the common redundancies purely due to inverse-preserving maps between local intervals. The general redundancy between $f\restr I$ and $f\restr J$ is captured by $F(a)$ and $f\restr J$ is represented as $F(a)f\restr I + \Delta_J$. In translational redundancy this is solely determined by $\phi:(I,\Sigma_{I}) \rightarrow (J,\Sigma_{J})$ which is also additionally measure-preserving by considering additional property of null-ideals and measures. Thus $\phi:(I,\Sigma_{I},\mu_I) \rightarrow (J,\Sigma_{J},\mu_J)$ and using {\bf signal matched category} as $L^2(\mathbf{C})$ where $\mathbf{C} \subseteq \mathbf{LocMeasure}$ as described in Section~\ref{sec:L2_func_space}; we have functorial representation as,  

\begin{equation}
\label{rep}
f^{\ssbullet} =
\underbrace{[(f\restr I)^{\ssbullet}                       \rule[-12pt]{0pt}{5pt}}_{\mbox{classic}}
,\underbrace{\Delta_J                \rule[-12pt]{0pt}{5pt}}_{\mbox{relative error}}
+\underbrace{L^2(\phi^{-1})(f\restr I)^{\ssbullet} \rule[-12pt]{0pt}{5pt}}_{\mbox{generative relative term}}
,...]
\end{equation}
where,
\begin{itemize}
	\item $(f\restr I)^{\ssbullet}$ : local signal Representation using some basis in $L^2(I,\Sigma_{I},\mu_I)$
	
	\item $(f\restr J)^{\ssbullet}$ : local signal Representation using some basis in $L^2(J,\Sigma_{J},\mu_J)$
	
	\item $L^2(\phi^{-1})(f\restr I)^{\ssbullet}$ : Transformed/Transfered local signal from $L^2(I,\Sigma_{I},\mu_I)$ to $L^2(J,\Sigma_{J},\mu_J)$
	
	\item $\Delta_J = (f\restr J)^{\ssbullet} - L^2(\phi^{-1})(f\restr I)^{\ssbullet}$ : Differential or error between transformed and observed local signal in $L^2(J)$.
	
\end{itemize}

By noting that $(\Bbb R,\hat\Sigma_{\Cal B},\mu)$ is the direct sum of $\langle(i,\Sigma_i,\mu_i)\rangle_{i\in (I,J,...)}$ a
family of measure spaces, the canonical isomorphism between $L^0(\Bbb R,\hat\Sigma_{\Cal B},\Cal N(\mu))$ and $\prod_{i\in (I,J,...)}L^0(\mu_i)$ induces an isomorphism between $L^2(\Bbb R,\hat\Sigma_{\Cal B},\mu)$ and the subspace
$\{u:u\in\prod_{i\in (I,J,...)}L^2(\mu_i),\,\|u\|=$ $\bigl(\sum_{i\in I}\|u(i)\|_2^2)^{1/2}<\infty\}$ of $\prod_{i\in (I,J,...)}L^2(\mu_i)$. Now the classical signal $f^{\ssbullet} \in {L^2}(\mu)$ belongs to a global space ${L^2}(\mu)$ canonically isomorphic to subspace of $\prod_{i\in (I,J,..)}{L^2}(\mu_i)$ via $f^{\ssbullet}\mapsto[(f\restr I)^{\ssbullet},(f\restr J)^{\ssbullet},...]$. In classical representation the components $(f\restr I)^{\ssbullet},(f\restr J)^{\ssbullet},...$ are always linearly independent and since all these belong to a global space the entire signal is expressed as single element $(f\restr I)^{\ssbullet}+(f\restr J)^{\ssbullet}+...$ where vector addition is in global space ${L^2}(\Bbb R,\hat\Sigma_{\Cal B},\mu)$ and components such as $(f\restr I)^{\ssbullet}$ are automatically thought of having zero value outside restricted interval on whole $(\Bbb R,\hat\Sigma_{\Cal B},\mu)$. However the same components in proposed representation belong to different sub-spaces $L^2(\mu_i)$ which are treated as separate objects, related by special arrows in category $\mathbf{Riesz}$ or $\mathbf{Hilb}$. These arrows make the components dependent or related while differentials together with arrows are the real innovations which are often encoded. When the arrows $L^2(\phi^{-1})$ are compactly represented using $\phi^{-1}$ then $L^2$ being fixed construction overall gain in encoding results since one encodes differential and $\phi^{-1}$ to generate $\Delta_J +L^2(\phi^{-1})(f\restr I)^{\ssbullet}$ instead of encoding $(f\restr J)^{\ssbullet}$ independently as in classic case. In various lossless compression standards such as PNG or DPCM in lossless JPEG the map $\phi$ is fixed specified by the filter type. This provides a category-theoretic mathematical explanation for achieving enhanced compression in differential encoding such as PNG or DPCM in lossless JPEG as compared to standard lossy JPEG using wavelets or discrete cosine transforms especially in certain class of images such as line drawings, text, and iconic graphics involving frequent isomorphisms between neighboring pixels as illustrated using examples in~\ref{sec:redundacy_examples}
\begin{center}
	\begin{table}[ht]
		\centering
		\resizebox{\linewidth}{!}{
			\begin{tabular}{|c|c|}\hline
				{\bf Classical function based model} & {\bf Proposed functorial model} \\ \hline
				\makecell{Classic signal $f^{\ssbullet} \in L^2(\Bbb R,\Sigma_{Leb},\mu) = {L^2}(\mu)$\\$f^{\ssbullet}\mapsto[(f\restr I)^{\ssbullet},(f\restr J)^{\ssbullet},...]$} & 
				\makecell{When $\mathbf{C} \subseteq \mathbf{LocMeasure}$, signal $f^{\ssbullet}$ = $[(f\restr I)^{\ssbullet},(f\restr J)^{\ssbullet},...]$\\$(f\restr J)^{\ssbullet}$ represented as $\Delta_J +L^2(\phi^{-1})(f\restr I)^{\ssbullet}$} \\ \hline
			\end{tabular}
		}
		\caption{Arrow-theoretic Redundancy, compression and relative information content of Signal.}
		\label{table:diff8_sig_rep}
	\end{table}
\end{center}

\subsection{Equivalent formulation using $\mathbf{Meas}^{\mathbf{2}}$}
\label{meas2}
So far we have discussed the functorial framework using a basic functor however in this section we formulate an equivalent viewpoint using the higher categorical concept of functor category $\mathbf{B}^{\mathbf{C}}$ referring~\cite{CWM}. This viewpoint has the advantage of not having to deal with abstract base category as this is often not given in the application but practically we have only concrete observed signals. Since generators (capturing the intuition of generative theory in~\cite{Leyton01}) of a signal generate concrete measurable waveforms therefore, rather than working with abstract category $\mathbf{C}$ we directly model them as functors $G_1,G_2,G_3,...:\mathbf{2}$\footnote{$\mathbf{2}$ is the category with two objects $O, M$ with identity morphisms and one non-trivial morphism $l:O \to M$.}$ \rightarrow \mathbf{Meas}$. The transfers (or isomorphisms) between the generators then automatically become natural transformations (or natural isomorphisms) $a_1,a_2,...$ forming some signal matched subcategory (or groupoid) of the usual functor category~\cite{CWM} namely $\mathbf{Meas}^{\mathbf{2}}$. Thus corresponding to 
objects $(f\restr I),(f\restr K),(f\restr J),...$ which are local real-valued partial functions on disjoint half-open intervals $I$, $J$, $K$, .... we associate generators as functors while transfers between generators take the form of natural transformations. From the signal representation point of view both these ways are equivalent and a reader can more or less substitute one for other using the fact that functor category $\mathbf{Meas}^{2}$ is same as $\mathbf{Meas}^{\rightarrow}$, using Equation~\ref{eq:sig_func_cat_1} along with arguments presented in earlier subsections of this section. Whenever the category $\mathbf{C}$ is not known apriori or abstract, in such cases the generators and their relationships can be concretely modeled as functors $G_1:\mathbf{2} \to \mathbf{D}$ and natural transformations $a:G_1 \to G_2$.

\begin{equation}
\label{eq:sig_func_cat}
\xymatrixcolsep{1.5cm}
\xymatrix{
	\mathbf{2} \ruppertwocell^{R}{\alpha}
	\rlowertwocell_{T}{\beta}
	\ar[r]|{S}
	&\mathbf{Meas}\\
}
\quad
\xymatrix{
	Rc_2=(\Bbb R,\Sigma_{\Cal B}) \ar@/^2.0pc/@{<->}[rr]^{(\beta \cdot \alpha)_{c_2}} \ar@{<->}[r]^{\alpha_{c_2}= h_1} 	& Sc_2=(\Bbb R,\Sigma_{\Cal B}) \ar@{<->}[r]^{\beta_{c_2}=h_2}	 & Tc_2=(\Bbb R,\Sigma_{\Cal B}) \\
	Rc_1=(I,\Sigma_{I}) \ar@/_2.0pc/@{<->}[rr]_{(\beta \cdot \alpha)_{c_1}}	\ar@{<->}[r]_{\alpha_{c_1}=\phi_1}	\ar@{->}[u]^{Rf=f_1} & Sc_1=(J,\Sigma_{J}) \ar@{<->}[r]_{\beta_{c_1}=\phi_2}	\ar@{->}[u]^{Sf=f_2}	 & Tc_1=(K,\Sigma_{K}) \ar@{->}[u]^{Tf=f_3}	\\
}
\end{equation}

The generators (capturing the intuition of generative theory in~\cite{Leyton01}) of a signal generate measurable waveforms therefore can be modeled as functors $R,S,T,...$. The transfers (or isomorphisms) between the generators then automatically become natural transformations (or natural isomorphisms) $\alpha,\beta,...$ forming some signal matched subcategory of the usual functor category $\mathbf{Meas}^{\mathbf{2}}$ in~\cite{CWM}.

\subsection{Dealing partial isomorphisms, practical observation errors and general limitations.}
\label{limitation}

In this work, the category for natural generative structure (or equivalently the subcategory of functor category with functors $\mathbf{2} \rightarrow \mathbf{Meas})$ is taken as a groupoid since isomorphism most appropriately models the transfer of previously occurring object in the generative theory of of~\cite{Leyton01}. Moreover the codomain of an isomorphism can be completely determined from its domain and serves to model redundancy of structures and signals in particular. Thus from a given observed signal thought of as a functor (or equivalently as a coproduct of measurable functions on disjoint half-open intervals), the goal is to reversely determine its generative groupoid structure. This functor model is generally faithful for groupoids whose objects generate waveforms with disjoint supports. In such cases from the image subcategory we can infer the isomorphic domain category. However when the objects in generative category generate waveforms with overlapping supports then faithfulness breaks down as from the category perspective superposition operation is not a coproduct operation. This introduces limitations on the generative groupoid which can be inferred from observed signal as described next,

\begin{enumerate}
	\item {\bf Non-faithfulness of functor}: In case of generators producing waveforms or measurable functions whose domain intervals are not naturally disjoint, the functor becomes non-faithful or in other words some structure is lost and therefore all isomorphisms are not preserved. This occurs since observed signal is a superposition of these waveforms through additive field property on $\Bbb R$ which is not a coproduct operation in $\mathbf{Meas}^{\rightarrow}$ or a coproduct in functor category $\mathbf{Meas}^{2}$. Common examples include combination of melodies in music signal or two different speakers uttering words simultaneously in speech signal. 
	
	\item {\bf Structures recognized in a category}: The constraints imposed by choice of codomain category which determine the type of isomorphisms that can be modeled. In the case of $\mathbf{Meas}^{\rightarrow}$ using isomorphisms of type $(h,\phi)$ only measurable structures on the generators can be modeled through isomorphisms which includes translations, scaling, amplitude changes but not any other structure such as topology of images etc.
	
	\item {\bf Practical limitations of generating systems}: Due to practical limitations of equipments it is quite possible that perfectly isomorphic generators may produce slightly dissimilar waveforms. While if there is a partial isomorphism between natural generators of the signal or in other words the generative category is a partial groupoids then certainly we have correspondingly dissimilar waveforms.
\end{enumerate}

These limitations of the functorial model can be dealt with in different ways as we discuss now. Note since direct sum is the coproduct in $\mathbf{Meas}$ we have for disjoint half-open intervals $I$ and $J$ the coproduct of two measurable functions $f_1$ and $f_2$ $\left(I \xrightarrow{f_1} \Bbb R\right) \bigoplus \left(J \xrightarrow{f_2} \Bbb R \right) \cong \left( I \bigoplus J \xrightarrow{f_1 \bigoplus f_2} \Bbb R \bigoplus \Bbb R \right)$ 

In the first limitation, if $f_1$, $f_2$ are two measurable functions (corresponding to two generators) the domains of which intersect each other then using Definition 34 and Theorem 49(b) we have three disjoint measurable spaces $dom f_1 \setminus dom f_2$, $dom f_2 \setminus dom f_1$, $dom f_1 \cap dom f_2$. Correspondingly we have three functions $f_1 \restr (dom f_1 \setminus dom f_2)$, $f_2 \restr (dom f_2 \setminus dom f_1)$, $(f_1 + f_2) \restr (dom f_1 \cap dom f_2)$ whose coproduct (direct sum) $f_1 \restr (dom f_1 \setminus dom f_2) \bigoplus (f_2 \restr (dom f_2 \setminus dom f_1)\bigoplus(f_1 + f_2) \restr (dom f_1 \cap dom f_2)$ is the observed signal to be represented. Thus the two generators get treated as being three corresponding to the three measurable functions on disjoint intervals and can be recovered or filtered only in special cases as discussed in third limitation. Extending to overlapping of finite $n$ generators we have $f_1$, $f_2$,...,$f_n$ as measurable functions of which we can similarly form coproduct of disjoint measurable functions.

The second limitation is dealt with by varying the codomain category $\mathbf{D}$ to categories of sheaves, graphs, topological measurable spaces etc. This can be corresponding to structures such as topology, graphs depending on specific class of signals we need to represent where a particular structure is prevalent.

The third limitation is on account of the arrows which recognize only measurable structure on $\Bbb R$ but don't respect the field structure thereby restricting faithful functorial modeling of partial groupoid as the generative category. Workaround for this limitation is using set-theoretic measure theory and considering the usual linear structure of the signal spaces such as $\mathfrak{L^0}_{X}$, ${L^0}_{X}$ or ${L^2}_{X}$ which take into account the field structure on $\Bbb R$. Referring Equation~\ref{eq:prac_sig_func}, if the isomorphism is partial which means $G_2$ is not totally isomorphic then differential $\Delta_J = F(G_2) - F(a)F(G_1)$ in Riesz space $\mathfrak{L^0}_{J}$ serves to indicate linear deviation of $G_2$ from $G_1$. In case of $a$ being total isomorphism, $\Delta_J$ also directly indicates the difference between theoretical local signal $F(a)F(G_1)$ from practical observed signal $F(G_2)$. Only when the generating and sensing equipments are ideal we will have theoretically the signal $F(a)F(G_1)$ (corresponding to object $(a)(G_1)$) equal to the observed signal $F(G_2)$ (corresponding to object $G_2$).          

\begin{equation}
\label{eq:prac_sig_func}
\xymatrix{
	G_1 \ar[r]^{a} & G_2
} 
\quad \\
\xrightarrow{F}
\quad
\xymatrix{
	& (\Bbb R,\Sigma_{\Cal B}) \\ 
	(\Bbb R,\Sigma_{\Cal B}) \ar[r]^{h} &   \ar@{->}[u]_{\Delta_J}\\
	(I,\Sigma_{I}) \ar[r]_{\phi} \ar@{->}[u]^{f\restr I = F(G_1)} & (J,\Sigma_{J}) \ar@{->}[u]_{F(a)[F(G_1)]}
}
\end{equation}

When $\Delta_J$ is relatively very small as compared to $F(a)F(G_1)$ it can be attributed to practical limitations of equipments and we can infer total isomorphism between $G_1$ and $G_2$. In case of partial isomorphism between $G_1$ and $G_2$, $\Delta_J$ is quite large. In this case we can continue to theoretically treat arrow $a$ in the base as total isomorphism and then operate on $\Delta_J$ as new local signal corresponding to subobject of $G_2$ apart from remaining subobject $(a)(G_1)$ isomorphic to $G_1$ as illustrated in the next example. 

\subsection{A prototypical example beneath the design of Portable network graphics.}\label{eg:2}
\pgfplotsset{
	standard/.style={
		axis x line=middle,
		axis y line=middle,
		enlarge x limits=0.05,
		enlarge y limits=0.05,
		every axis x label/.style={at={(current axis.right of origin)},anchor=north west},
		every axis y label/.style={at={(current axis.above origin)},anchor=north east},
		every axis plot post/.style={mark options={fill=white}}
	}
}

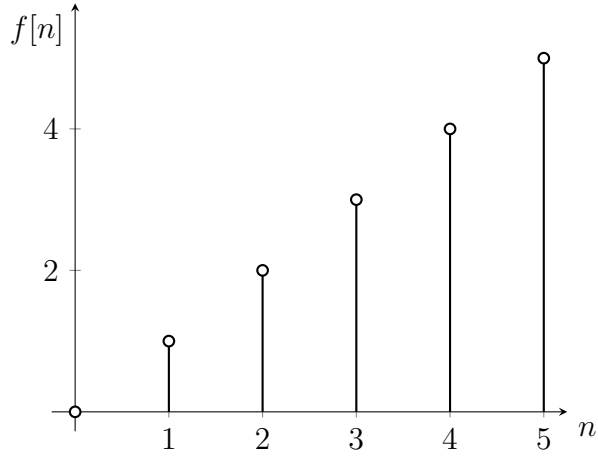
\begin{figure}[!h]
	\centering
	\begin{tikzpicture}
	\begin{axis}[%
	standard,
	domain = 0:5,
	samples = 6,
	xlabel={$n$},
	ylabel={$f[n]$},
	ymin=0,
	ymax=5.5]
	\addplot+[ycomb,black,thick] table {
		0 0
		1 1
		2 2
		3 3
		4 4
		5 5
		6 6
	};
	\end{axis}
	\end{tikzpicture}
	\caption{Sequence $f[n] = [1,2,3,4,5,...]$: A prototype differential coding example underlying PNG.}
	\label{fig:1}
\end{figure}
As shown in Figure~\ref{fig:1}, consider a signal as sequence of real values $1,2,3,....$ terminating at some finite $n$. This well-known sequence is encoded using differential coding by subtracting from each value the previous occurring value to obtain highly compressible sequence $1,1,1,...$. This example can be recast in a functorial way making use of the partial groupoid structure in base category. The example is specifically used to illustrate the limitation where functor is non-faithful and therefore we invoke set-theoretic measure theory to provide a work-around.

\begin{figure}[!h]
	\centering
	$\xymatrixcolsep{0.8cm}
	\xymatrix{
		&  & (\Bbb R,\Sigma_{\Cal B}) \\
		& (\Bbb R,\Sigma_{\Cal B}) \ar[r]^{id_{(\Bbb R,\Sigma_{\Cal B})}} & \ar@{->}[u]_{\Delta_K=f[3]-(id,\phi_2)f[2]}\\ 
		(\Bbb R,\Sigma_{\Cal B}) \ar[r]^{id_{(\Bbb R,\Sigma_{\Cal B})}} &   \ar@{->}[u]_{\Delta_J=f[2]-(id,\phi_1)f[1]} & \\
		(\{1\},\{\emptyset,\{1\}\}) \ar[r]_{\phi_1} \ar@{->}[u]^{f[1]} & (\{2\},\{\emptyset,\{2\}\}) \ar@{->}[u]_{(id,\phi_1)f[1]}
		\ar[r]_{\phi_2} & (\{3\},\{\emptyset,\{3\}\}) \ar@{->}[uu]_{(id,\phi_2)f[2]}\\
		G_1 \ar[r]^{a_1} & G_2' \amalg G_2'' \ar[r]^{a_2} & G_3' \amalg G_3'' \amalg G_3''' 
	}$
	\caption{The functorial signal representation model for $f[n]$ (Identities and composite arrows are not shown).}
	\label{fig:2}
\end{figure}
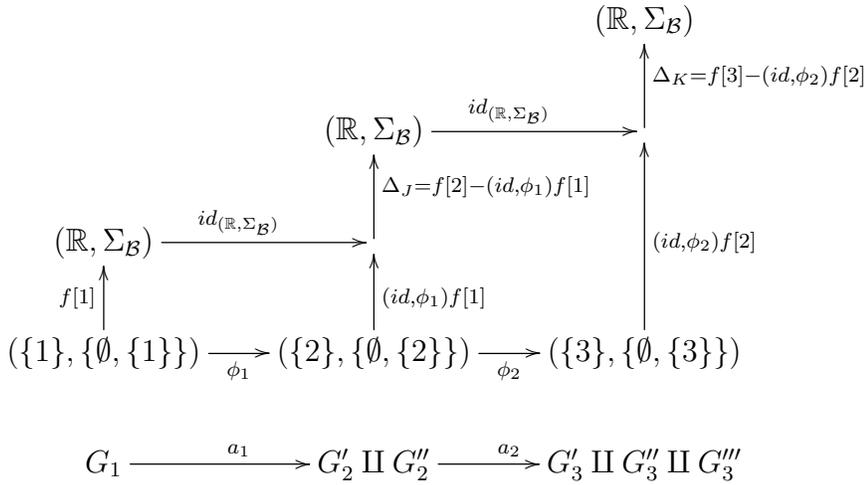

We utilize a subsequence of just three samples $f[1],f[2],f[3]$ to illustrate the limitations and workarounds as discussed in Section~\ref{limitation}. Consider a subgroupoid within the category $\mathbf{Set}$ with six singletons as objects $G_1, G_2', G_2'', G_3', G_3'', G_3'''$ and isomorphisms $a_1$, $a_2$, $a_2'$, $A_1$, $A_2$. Note that we have not considered identities and compositions in this groupoid explicitly as they are not relevant to limitations we are discussing; however they are present implicitly. Now define a functor $F$ which takes these objects and arrows to measurable functions as illustrated in Table~\ref{table:diff_coding_1}. Now since the generators generate measurable functions whose supports are not completely disjoint the observed signal is sequence $f[1] = g_1, f[2]= g_2 + \Delta_J, f[3]= g_3 + g_3' + \Delta_K$ where the field addition is not a coproduct operation in $\mathbf{Meas}^{\rightarrow}$. This forces us to reduce the base category to discrete groupoid with objects $G_1, G_2' \amalg G_2'', G_3' \amalg G_3''\amalg G_3'''$ corresponding to measurable functions $f[1],f[2],f[3]$ if we want functor to be isomorphism-preserving. Fortunately using Equation~\ref{eq:prac_sig_func} as discussed in workaround for limitations we can bring the set-theory once again to undo the field addition using inverses or subtraction operation. Utilizing the fact that $G_2' \amalg G_2''$ is partially isomorphic to $G_1$ we can represent this isomorphic subobject as $F(a_1)F(G_1)$ and recover the remaining subobject as $\Delta_J = F(G_2' \amalg G_2'') - F(a_1)F(G_1)$ as shown in Figure~\ref{fig:2}.    

\begin{equation}
\label{eq:par_groupoid_base}
\xymatrix{
	&	 & G_3''' \\ 
	& G_2'' \ar@{<.>}[ru]^{A_2} \ar@{<->}[r]^{a_2'} & G_3'' \ar@{<.>}[u] \\
	G_1 \ar@{<.>}[ru]^{A_1} \ar@{<->}[r]^{a_1} & G_2' \ar@{<.>}[u] \ar@{<->}[r]^{a_2} & G_3' \ar@{<.>}[u]
} 
\quad
\xymatrix{
	&	 & (\Bbb R,\Sigma_{\Cal B}) \\
	&	 & (K,\Sigma_{K}) \ar@{->}[u]_{\Delta_K} \\
	& (\Bbb R,\Sigma_{\Cal B})  \ar@{<->}[r]^{id} \ar@{<.>}[ruu]	 & (\Bbb R,\Sigma_{\Cal B}) \\
	& (J,\Sigma_{J}) \ar@{<->}[r]_{\phi_2}	\ar@{->}[u]^{\Delta_J} \ar@{<.>}[ruu]	 & (K,\Sigma_{K}) \ar@{->}[u]^{g_3'} \\
	(\Bbb R,\Sigma_{\Cal B}) \ar@{<->}[r]^{id} \ar@{<.>}[ruu]	& (\Bbb R,\Sigma_{\Cal B}) \ar@{<->}[r]^{id}	 & (\Bbb R,\Sigma_{\Cal B}) \\
	(I,\Sigma_{I})	\ar@{<->}[r]_{\phi_1}	\ar@{->}[u]^{g_1} \ar@{<.>}[ruu]	& (J,\Sigma_{J}) \ar@{<->}[r]_{\phi_2}	\ar@{->}[u]^{g_2}	 & (K,\Sigma_{K}) \ar@{->}[u]^{g_3}
}
\end{equation}

\begin{center}
	\begin{table}[ht]
		\centering
		\resizebox{\linewidth}{!}{
			\begin{tabular}{|c|c|}\hline
				{\bf Functorial framework parameters} & {\bf Values for differential coding example} \\ \hline
				\makecell{Base Category $\mathbf{C}$ Objects} & 
				\makecell{$G_1$, $G_2'$, $G_2''$, $G_3'$, $G_3''$, $G_3'''$\\ $G_2''$, $G_3'''$ model $\Delta_J$, $\Delta_K$ since example is highly structured.} \\ \hline
				\makecell{Base Category $\mathbf{C}$ Arrows} & 
				\makecell{$a_1$, $a_2$, $a_2'$, $A_1$, $A_2$\\ $A_1$, $A_2$ model $\Delta_J$, $\Delta_K$ inter-relationship in this example.} \\ \hline
				\makecell{Image subcategory $F(\mathbf{C})$ Objects} & 
				\makecell{$F(G_1) = g_1: (I,\Sigma_{I}) \to (\Bbb R,\Sigma_{\Cal B})$ where $(I,\Sigma_{I})=(\{1\},\{\emptyset,\{1\}\})$\\ $F(G_2') = g_2: (J,\Sigma_{J}) \to (\Bbb R,\Sigma_{\Cal B})$ where $(J,\Sigma_{J})=(\{2\},\{\emptyset,\{2\}\})$\\$F(G_3') = g_3: (K,\Sigma_{K}) \to (\Bbb R,\Sigma_{\Cal B})$ where $(K,\Sigma_{K})=(\{3\},\{\emptyset,\{3\}\})$\\$g_1$, $g_2$, $g_3$ are all measurable functions each with value $1 \in \Bbb R$\\$F(G_2'') = \Delta_J: (J,\Sigma_{J}) \to (\Bbb R,\Sigma_{\Cal B})$ where $(J,\Sigma_{J})=(\{2\},\{\emptyset,\{2\}\})$\\ $F(G_3'') = g_3': (K,\Sigma_{K}) \to (\Bbb R,\Sigma_{\Cal B})$ where $(K,\Sigma_{K})=(\{3\},\{\emptyset,\{3\}\})$\\$F(G_3''') = \Delta_K: (K,\Sigma_{K}) \to (\Bbb R,\Sigma_{\Cal B})$ where $(K,\Sigma_{K})=(\{3\},\{\emptyset,\{3\}\})$\\$\Delta_J$, $g_3'$, $\Delta_K$ are all measurable functions each with value $1 \in \Bbb R$ .} \\ \hline
				\makecell{Image subcategory $F(\mathbf{C})$ Arrows} & 
				\makecell{$F(a_1) = (id,\phi_1): g_1 \to g_2$ such that $id \circ g_1 = g_2 \circ \phi_1$ or $g_2 = g_1 \circ \phi_1^{-1}$\\ $F(a_2) = (id,\phi_2): g_2 \to g_3$ such that $id \circ g_2 = g_3 \circ \phi_2$ or $g_3 = g_2 \circ \phi_2^{-1}$\\$F(a_2') = (id,\phi_2): \Delta_J \to g_3'$ where $id \circ \Delta_J = g_3' \circ \phi_2$ or $g_3' = \Delta_J \circ \phi_2^{-1}$\\$\phi_1$, $\phi_2$ are all measurable functions each being isomorphism.\\$F(A_1) = (id,\phi_1): g_1 \to \Delta_J$ such that $id \circ g_1 = \Delta_J \circ \phi_1$ or $\Delta_J = g_1 \circ \phi_1^{-1}$\\ $F(A_2) = (id,\phi_2): \Delta_J \to \Delta_K$ such that $id \circ \Delta_J = \Delta_K \circ \phi_2$ or $\Delta_K = \Delta_J \circ \phi_2^{-1}$\\$\Delta_J$, $\Delta_K$ only in this example are identical with $g_2$ $g_3'$ due to special structure.} \\ \hline
			\end{tabular}
		}
		\caption{The differential coding prototype example underlying the Image enconding standards such as PNG.}
		\label{table:diff_coding_1}
	\end{table}
\end{center}

Thus in case of generators such as $G_1$ and $G_2 = G_2' \amalg G_2''$ being partially isomorphic either they are treated as having no interconnected isomorphism or as being connected through a (partial) isomorphism such as $a: G_1 \to G_2$ as shown in Figure~\ref{fig:3}. In that case we represent the measurable function $F(G_2) = F(a_1)F(G_1) + \Delta_J$ interpreted as generalized element with differential. Such a model using partial isomorphism $a_1$ can be utilized especially when the sequence is expected to be highly structured in the sense that differentials themselves are also related to each other as in this prototype example. 

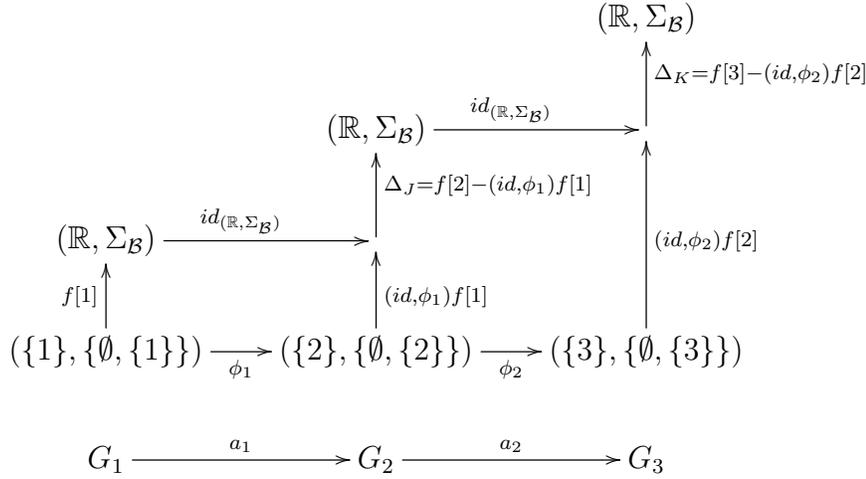
\begin{figure}[!h]
	\centering
	$\xymatrixcolsep{0.8cm}
	\xymatrix{
		&  & (\Bbb R,\Sigma_{\Cal B}) \\
		& (\Bbb R,\Sigma_{\Cal B}) \ar[r]^{id_{(\Bbb R,\Sigma_{\Cal B})}} & \ar@{->}[u]_{\Delta_K=f[3]-(id,\phi_2)f[2]}\\ 
		(\Bbb R,\Sigma_{\Cal B}) \ar[r]^{id_{(\Bbb R,\Sigma_{\Cal B})}} &   \ar@{->}[u]_{\Delta_J=f[2]-(id,\phi_1)f[1]} & \\
		(\{1\},\{\emptyset,\{1\}\}) \ar[r]_{\phi_1} \ar@{->}[u]^{f[1]} & (\{2\},\{\emptyset,\{2\}\}) \ar@{->}[u]_{(id,\phi_1)f[1]}
		\ar[r]_{\phi_2} & (\{3\},\{\emptyset,\{3\}\}) \ar@{->}[uu]_{(id,\phi_2)f[2]}\\
		G_1 \ar[r]^{a_1} & G_2 \ar[r]^{a_2} & G_3 
	}$
	\caption{The functorial signal representation model for $f[n]$ (Identities and composite arrows are not shown).}
	\label{fig:3}
\end{figure}

since the relationships are completely determined by fixed maps $\phi_1,\phi_2,...$ which are also inverse-measure-preserving isomorphisms by considering counting measure on all measurable spaces which are finite sets with a singleton. Hence functorial signal representation is directly given by Equation~\ref{rep} which becomes,

\begin{equation}
\label{rep_prototype}
Signal= [f[1],\Delta_J + l^2(\phi_1^{-1})f[1],\Delta_K + l^2(\phi_2^{-1})[\Delta_J + l^2(\phi_1^{-1})f[1]],...]
\end{equation}
since in this case,
\begin{itemize}
	\item $(f\restr I)^{\ssbullet}=f[1]$ : local signal using basis $\hat{e}_I$ in $L^2(I,\Sigma_{I},\mu_I)=l^2(\{1\},\{\emptyset,\{1\}\},count)$
	
	\item $(f\restr J)^{\ssbullet}=f[2]$ : local signal using basis $\hat{e}_J$ in $L^2(J,\Sigma_{J},\mu_J)=l^2(\{2\},\{\emptyset,\{2\}\},count)$
	
	\item $(f\restr K)^{\ssbullet}=f[3]$ : local signal using basis $\hat{e}_K$ in $L^2(K,\Sigma_{K},\mu_K)=l^2(\{3\},\{\emptyset,\{3\}\},count)$
	
	\item $l^2(\phi_1^{-1})f[1]$ : Transformed local signal from $L^2(I,\Sigma_{I},\mu_I)$ to $L^2(J,\Sigma_{J},\mu_J)$
	
	\item $\Delta_J = f[2] - l^2(\phi_1^{-1})f[1]$ : Error between transformed and observed local signal in $L^2(J)$.
	
	\item $\Delta_K = f[3] - l^2(\phi_1^{-1})f[2]$ : Error between transformed and observed local signal in $L^2(K)$.
	
\end{itemize}

The functorial framework mathematically makes precise the underlying intuition behind working of the classic filtering phase of all major image encoding standards utilizing differential coding techniques. Especially if the generative category is such that overall most of the generators are isomorphic to their preceding generators then this type of differential filtering will result into almost sparse sequence the reason for which was illustrated through the prototype example. For practical examples see Section~\ref{red_exam}.

\subsection{Differences between classic versus functorial signal space structures}
In this section, we discuss differences between classic versus functorial signal space structures. This includes standard properties such as linearity, lattice and multiplication of classical signal spaces  $L^0(\Bbb R,\Sigma_{Leb},\Cal N(\mu)$ (or $\mathfrak{L}^0_{X}$) along with additional norm and inner-product of $L^2(\Bbb R,\Sigma_{Leb},\mu)$. In general all these properties in subcategory of $\mathbf{Riesz}$ or $\mathbf{BanLatt}$ are simply local or valid on objects within this category. The structure-preserving morphisms in these subcategories are linear operators preserving these properties considered as structures on objects.  

\subsection{Linear structure}
\label{sec:linear}
For the spaces given by Definitions~\ref{def:virtually_measurable_fun_space},~\ref{def:measurable_fun_space},~\ref{def:L0_space},\ref{def:Lp_space}, we have the following properties, 

\begin{enumerate}
	\item Let $f$, $f'$, $g$, $g' \in \mathfrak{L}^0$,
	$f \eae f'$ and $g \eae g'$ then $f + g \eae f' + g'$. This defines
	addition on $\mathfrak{L}^0_{X}$ a subspace with $f$, $f'$, $g$, $g'$
	defined on complete $X$. Also it defines addition on $L^0(\mu)$ by setting
	$f^{\ssbullet}+g^{\ssbullet}=(f+g)^{\ssbullet}$ for all $f$, $g \in \mathfrak{L}^0$.
	Let $f$, $g \in  \mathfrak{L}^2$. If
	$c$, $c' \in \Bbb R$ then $|c + c'|^2 \le 2^2\max(|c|^2,|c'|^2)$, therefore $|f+g|^2 \leae 2^2(|f|^2 \vee |g|^2)$.
	But $|f+g|^2 \in \mathfrak{L}^0$ and
	$2^2(|f|^2 \vee |g|^2)$ is integrable so $|f+g|^2$ is integrable. Hence
	$f+g \in \mathfrak{L}^2$ for all $f$, $g\in \mathfrak{L}^2$; therefore
	$f^{\ssbullet}+g^{\ssbullet}=(f+g)^{\ssbullet}$ for all $f$, $g \in \mathfrak{L}^2$
	defining an addition on $L^2(\mu)$.
	
	\item If $f$, $g \in \mathfrak{L}^0$ and $f \eae g$,
	then $cf \eae cg $ for every $c\in\Bbb R$. This defines
	scalar multiplication on $\mathfrak{L}^0_{X}$ with $f$, $g$
	defined on complete $X$. And it defines scalar multiplication on $L^0(\mu)$ by setting
	$c \cdot f^{\ssbullet} = (cf)^{\ssbullet}$ for
	all $f\in \mathfrak{L}^0$ and every $c\in\Bbb R$. Now let $f\in \mathfrak{L}^2$ and $c\in\Bbb R$ then
	$|cf|^2 = |c|^2 |f|^2$ is integrable, so $cf \in \mathfrak{L}^2$.
	Therefore $c f^{\ssbullet} \in L^2(\mu)$ whenever $f^{\ssbullet} \in L^2(\mu)$ and $c \in\Bbb R$.
	
\end{enumerate} 

Thus $\mathfrak{L}^0_{X}$, $L^0(\mu)$ are linear spaces over $\Bbb R$, with zero $\mathbf{0}$ (the function with domain $X$ and
constant value $0$) and $\mathbf{0}^{\ssbullet}$ respectively. The negatives are given by $-(f)$ and $-(f^{\ssbullet})=(-f)^{\ssbullet}$ respectively. Also $L^2(\mu)$ is a linear subspace of $L^0(\mu)$. Thus all axioms of a linear space are satisfied namely,
\begin{itemize}
	\item {$f+(g+h)=(f+g)+h$ for all $f$, $g$, $h \in \mathfrak{L}^0$,} hence {$f'+(g'+h')=(f'+g')+h'$ for all $f'$, $g'$, $h' \in \mathfrak{L}^0_{X}$.} and {$f^{\ssbullet}+(g^{\ssbullet}+h^{\ssbullet})=(f^{\ssbullet}+g^{\ssbullet})+h^{\ssbullet}$ for all $f^{\ssbullet}$, $g^{\ssbullet}$, $h^{\ssbullet} \in L^0(\mu)$.}
	
	\item {$f+\mathbf{0}=\mathbf{0}+f=f$ for every $f\in \mathfrak{L}^0$,} hence {$f'+\mathbf{0}=\mathbf{0}+f'=f'$ for every $f'\in \mathfrak{L}^0_{X}$,} and {$f^{\ssbullet}+\mathbf{0}^{\ssbullet}=\mathbf{0}^{\ssbullet}+f^{\ssbullet}=f^{\ssbullet}$ for every $f^{\ssbullet} \in L^0(\mu)$.}

	\item {$f+ (-f) \eae \mathbf{0}$ for every $f\in \mathfrak{L}^0$,} hence {$f'+ (-f') \eae \mathbf{0}$ for every $f'\in \mathfrak{L}^0_{X}$,} and {$f^{\ssbullet}+(-f)^{\ssbullet}=\mathbf{0}^{\ssbullet}$ for every $f^{\ssbullet} \in L^0(\mu)$.}
	
	\item {$f + g=g + f$ for all $f$, $g\in  \mathfrak{L}^0$,} hence {$f'+g'=g'+f'$ for all $f'$, $g'\in \mathfrak{L}^0_{X}$,} and {$f^{\ssbullet} + g^{\ssbullet}=g^{\ssbullet} + f^{\ssbullet}$ for all $f^{\ssbullet}$,
		$g^{\ssbullet} \in L^0(\mu)$.}
	
	\item {$c(f + g)=cf + cg$ for all $f$, $g\in \mathfrak{L}^0$ and $c\in\Bbb R$,} hence {$c(f'+g')=cf'+cg'$ for all $f'$, $g'\in\mathfrak{L}^0_{X}$ and $c\in\Bbb R$,} and {$c(f^{\ssbullet} + g^{\ssbullet})=cf^{\ssbullet} + cg^{\ssbullet}$ for all $f^{\ssbullet}$, $g^{\ssbullet}\in L^0(\mu)$ and $c\in\Bbb R$.}
	
	\item {$(c+c')f=cf+c'f$ for all $f\in \mathfrak{L}^0$ and $c$,$c'\in\Bbb R$,} hence {$(c+c')g=cf+c'g$ for all $g\in \mathfrak{L}^0_{X}$ and $c$,$c'\in\Bbb R$,} and {$(c+c')f^{\ssbullet}=cf^{\ssbullet}+c'f^{\ssbullet}$ for all $f\in L^0(\mu)$ and $c$, $c'\in\Bbb R$.}
	
	\item {$(cc')f=c(c'f)$ for all $f\in \mathfrak{L}^0$ and $c$, $c'\in\Bbb R$,} hence {$(cc')g=c(c'g)$ for all $g\in \mathfrak{L}^0_{X}$ and $c$, $c'\in\Bbb R$,} and {$(cc')f^{\ssbullet}=c(c'f^{\ssbullet})$ for all $f^{\ssbullet}\in L^0(\mu)$ and $c$, $c'\in\Bbb R$.}
	
	\item {$1f = f$ for all $f\in \mathfrak{L}^0$,} hence {$1f'=f'$ for all $f'\in \mathfrak{L}^0_{X}$,} and {$1f^{\ssbullet}=f^{\ssbullet}$ for all $f^{\ssbullet}\in L^0(\mu)$. }
\end{itemize}

\subsection{Partial order and lattice structure}
\label{sec:partialorder}
For the spaces given by Definitions~\ref{def:virtually_measurable_fun_space},~\ref{def:measurable_fun_space},~\ref{def:L0_space},\ref{def:Lp_space}, we have properties related to partial order and lattice in addition to linearity.

\begin{enumerate}
	\item Let $f$, $f'$, $g$, $g' \in \mathfrak{L}^0$, $f \eae f'$,
	$g \eae g'$ and $f \leae g $, then $f' \leae g'$. This defines a relation $\le$ on
	$\mathfrak{L}^0_{X}$ by declaring $f_1 \le g_1$ iff $f \leae g$ where $f \eae f_1$
	and $g \eae g_1$ for $f_1$, $g_1 \in \mathfrak{L}^0_{X}$ which simply means $f_1 \leae g_1$ for $f_1$, $g_1$ defined on whole domain ${X}$. Also it defines a relation $\le$ on $L^0(\mu)$ by declaring that
	$f^{\ssbullet}\le g^{\ssbullet}$ iff $f\leae g$. For $p\in[1,\infty]$, $f$, $g \in \mathfrak{L}^p$ means $f$, $g \in \mathfrak{L}^0$ such that $|f|^p,|g|^p$ are integrable thus $\leae$ inherited from $\mathfrak{L}^0$  
	defines a relation $\le$ on $L^p(\mu)$ by saying that $f^{\ssbullet}\le g^{\ssbullet}$ iff $f\leae g$ for $f$, $g \in \mathfrak{L}^p$. Alternately we can inherit this relation directly from $L^0(\mu)$ by noting that $L^p(\mu)$ is its linear subspace. This of-course holds for special case $p=2$.
	
	\item Next $\le$ is a partial order on $\mathfrak{L}^0_{X}$,$L^0(\mu)$ and $L^p(\mu)$. (i) Let $f$, $g$, $h\in\mathfrak{L}^0$, if $f\leae g$ and $g\leae h$, then $f\leae h$. Therefore $f_1\le h_1$ whenever $f_1$, $g_1$,
	$h_1 \in \mathfrak{L}^0_{X}$, $f_1\le g_1$ and $g_1\le h_1$. Similarly $f^{\ssbullet}\le h^{\ssbullet}$ whenever $f^{\ssbullet}$, $g^{\ssbullet}$,
	$h^{\ssbullet}\in L^0(\mu)$, $f^{\ssbullet}\le g^{\ssbullet}$ and $g^{\ssbullet}\le h^{\ssbullet}$.   (ii) Let $f\in \mathfrak{L}^0$ and
	$f\leae f$;  then $f_1 \le f_1$, $f^{\ssbullet} \le f^{\ssbullet}$ for every $f^{\ssbullet}\in L^0(\mu)$.   (iii) Let $f$,
	$g\in \mathfrak{L}^0$, $f \leae g$ and $g \leae f$, then $f \eae g$, therefore if $f_1\le g_1$ and $g_1\le f_1$ then $f_1=g_1$ in $\in \mathfrak{L}^0_{X}$, similarly $f^{\ssbullet}\le g^{\ssbullet}$ and $g^{\ssbullet}\le f^{\ssbullet}$ then $f^{\ssbullet}=g^{\ssbullet}$ in $L^0(\mu)$. For linear subspace $L^p(\mu)$, $\le$ is a partial order on inherited from $L^0(\mu)$ and holds for $L^2(\mu)$.
	
	\item  Adding linear properties of $\mathfrak{L}^0_{X}$, $L^0(\mu)$ and $L^p(\mu)$, with $\le$, make these {\bf partially
		ordered linear spaces}, which are real linear spaces with a partial order $\le$ such that for elements $x,y,z$ of these we have (i) if $x \le y$ then $x+z\le y+z$ for every $z$, since if $f$, $g$, $h\in \mathfrak{L}^0$ and
	$f \leae g$, then $f+h \leae g+h$. (ii) if $0\le x$ then $0\le cx$ for every $c\ge 0$, since if $f\in \mathfrak{L}^0$ and
	$f \ge 0$ a.e., then $cf \ge 0$ a.e.\ for every $c \ge 0$.
	
	\item Next $\mathfrak{L}^0_{X}$, $L^0(\mu)$ and $L^p(\mu)$ are {\bf Riesz spaces} or {\bf
		vector lattices}, which are partially ordered linear spaces such that
	$x \vee y=\sup\{x,y\}$ and $x\wedge y=\inf\{x,y\}$ are defined for all
	$x$, $y \in $ $\mathfrak{L}^0_{X}$, $L^0(\mu)$ or $L^p(\mu)$
	Let $f$, $g\in \mathfrak{L}^0$ such that $f^{\ssbullet}=x$
	and $g^{\ssbullet}=y$.  Then $f\vee g$, $f\wedge g \in\mathfrak{L}^0$, hence $f_1\vee g_1$, for $f_1\wedge g_1 \in\mathfrak{L}^0_{X}$. Now expressing {$(f \vee g)(x)=\max(f(x),g(x))$, \quad $(f \wedge
		g)(x)=\min(f(x),g(x))$} for $x\in (\dom f\cap\dom g)$. But, for any $h\in  \mathfrak{L}^0$, we have
	
	\centerline{$f \vee g \leae h\iff f\leae h$ and $g\leae h$,}
	
	\centerline{$h\leae f \wedge g\iff h\leae f$ and $h\leae g$,}
	
	Hence for any $z \in L^0(\mu)$ we have
	
	\centerline{$(f\vee g)^{\ssbullet}\le z\iff x \le z$ and $y \le z$,}
	
	\centerline{$z \le(f\wedge g)^{\ssbullet} \iff z \le x$ and $z \le y$.}
	
	Indeed $(f\vee g)^{\ssbullet}=\sup\{x,y\}=x \vee y$, \quad $(f \wedge g)^{\ssbullet}=\inf\{x,y\}=x \wedge y$ in $L^0(\mu)$. Again this lattice structure in linear subspace $L^p(\mu)$ is inherited from $L^0(\mu)$ and present in $L^2(\mu)$.
	Finally there are some additional properties of Archimedeaness and Dedekind completeness that can be found in~\cite{fremlinmt2}.
	
\end{enumerate} 

\subsection{Multiplicative structure}
\label{sec:multiplicative}
For the spaces given by Definitions~\ref{def:virtually_measurable_fun_space},~\ref{def:measurable_fun_space},~\ref{def:L0_space},\ref{def:Lp_space}, we have properties related to multiplicative structure in addition to being Riesz spaces. Let $f$, $f'$, $g$, $g' \in \mathfrak{L}^0$, $f \eae f'$ and $g \eae g'$ then $f \times g \eae f' \times g'$. Thus multiplication in $\mathfrak{L}^0_{X}$ is the usual multiplication from $\mathfrak{L}^0$, while the multiplication on $L^0(\mu)$ is defined by setting $f^{\ssbullet} \times g^{\ssbullet}=(f\times g)^{\ssbullet}$ for all $f$, $g\in \mathfrak{L}^0$. This also valid for the subspace $L^p(\mu)$. For all $x$, $y$, $z \in L^0(\mu)$ and $c \in \Bbb R$, the following properties can be easily verified,

\begin{itemize}
	\item $x \times (y \times z) =(x \times y)\times z$,
	
	\item $x \times \mathbf{1}^{\ssbullet}= \mathbf{1}^{\ssbullet} \times x =x$, where $\mathbf{1}$ is the
	equivalence class of the function with constant value $1$,
	
	\item $c(x\times y)=cx\times y=x\times cy$,
	
	\item $x\times(y+z)=(x\times y)+(x\times z)$,
	
	\item $(x+y)\times z=(x\times z)+(y\times z)$,
	
	\item $x\times y=y\times x$,
	
	\item $|x\times y|=|x|\times |y|$,
	
	\item $x\times y=0$ iff  $|x|\wedge|y|=0$,
	
	\item $|x|\le|y|$ iff there is a $z$ such that
	$|x|\le \mathbf{1}^{\ssbullet}$ and $x=y\times z$.
\end{itemize}
\section{Redundancy and Examples}
\label{red_exam}
In this Section, we contribute a novel category-theoretic Definition~\ref{def:redundancy} of intra-signal redundancy and some of its special cases~\ref{def:trans_redundancy},~\ref{def:aff_redundancy},~\ref{def:aff_amp_redundancy} using the isomorphism arrow in a category. Next we illustrate taking examples of some real-world images how compression occurs. The well-known heuristic yielding better compression of PNG standard in image types such as line drawings, iconic image, text etc. compared to classic JPEG standard at a given SNR can be precisely explained using category-theory. Finally we comment on base-structured category perspective of signals and their spaces leading to possible unification of various existing techniques of representing signals.    

\begin{figure}
	\resizebox*{\textwidth}{!}{
		\begin{tikzpicture}
		\draw[domain=-1.57:1.57,samples=100] plot(\x,{sin(2*\x r)});
		\draw[domain=7.85:11,samples=100] plot(\x,{sin(2*\x r)});
		\draw plot [smooth] coordinates {(1.57,0) (3,-0.5) (4,0.75) (5,1) (6,-0.4) (7,1) (7.85,0)};
		\draw[->] (-4,0) -- (13,0) node [pos=1,below] {$t$};
		\draw plot [smooth] coordinates {(-1.57,0) (-2.5,1) (-2.8,0) (-3.5,-0.5) (-3.7,-1) (-4,0)};
		\draw plot [smooth] coordinates {(11,0) (11.5,-0.5) (12,1) (12.5,0.5)};
		\draw[->] (-3.5,-1) --(-3.5,1.5) node [pos=1,left] {$\mathbb{R}$};
		\draw[-,dashed] (-1.57,1.5) -- (-1.57,-1.5);
		\draw[-,dashed] (1.57,1.5) -- (1.57,-1.5);
		
		\draw[-,dashed] (7.85,1.5) -- (7.85,-1.5);
		\draw[-,dashed] (11,1.5) -- (11,-1.5);
		\draw[<->] (-1.57,-1.3) --(1.57,-1.3) node[midway,fill=white] {$I$};
		\draw[<->] (7.85,-1.3) --(11,-1.3) node[midway,fill=white] {$J$};
		\node at (0,1.3) [draw=none] {$f$};
		\node at (9.4,1.3) [draw=none] {$f'$};
		\node at (-3.5,0) [draw=none,above left] {$0$};
		\end{tikzpicture}
	}
	\caption{A global signal where the local signal $f=F(G_1)=G_1l \in L^2(I,\Sigma_{I},\mu)$ in half-open interval $I$ is translated to $f'=F(G_2)=G_2l \in L^2(J,\Sigma_{J},\nu)$ in half-open interval $J$}
	\label{fig:redundancy}
\end{figure}
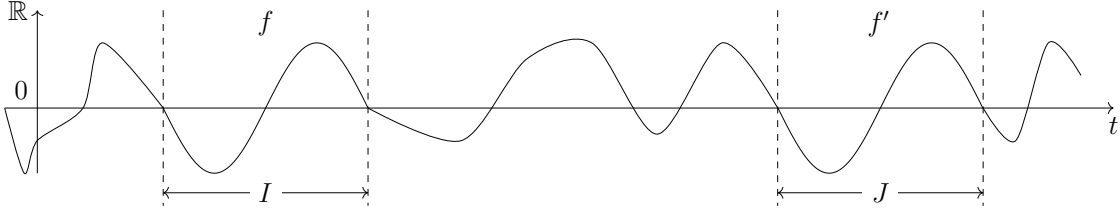 

\begin{equation}
\label{eq:sig_func_red}
\xymatrix{
	G_1 \ar[r]^{a} & G_2
} 
\quad \\
\xrightarrow{F}
\quad
\xymatrix{
	(\Bbb R,\Sigma_{\Cal B}) \ar[r]^{h} & (\Bbb R,\Sigma_{\Cal B})  \\
	(I,\Sigma_{I}) \ar[r]_{\phi} \ar@{->}[u]^{f\restr I = F(G_1)} & (J,\Sigma_{J}) \ar@{->}[u]_{F(G_2) \cong F(a)[F(G_1)]}
}
\end{equation}
\begin{equation}
\label{eq:sig_func_cat_1_red}
\xymatrixcolsep{1.5cm}
\xymatrix{
	\mathbf{2}
	\ar@/^2pc/[rr]_{\quad}^{G_1}="1"
	\ar@/_2pc/[rr]_{G_2}="2"
	&& \mathbf{Meas}
	\ar@{}"1";"2"|(.2){\,}="7"
	\ar@{}"1";"2"|(.8){\,}="8"
	\ar@{=>}"7" ;"8"^{a} } \quad 
\quad
\xymatrix{
	G_1M=(\Bbb R,\Sigma_{\Cal B}) \ar@{<->}[r]^{a_{M}= h} 	& Sc_2=(\Bbb R,\Sigma_{\Cal B}) \\
	G_1O=(I,\Sigma_{I})	\ar@{<->}[r]_{a_{O}=\phi}	\ar@{->}[u]^{G_1l=f\restr I} & Sc_1=(J,\Sigma_{J}) 	\ar@{->}[u]^{G_2l = f\restr J} \\
}
\end{equation}

Using Figure~\ref{fig:redundancy} and Equations~\ref{eq:sig_func_cat_1_red},~\ref{eq:sig_func_red} we give the following definitions;

\begin{definition}
	\label{def:redundancy}
	Consider a complete signal modeled as some subcategory of a functor category $\mathbf{Meas}^{2}$ with functors $G_1, G_2,...$ and natural transformations $a_1, a_2,...$, where $\mathbf{2}$ is the category with two objects $O, M$ with identity morphisms and one non-trivial morphism $l:O \to M$ and $\mathbf{Meas}$ is the category of measurable spaces and measurable functions. A local signal $G_2l$ corresponding to the generator $G_2: \mathbf{2} \rightarrow \mathbf{Meas}$ is defined as {\bf redundant} relative to a local signal $G_1l$ corresponding to the generator $G_1: \mathbf{2} \rightarrow \mathbf{Meas}$ iff there exists a isomorphism $(h,\phi): G_1l \rightarrow G_2l$ between them, resulting from the natural isomorphism $a: G_1 \rightarrow G_2$ between its generators. 
	
	Alternatively, given a signal modeled as some faithful isomorphism-preserving functor $F:\mathbf{C} \to \mathbf{Meas}^{\rightarrow}$ with generators $G_1, G_2,...$ and transfers $a_1, a_2,...$ as objects and arrows of $\mathbf{C}$, a local signal $F(G_2)$ is defined as {\bf redundant} relative to local signal $F(G_1)$ iff there exists an isomorphism $F(a): F(G_1) \leftrightarrow F(G_2)$ which the image of isomorphism $a: G_1 \rightarrow G_2$ in $\mathbf{C}$. 
\end{definition}

\subsection{Special cases of redundancies: Translational, affine, affine amplified/attenuated}

The specific cases of isomorphisms $(h,\phi)$ lead to three special type of redundancies common in signal theory. 

Let $I = [a,b)$ and $J = [a+T,b+T)$ be its translation on real line by $T \in \Bbb R$. More precisely this means there exists a map $\phi : I \rightarrow J$ given by $\phi(i) = i + T, i \in I$. Now $\phi$ is both one-to-one and onto or bijection of sets, hence its inverse is given by $\phi^{-1}(j) = j - T, j \in J$. Moreover by using the result that if set $A$ is measurable then for any $T \in \Bbb R$, $A + T = \{a + T; a \in A\}$ is measurable, we conclude that both $\phi$ and $\phi^{-1}$ are measurable. Also by using Lebesgue subspace measures $\mu_I, \mu_J$ on $I$ and $J$, $\phi$, $\phi^{-1}$ are both inverse-measure-preserving since translation leaves Lebesgue measure invariant. 

\begin{definition}
	\label{def:trans_redundancy}
	Consider a complete signal modeled as some subcategory of a functor category $\mathbf{Meas}^{2}$ with functors $G_1, G_2,...$ and natural transformations $a_1, a_2,...$, where $\mathbf{2}$ is the category with two objects $O, M$ with identity morphisms and one non-trivial morphism $l:O \to M$ and $\mathbf{Meas}$ is the category of measurable spaces and measurable functions. A local signal $G_2l$ corresponding to the generator $G_2: \mathbf{2} \rightarrow \mathbf{Meas}$ is defined as {\bf translation redundant} relative to a local signal $G_1l$ corresponding to the generator $G_1: \mathbf{2} \rightarrow \mathbf{Meas}$ iff there exists a natural isomorphism $(id_{\Bbb R},\phi): G_1l \rightarrow G_2l$ between them, where the measurable map $\phi : (I,\Sigma_I) \rightarrow (J,\Sigma_J)$ given by $\phi(i) = i+T, \forall i \in I$ defines the translation of domain of $G_1l$ signal to the domain of $G_2l$ signal.
	
	Alternatively, given  a signal as faithful functor $F:\mathbf{C} \to \mathbf{Meas}^{\rightarrow}$, a local signal $F(G_2)$ is defined as {\bf translation redundant} relative to local signal $F(G_1)$ iff there exists an isomorphism $F(a)= (id_{\Bbb R},\phi): F(G_1) \rightarrow F(G_2)$ between them, where the measurable map $\phi : (I,\Sigma_I) \rightarrow (J,\Sigma_J)$ given by $\phi(i) = i+T, \forall i \in I$ defines the translation of $F(G_1)$ signal domain to the domain of $F(G_2)$.
\end{definition}

Let $I = [a,b)$ and $J = [Sa,Sb)$ be its scaling on real line by $S \in \Bbb R$. More precisely this means there exists a map $\phi : I \rightarrow J$ given by $\phi(i) = Si, i \in I$. Again $\phi$ is both one-to-one and onto or bijection of sets, hence its inverse is given by $\phi^{-1}(j) = (1/S)j, j \in J$. Moreover by using the result that if set $A$ is measurable then for any $S \in \Bbb R$, $SA = \{Sa; a \in A\}$ is measurable, we conclude that both $\phi$ and $\phi^{-1}$ are measurable. Also by using Lebesgue subspace measure $\mu_I$ on $I$ and if $\mu_J$ is usual Lebesgue subspace measure on $J$, then setting $\nu_J = (1/S)\mu_J$ we conclude $\phi$, $\phi^{-1}$ are both inverse-measure-preserving.

\begin{definition}
	\label{def:aff_redundancy}
	Consider a complete signal modeled as some subcategory of a functor category $\mathbf{Meas}^{2}$ with functors $G_1, G_2,...$ and natural transformations $a_1, a_2,...$, where $\mathbf{2}$ is the category with two objects $O, M$ with identity morphisms and one non-trivial morphism $l:O \to M$ and $\mathbf{Meas}$ is the category of measurable spaces and measurable functions. A local signal $G_2l$ corresponding to the generator $G_2: \mathbf{2} \rightarrow \mathbf{Meas}$ is defined as {\bf affine redundant} relative to a local signal $G_1l$ corresponding to the generator $G_1: \mathbf{2} \rightarrow \mathbf{Meas}$ iff there exists a natural isomorphism $(id_{\Bbb R},\phi): G_1l \rightarrow G_2l$ between them, where the measurable map $\phi : (I,\Sigma_I) \rightarrow (J,\Sigma_J)$ given by $\phi(i) = Si+T, \forall i \in I$ models the affine transformation of domain of $G_1l$ signal into the domain of $G_2l$ signal.
	
	Alternatively, given a signal as faithful functor $F:\mathbf{C} \to \mathbf{Meas}^{\rightarrow}$, a local signal $F(G_2)$ is defined as {\bf affine redundant} relative to local signal $F(G_1)$ iff there exists an isomorphism $F(a)= (id_{\Bbb R},\phi): F(G_1) \rightarrow F(G_2)$ between them, where the measurable map $\phi : (I,\Sigma_I) \rightarrow (J,\Sigma_J)$ given by $\phi(i) = Si+T, \forall i \in I$ models affine transformation of the domain of $F(G_1)$ signal into the domain of $F(G_2)$ signal.
\end{definition}

Finally additionally using a linear Borel measurable function $h: (\Bbb R,\Sigma_{\Cal B}) \rightarrow (\Bbb R,\Sigma_{\Cal B})$. in addition to measurable function $\phi$ as above we can model amplified or attenuated redundant signal.

\begin{definition}
	\label{def:aff_amp_redundancy}
	Consider a complete signal modeled as some subcategory of a functor category $\mathbf{Meas}^{2}$ with functors $G_1, G_2,...$ and natural transformations $a_1, a_2,...$, where $\mathbf{2}$ is the category with two objects $O, M$ with identity morphisms and one non-trivial morphism $l:O \to M$ and $\mathbf{Meas}$ is the category of measurable spaces and measurable functions. A local signal $G_2l$ corresponding to the generator $G_2: \mathbf{2} \rightarrow \mathbf{Meas}$ is defined as {\bf affine amplified/attenuated redundant} relative to a local signal $G_1l$ corresponding to the generator $G_1: \mathbf{2} \rightarrow \mathbf{Meas}$ iff there exists a natural isomorphism $(h,\phi): G_1l \rightarrow G_2l$ between them, where the measurable map $\phi : (I,\Sigma_I) \rightarrow (J,\Sigma_J)$ given by $\phi(i) = Si+T, \forall i \in I$ and $h: (\Bbb R,\Sigma_{\Cal B}) \rightarrow (\Bbb R,\Sigma_{\Cal B})$ being a linear Borel measurable function both together model the affine transformation of domain of $G_1l$ signal into the domain of $G_2l$ signal along with amplification/attenuation of $G_1l$ into $G_2l$. 
	
	Alternatively, given a signal as faithful functor $F:\mathbf{C} \to \mathbf{Meas}^{\rightarrow}$, a local signal $F(G_2)$ is defined as {\bf affine amplified/attenuated redundant} relative to local signal $F(G_1)$ iff there exists an isomorphism $F(a)= (h,\phi): F(G_1) \rightarrow F(G_2)$ between them, where the measurable map $\phi : (I,\Sigma_I) \rightarrow (J,\Sigma_J)$ given by $\phi(i) = Si+T, \forall i \in I$ and $h: (\Bbb R,\Sigma_{\Cal B}) \rightarrow (\Bbb R,\Sigma_{\Cal B})$ being a linear Borel measurable function both together model the affine transformation of domain of $F(G_1)$ signal into the domain of $F(G_2)$ signal along with amplification/attenuation of $F(G_1)$ into $F(G_2)$.
\end{definition}

\subsection{Examples: Groupoid base structures in signals. Category-theoretic explanation for superior compression of differential encoding standards.}
\label{sec:redundacy_examples}

In this section, first we demonstrate the inherent natural generative groupoid structures in real-world image and audio signals. Figure~\ref{fig:5} is a sample iconic image of an umbrella while Figure~\ref{fig:7} is another sample image of a real cameraman. By assigning an underlying generator corresponding to each pixel value one can immediately notice that each image is generated by maximal transfer (isomorphism) of existing generators in the language of~\cite{Leyton01}. This is often the case with line drawings, iconic image, text etc. Also in some real-world photographic images there is a lot of groupoid structure resulting from maximization of transfer such as in the case of background sky, the grass and coat of cameraman. New images as in Figures~\ref{fig:6},~\ref{fig:9}, are 3D surface plots while Fig~\ref{fig:9} is 2D plot of images produced by subtracting from each pixel value the previous occurring neighboring pixel value often termed as differential coding to obtain highly compressible or sparse images as observed. These signals are representations of original signals by using a signal-space as the category $L^2(\mathbf{C})$ as where $\mathbf{C} \subseteq \mathbf{LocMeasure}$ as discussed in Sections~\ref{sec:sig_matched_cat_1},~\ref{sec:sig_matched_cat_2} and Equation~\ref{rep_prototype}. Notice that only where isomorphisms of the underlying generators break down or in other words generators not related to each other, non-sparse signal values $\Delta_I, \Delta_J, \Delta_K, ....$ are observed such as near the boundary of umbrella etc. Thus using category-theory we can explain the known heuristic yielding better compression of PNG standard in image types such as line drawings, iconic image, text etc. compared to classic JPEG standard at a given SNR.

\begin{figure}[!t]
	\centering
	\includegraphics[width=1.0in]{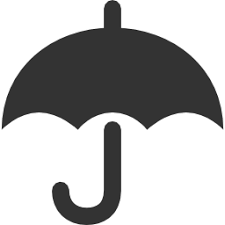}
	\caption{A simple iconic image of umbrella~\cite{umbrella}}
	\label{fig:5}
\end{figure}

\begin{figure}[!t]
	\centering
	\resizebox*{\textwidth}{!}{
		\includegraphics[width=5in]{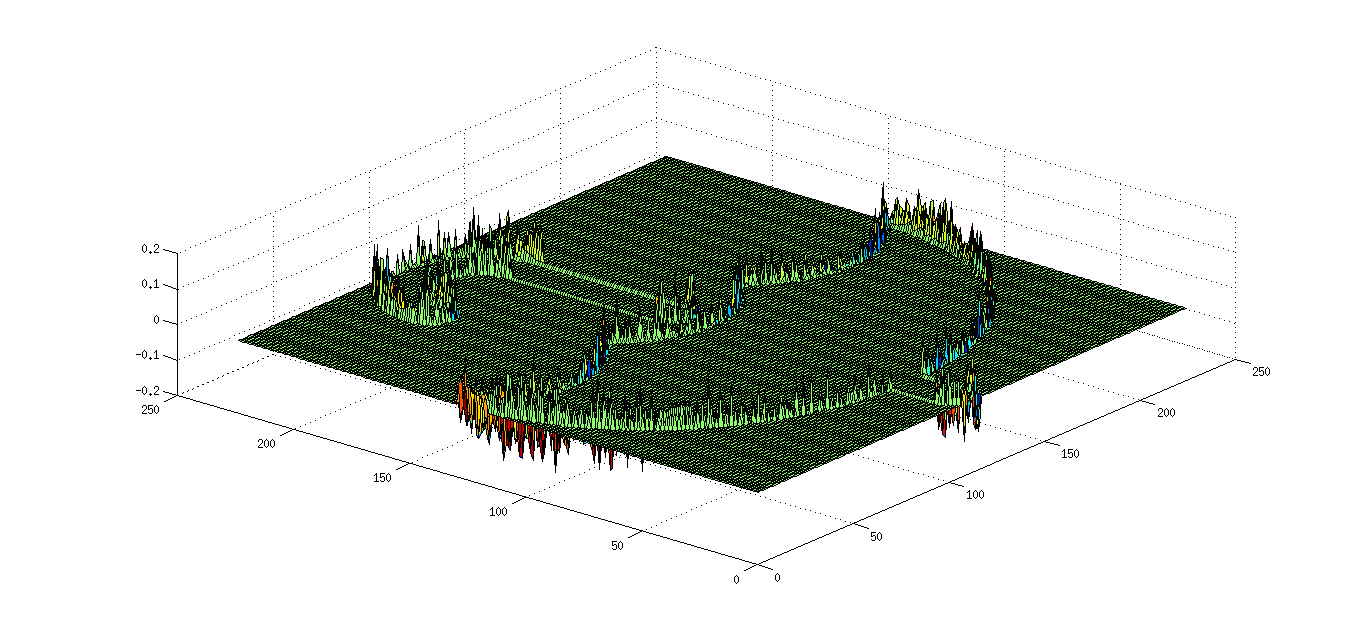}
	}
	\caption{Surface 3D plot of filtered umbrella image}
	\label{fig:6}
\end{figure}

\begin{figure}[!t]
	\centering
	\includegraphics[width=3in]{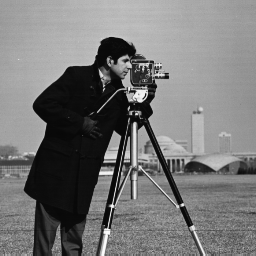}
	\caption{Sample photographic cameraman image~\cite{cameraman}}
	\label{fig:7}
\end{figure}

\begin{figure}[!t]
	\centering
	\includegraphics[width=6in]{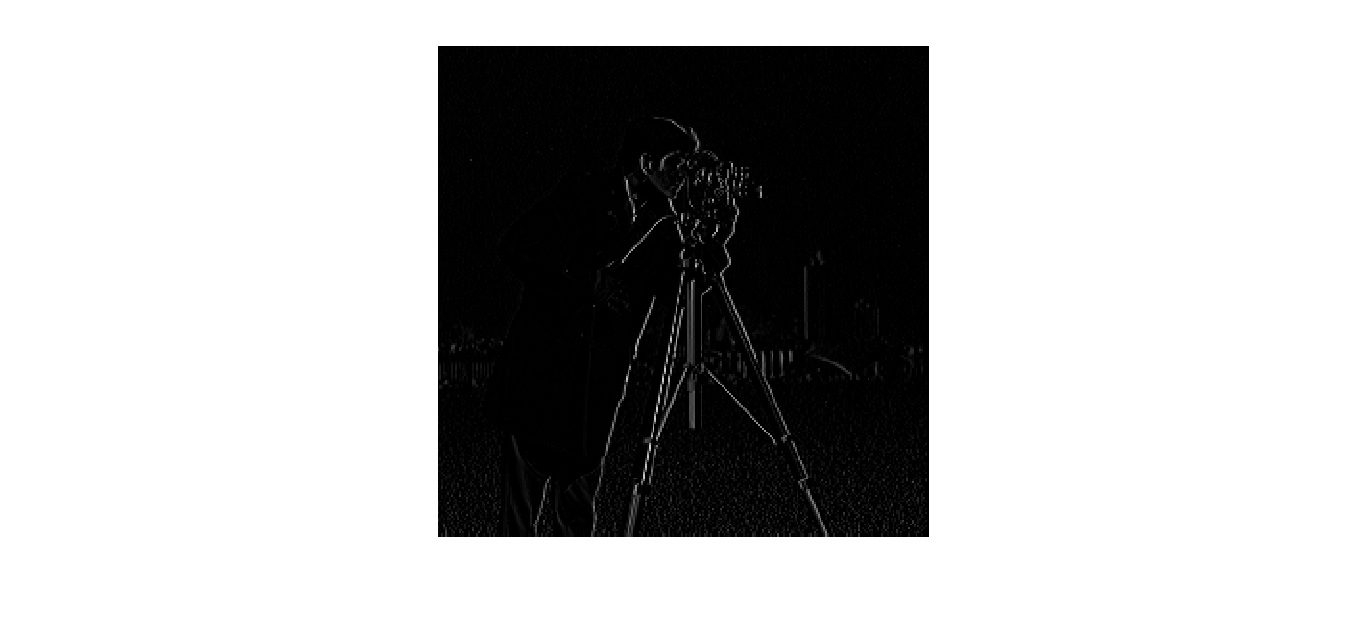}
	\caption{2D plot of filtered cameraman}
	\label{fig:8}
\end{figure}

\begin{figure}[!t]
	\centering
	\resizebox*{\textwidth}{!}{
		\includegraphics[width=5in]{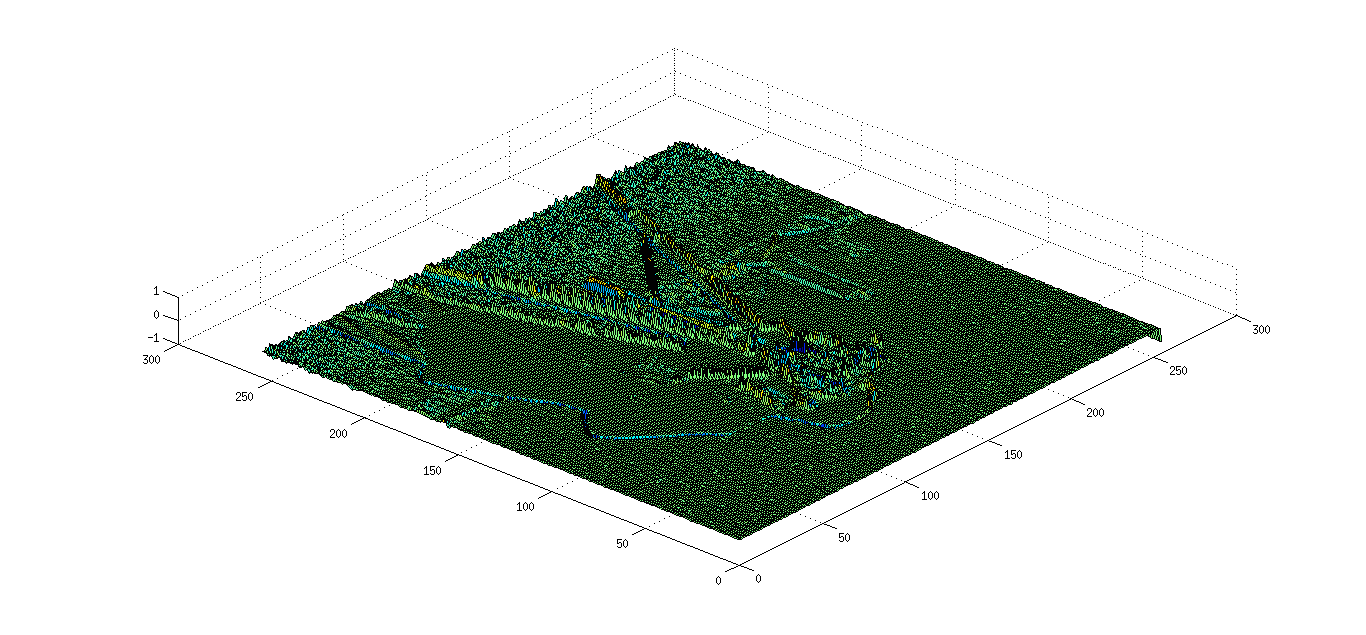}
	}
	\caption{Surface 3D plot of filtered cameraman}
	\label{fig:9}
\end{figure}

However such kind of transfers are also inherently present in other kinds of signals such as music etc since the underlying melodies are maximally transfered by artists as argued in~\cite{Leyton01}. We illustrate these groupoid structures in a real-world sample audio BBC countdown compilation signal as shown in Figure~\ref{fig:10}. The three marked rectangles illustrate the translation transfer of particular melody twice in a window of 30 seconds as shown superimposed on each-other in Figure~\ref{fig:11}  

\begin{figure}[!t]
	\centering
	\resizebox*{\textwidth}{!}{
		\includegraphics{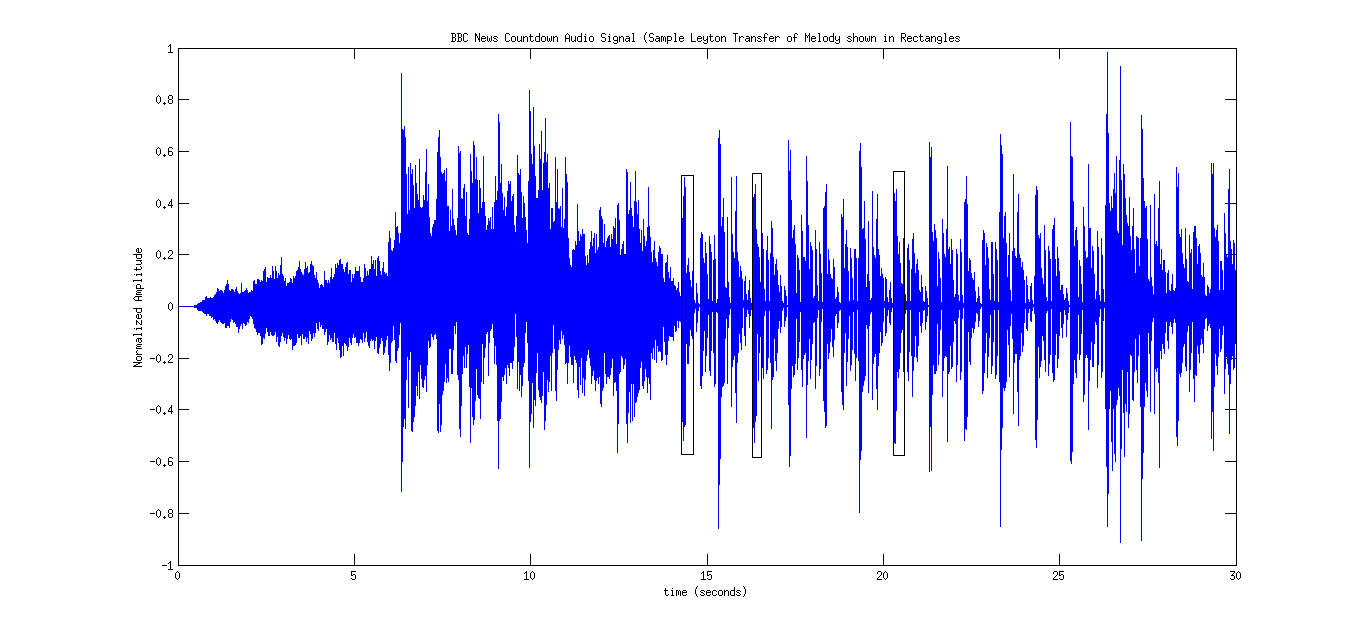}
	}
	\caption{An audio signal of BBC News countdown with sample translation transfer of local melodies in boxed rectangles~\cite{bbc}}
	\label{fig:10}
\end{figure}

\begin{figure}[!t]
	\centering
	\resizebox*{\textwidth}{!}{
		\includegraphics{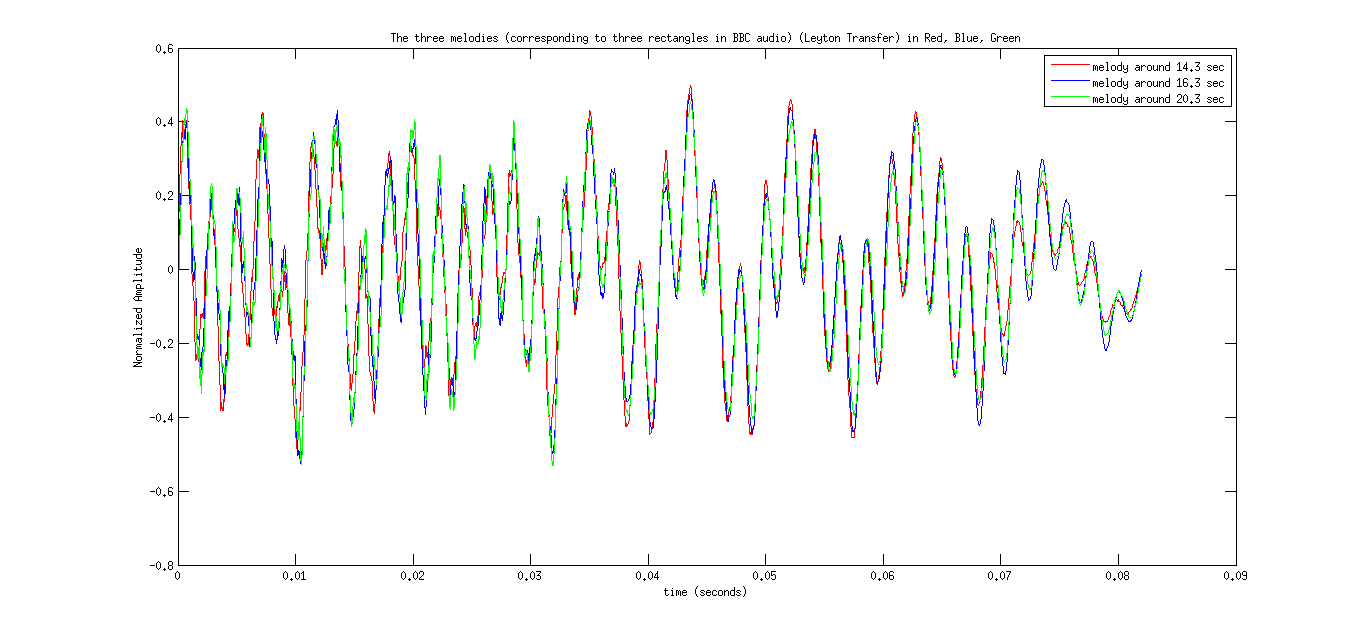}
	}
	\caption{Transfered melodies in sample audio superimposed.}
	\label{fig:11}
\end{figure}

\subsection{A base structured category perspective and a towards unified view of signal representation techniques}
By taking into account that every signal of interest has an appropriate generative model through intuition of \cite{Leyton01} which we have mathematically modeled as some (abstract) category $\mathbf{C}$, the signal model becomes functorial or precisely as $F:\mathbf{C} \to \mathbf{D}$. In~\cite{salilp1} we argued that the construction $(F,\mathbf{C},\mathbf{D})$ offers a distinctly refreshingly new perspective on any ordinary functor ${F}: \mathbf{C} \rightarrow \mathbf{D}$. It is the mathematical expression which combines the intuitive generative perspective of Leyton's Theory in Psychology~\cite{Leyton86a},~\cite{Leyton01} along with well-known Grothendieck's relative point of view~\cite{grorel} by treating it as a fibration. The distinct perspective was that while general fibred categories carry both base category as well as fibre category structure through cartesian and vertical arrows, the base structured categories carry only the base structure concealing the vertical structure within the objects since the vertical arrows are just identity arrows.

Recalling that every functor being a structure preserving morphism we can model same signal equivalently as the base structured category $(F,\mathbf{C},\mathbf{D})$ where the first component of its objects and morphisms signifies natural generators and transfers in generative theory of \cite{Leyton01}.  However by considering objects $F(G_1),F(G_2),...$ set-theoretically as vectors in linear spaces, it becomes possible to view the signal spaces as concrete categories. We illustrate in the case of signal as a functor $F:\mathbf{C} \to \mathbf{Meas}^{\rightarrow}$ or equally as the subcategory $F(\mathbf{C})$ of $\mathbf{Meas}^{\rightarrow}$.  Now we take both set-theoretic and category-theoretic viewpoints simultaneously by noting that $(f\restr I)$ is both an object of $\mathbf{Meas}^{\rightarrow}$ and an element of Riesz space $\mathfrak{L^0}_I$ using the underlying field property of $\Bbb R$. But $\mathfrak{L^0}_I$ is itself an object of category $\mathbf{Riesz}$ thus in spirit of classical signal space if we say global signal $f \in \mathfrak{L^0}_{\Bbb R}$ then functorial signal $F(\mathbf{C})$ belongs to a {\bf signal-matched subcategory} $\mathbf{S}$ of $\mathbf{Riesz}$. The local signal signals $(f\restr I) \in \mathfrak{L^0}_I$, $(f\restr J) \in \mathfrak{L^0}_J$,... where $\mathfrak{L^0}_I$, $\mathfrak{L^0}_J$ form objects of that subcategory while the arrows $(h,\phi)$ giving rise to Riesz homomorphisms $T$ between these objects in certain cases (see Propositions~\ref{prop:L0_func_space},~\ref{prop:L2_func_space} and~\ref{prop:lin_bor_fun} for such cases)  form the arrows of that subcategory. Often since in classic signal representations we don't distinguish between $f,g \in \mathfrak{L^0}_{\Bbb R}$ where $f\eae g$ we consider $(f\restr I)^{\ssbullet}$ by passing to quotient spaces ${L^0}(\mu_I)$ or ${L^2}(\mu_I)$ (with additional property that $|(f\restr I)|^2$ is integrable) by defining additional functor $F': \mathbf{S} \rightarrow \mathbf{Riesz}$ or $F': \mathbf{S} \rightarrow \mathbf{Hilb}$ which sends objects such as $\mathfrak{L^0}_I$ to ${L^0}(\mu_I)$ (or ${L^2}(\mu_I)$) and homomorphisms $g\mapsto h^{-1}g\phi:\mathfrak{L}^0_J \to\mathfrak{L}^0_I$ to $g^{\ssbullet}\mapsto (h^{-1}g\phi)^{\ssbullet}:{L}^0(\mu_J) \to {L}^0(\mu_I)$ or $g^{\ssbullet}\mapsto (h^{-1}g\phi)^{\ssbullet}:{L}^2(\mu_J) \to {L}^2(\mu_I)$. Since there is no standard arrow category in which $(f\restr I)^{\ssbullet}$ of $L^0(\mu_I)$ could be considered as an object, we have used an additional functor $F'$. We term category $\mathbf{S}$ or $F'(\mathbf{S})$ as {\bf signal matched space}, reflecting the intuitive fact that such a category is compatible with the generative structure of the specific signal to be represented in contrast to a generically fixed space for all signals.

As a simple example of translational redundancy in a one-dimensional (global) time signal as shown in Figure~\ref{fig:redundancy}. The functor $F:\mathbf{C} \to \mathbf{Meas}^{\rightarrow}$ and duality is illustrated in Equation~\ref{eqn:redundancy}.

\begin{equation}
\label{eqn:redundancy}
\xymatrix{
	(\Bbb R,\Sigma_{\Cal B}) \ar[r]^{id} & (\Bbb R,\Sigma_{\Cal B})  \\
	(I,\Sigma_{I}) \ar[r]_{\phi} \ar@{->}[u]^{f = g\phi} & (J,\Sigma_{J}) \ar@{->}[u]_{f'= g} \\
	G_1 \ar[r]^{a} & G_2
} 
\quad
\xymatrix{
	(\Sigma_{\Cal B},\symmdiff,\cap) \ar@{->}[d]_{f^{op} = \pi \xi} & (\Sigma_{\Cal B},\symmdiff,\cap) \ar[l]_{id} \ar@{->}[d]^{f'^{op} = \xi (E \mapsto g^{-1}[E])} \\
	(\Sigma_{I},\symmdiff,\cap)   & (\Sigma_{J},\symmdiff,\cap) \ar[l]^{\pi} \\
	G_1  & G_2 \ar[l]_{a^{op}}
}
\end{equation}

In this case the objects $F(G_1),F(G_2),...$ belong to linear spaces $\mathfrak{L}^0$ of real-valued measurable functions on $I$ and $J$ respectively. Now since $a: G_1 \rightarrow G_2$ in this case is translation, $F(a) = (\phi,id)$ and therefore the linear transformation $T$ in opposite direction is purely determined by $\phi$. The spaces $\mathfrak{L}^0$ of arrows containing $F(G_1)$, $F(G_2),...$ become the objects of such a signal space while the arrows such as $F(a),F(a'),...$ give rise to the transformations of those spaces in contravariant reverse way. In case of $a: G_1 \rightarrow G_2$ purely determined by $\phi$ where $\phi$ also has the additional property of being inverse-measure-preserving (which includes the popular cases of translation or scaling between $I$ and $J$) we can consider measures on spaces $(I,\Sigma_{I})$ $(J,\Sigma_{J})$ and consider equivalence classes $f^{\ssbullet}$, $f'^{\ssbullet}$ of arrows under $\eae$. Then one can invoke the functorial nature of $L^0$ or $L^2$ as studied in Section~\ref{sec:funcsigspace} and {\bf signal matched space} becomes the image subcategory of $L^2|_{\mathbf{C}}:$ $\mathbf{C}$ $\rightarrow \mathbf{BanLatt}$ or $L^2|_{\mathbf{C}}:$ $\mathbf{C}$ $\rightarrow \mathbf{Hilb}$ if we are interested only in linear and norm structure instead of additional lattice and multiplicative structures. In conclusion both local signals such as $f\restr I$ and local spaces $L^2(I,\Sigma_{I},\mu_I)$ containing $f\restr I^{\ssbullet}$ have {\bf dual structures viz the translation structure coming from base arrow and the linear structure coming from the underlying sets with additional properties} in contrast to pure linear structure in classic representation techniques.

At the end of this paper we briefly hint at the possible potential of unifying existing signal representation techniques. By variation of base category to include graph with edge composition or topology through $\mathfrak{Opn}_X$ or faces of simplical complex instead of measure and varying codomain category to $\mathbf{Vect}$ instead of $\mathbf{Riesz}$ appropriately it should be possible cast these~\cite{robinson-nt},~\cite{shuman_SPM_2013} as sort of special cases of functorial framework in near future. The case of treating the Weyl Hiesenberg or affine Group as base category and codomain category as $\mathbf{Hilb}$ covers the Gabor or wavelet signal representation in $L^2(\Bbb R)$ as studied in classic paper~\cite{heilwalnut}. Further the isomorphisms in these categories could potentially model other types of structural redundancies in the similar fashion we used isomorphisms of measurable spaces. 
\section{Conclusion and Extension}
\label{conclude}
The paper developed concept first presented as an abstract at CT 2016 \cite{salilct2016}. Table~\ref{table:summary_sig_rep} summarizes the comparison between classic signal representation framework with the new proposed functorial framework. By modeling a source with memory as a groupoid in tune with generative intuition we seek to capture isomorphic relationships between waveforms generated by the source directly impacting the amount of perceived information in signal. The memoryless source as a special case having no interdependencies of successive messages is modeled as a discrete category. In functorial framework, the relative perspective offered by category theory provides an authentic tool to model interdependence between sub-signals. This leads to arrow-theoretic structural definition of redundancy. It also becomes possible to understand compression in a natural category-theoretic way. By treating objects as trivial categories we can utilize additional set-theoretic properties of objects independently along with treating them simply as objects of enclosing category. This novel concept of using set-theory alongside category-theory simultaneously was utilized in the proposed functorial framework.

\bibliography{p1.bib}
\bibliographystyle{acm}

\appendix
\label{sec:measuretheory}
The references for appendices pertaining measure-theory and functional analysis are~\cite{fremlinmt1},~\cite{fremlinmt2},~\cite{fremlinmt3}~\cite{fremlinmt4}.

\subsection{Appendix for objects of $\mathbf{C}$}
First we quickly recall the definitions of $\sigma$-algebra and measure space. 

\begin{definition}
	If $X$ is a set, then a {\bf $\sigma$-algebra} { (or a {\bf $\sigma$-field})}
	of subsets of $X$ is a family $\Sigma_X \subseteq\Cal PX$ of subsets of $X$ such that
	
	\qquad(i) $\emptyset\in\Sigma_X$;
	
	\qquad(ii) for every $E\in\Sigma_X$, its complement $X\setminus E$ in $X$
	belongs to $\Sigma_X$;
	
	\qquad(iii) for every (countable) sequence $\langle E_n\rangle_{n\in\Bbb N}$ in
	$\Sigma_X$, its union $\bigcup_{n\in\Bbb N}E_n$ belongs to $\Sigma_X$.
\end{definition}

\begin{definition}
	
	A {\bf measure space} is a triple $(X,\Sigma_X,\mu)$, where
	
	\quad(i) $X$ is a set;
	
	\quad(ii) $\Sigma_X$ is a $\sigma$-algebra of subsets of $X$;
	
	\quad(iii) $\mu:\Sigma_X\to[0,\infty]$ is a function such that
	
	\qquad($a$) $\mu\emptyset=0$;
	
	\qquad($b$) if $\langle E_n\rangle_{n\in\Bbb N}$ is any disjoint
	sequence in $\Sigma_X$, then $\mu(\bigcup_{n\in\Bbb N}E_n)=\sum_{n=0}^{\infty}\mu E_n$.
	
	The function $\mu$ is called a {\bf measure on $X$}, and the members of $\Sigma_X$ are called {\bf measurable} sets.
\end{definition}

Let $(X,\Sigma_X,\mu)$ be any measure
space. A set $N\subseteq X$ is {\bf negligible} (or
{\bf null} or {\bf $\mu$-negligible}) if
there exists a set $E\subseteq\Sigma_X$ such that $N\subseteq E$ and
$\mu E=0$.

\begin{proposition}
	\label{propneg}
	
	Let $\Cal N$ be the family of negligible subsets of $X$. Then 
	
	\qquad(i) $\emptyset\in\Cal N$;
	
	\qquad(ii) if $A\subseteq B\in\Cal N$ then $A\in\Cal N$;
	
	\qquad(iii) if $\langle A_n\rangle_{n\in\Bbb N}$ is any sequence
	in $\Cal N$, $\bigcup_{n\in\Bbb N}A_n\in\Cal N$.
	
\end{proposition}

These are easy to prove applying the basic definition of negligible sets and properties of measure spaces; Of course $\mu\emptyset=0$. Next there exists a set $E\in\Sigma_X$ such that $\mu E=0$ and $B\subseteq E$ which implies $A\subseteq E$. Finally For each $n\in\Bbb N$ we choose an $E_n\in\Sigma_X$ such that $A_n\subseteq E_n$ and $\mu E_n=0$. But $E=\bigcup_{n\in\Bbb N}E_n\in\Sigma_X$ by definition; while $\bigcup_{n\in\Bbb N}A_n\subseteq\bigcup_{n\in\Bbb N}E_n$, and since $\mu(\bigcup_{n\in\Bbb N}E_n)\le\sum_{n=0}^{\infty}\mu E_n$, so $\mu(\bigcup_{n\in\Bbb N}E_n)=0$ again implying $\bigcup_{n\in\Bbb N}A_n \in\Cal N $.  

In general a family of sets satisfying the conditions (i)-(iii) of Proposition~\ref{propneg} is called a {\bf $\sigma$-ideal} of sets. In particular $\Cal N$ is called the {\bf null ideal} of the measure $\mu$.

Conventionally the term {measurable space} is used to mean a pair $(X,\Sigma_{X})$ where $X$ is a set along with $\Sigma_{X}$ as a $\sigma$-algebra of its subsets. However for the purpose of this sequel we shall avoid using this terminology. Here by the phrase {\bf measurable space} we mean the triple $(X,\Sigma_{X},\Cal N)$ where $\Cal N$ is the null ideal of the measure $\mu$.

\begin{proposition}
	\label{eq_rel_ideal}
	Let $(X,\Sigma_X,\mu)$ be a measure space and for
	$E$, $F\in\Sigma_X$ write $E\sim F$ if $E\symmdiff F \in \Cal I$ where $\Cal I$ is $\sigma$-ideal of subsets of $X$. We now  show that
	$\sim$ is an equivalence relation on $\Sigma_X$.   Let $\frak B$ be the
	set of equivalence classes in $\Sigma_X$ for $\sim$;  for $E\in\Sigma_X$,
	write $E^{\ssbullet}\in\frak B$ for its equivalence class. We also show that
	there is a partial ordering $\subseteq$ on $\frak B$ defined by saying
	that, for $E$, $F\in\Sigma_X$,
	
	\centerline{$E^{\ssbullet}\subseteq F^{\ssbullet}
		\,\iff\,E\setminus F \in \Cal I$.}
\end{proposition}
\begin{proof}
	Proof : 1. Reflexivity $E\sim E$ since $E\symmdiff E = \emptyset \in \Cal I$.
	2. Symmetry: if $E\symmdiff F \in \Cal I$ then $F\symmdiff E \in \Cal I$ and vice versa.
	3. Transitivity: if $E\symmdiff F \in \Cal I$ and $F\symmdiff G \in \Cal I$ then 
	using venn diagrams and the second clause that if $A\subseteq B\in\Cal I$ then $A\in\Cal I$
	it follows that $E\symmdiff G \in \Cal I$.
	
	Next the axioms for the partial order are easily verified. 
	1. Reflexivity: {$E^{\ssbullet}\subseteq E^{\ssbullet}
		\,\iff\,E\setminus E = \emptyset \in \Cal I$.} 
	2. Antisymmetry: If $E^{\ssbullet}\subseteq F^{\ssbullet}
	\,\iff\,E\setminus F \in \Cal I$ and $F^{\ssbullet}\subseteq E^{\ssbullet}
	\,\iff\,F\setminus E \in \Cal I$ then $E\symmdiff F \in \Cal I$ and therefore
	$E^{\ssbullet} = F^{\ssbullet}$ using the third clause of ideal definition.
	3. Transitivity: If $E^{\ssbullet}\subseteq F^{\ssbullet}
	\,\iff\,E\setminus F \in \Cal I$ and $F^{\ssbullet}\subseteq G^{\ssbullet}
	\,\iff\,F\setminus G \in \Cal I$ then again using venn diagrams and second clause
	of ideal definition, it readily follows that $E^{\ssbullet}\subseteq G^{\ssbullet}
	\,\iff\,E\setminus G \in \Cal I$.
\end{proof}

As a special case when $\Cal I = \Cal N$, we have the following corollary.

\begin{corollary}
	\label{cor:eq_rel_ideal}
	Let $(X,\Sigma_X,\mu)$ be a measure space, and for
	$E$, $F\in\Sigma_X$ write $E\sim F$ if $\mu(E\symmdiff F)=0$.   Show that
	$\sim$ is an equivalence relation on $\Sigma_X$.   Let $\frak B$ be the
	set of equivalence classes in $\Sigma_X$ for $\sim$;  for $E\in\Sigma_X$,
	write $E^{\ssbullet}\in\frak B$ for its equivalence class.   Show that
	there is a partial ordering $\subseteq$ on $\frak B$ defined by saying
	that, for $E$, $F\in\Sigma_X$,
	
	\centerline{$E^{\ssbullet}\subseteq F^{\ssbullet}
		\,\iff\,\mu(E\setminus F)=0$.}
\end{corollary}

\subsection{Appendix for objects of $\mathbf{C}^{op}$}

To understand and utilize the dual relationship between the categories involving measure spaces and measure algebras, first we need to briefly review the intuition relating sets to and Boolean algebras. 	


\begin{definition}
	If $X$ be a set, then an {\bf algebra} or a {\bf field} is a family $\Cal A\subseteq\Cal PX$  of
	subsets of $X$ such that
	
	\inset{(i) $\emptyset\in\Cal A$;
		
		(ii) for every $A\in\Cal E$, its complement $X\setminus A$
		belongs to $\Cal A$;
		
		(iii) for every $E$, $F\in\Cal A$, $E\cup F\in\Cal A$. (closed under finite unions).}
\end{definition}

Since closure under countable unions implies closure under finite unions; every $\sigma$-algebra of subsets of $X$ is always
an algebra of subsets of $X$.

A {\bf Boolean algebra} is a Boolean ring $(\frak B,+,.)$ (Boolean means $b^2=b$ for every $b\in\frak B$)
with a multiplicative identity $1_{\frak B}$.

\begin{proposition}
	If $\Cal A\subseteq\Cal PX$ is an {\bf algebra} of subsets of any set $X$ then $(\Cal A,\symmdiff,\cap)$ is a Boolean algebra;  its additive identity $0_{\Cal A}$ is $\emptyset$ while its multiplicative identity $1_{\Cal A}$ is $X$. Here the operations $\symmdiff$, $\cup$ are the usual operations of set theory.
\end{proposition}
\begin{proof}
	We verify the axioms, which are all easily established, using Venn diagrams.
	
	1. {$A\symmdiff B \in \Cal A$ for all $A$, $B\in \Cal A$, (additive closure)}. 
	
	2. {$(A\symmdiff B)\symmdiff C=A\symmdiff(B\symmdiff C)$ for all $A$, $B$, $C\in \Cal A$, (additive associativity)} 
	
	3. {$A\symmdiff\emptyset=\emptyset\symmdiff A=A$ for every $A\in \Cal A$, (additive identity)} 
	
	4. {$A\symmdiff A=\emptyset$ for every $A\in \Cal A$,} 
	
	\noindent so that every element of $\Cal A$ is its own inverse in 
	$(\Cal A,\symmdiff)$, and $(\Cal A,\symmdiff)$ is a group; 
	
	5.{$A\symmdiff B=B\symmdiff A$ for all $A$, $B\in \Cal A$, (commutativity)} 
	
	\noindent so that $(\Cal A,\symmdiff)$ is an abelian group; 
	
	6. {$A\cap B\in \Cal A$ for all $A$, $B\in \Cal A$, (multiplicative closure)} 
	
	7. {$(A\cap B)\cap C=A\cap(B\cap C)$ for all $A$, $B$, $C\in \Cal A$, (multiplicative associativity)} 
	
	8. {$A\cap(B\symmdiff C)=(A\cap B)\symmdiff(A\cap C)$, $(A\symmdiff 
		B)\cap C=(A\cap C)\symmdiff(B\cap C)$ for all $A$, $B$, $C\in \Cal A$,(distributivity)} 
	
	\noindent so that $(\Cal A,\symmdiff,\cap)$ is a ring; 
	
	9. {$A\cap A=A$ for every $A\in \Cal A$,} (Boolean axiom $b^2=b$) 
	
	\noindent so that $(\Cal A,\symmdiff,\cap)$ is a Boolean ring; 
	
	10. {$A\cap X=X\cap A=A$ for every $A\in \Cal A$, (multiplicative identity)}  
	
	\noindent so that $(\Cal A,\symmdiff,\cap)$ is a Boolean algebra and 
	$X$ is its identity.
\end{proof} 

Thus every algebra of sets $(\Cal A,\symmdiff,\cap)$ is isomorphic to a Boolean algebra $(\frak B,+,.)$ where the elements $a,b,c,..\in \frak B$ of Boolean algebra are in bijective correspondence with sets $A,B,C,...\in \Cal A$. While the operations satisfy $a+b = A \symmdiff B$; $a.b = A \cap B$; $A \setminus B = a + a.b$; $A \cup B = a + b + a.b$.
Conversely it is also true that every Boolean algebra $(\frak B,+,.)$ is isomorphic to some algebra of sets $(\Cal A,\symmdiff,\cap)$. This is often known as set-theoretical version of Stone’s representation theorem.

Underlying the structure of every Boolean algebra $(\frak B,+,.)$ is a {\bf partially ordered set} $(\frak B,\leq)$ where the partial order is defined by $a \leq b$ whenever $a \vee b = b$ (or equivalently $a \wedge b = a$). Further just as in the case of partial order, we have corresponding notions of upper bounds, lower bounds, join (supremum or least upper bound), meet (infimum or greatest lower bound), (Dedekind) completeness in the case of a Boolean algebra. The binary operations join ($\vee$) and meet ($\wedge$) on elements of algebra correspond to the set-theoretic union and intersection operations. Thus $a + b + a.b = a \vee b$ and $a.b = a \wedge b$ while the lower bound is the additive identity $0_{\frak B}$ or $\emptyset$ its upper bound is the multiplicative identity $1_{\frak B}$ or $X$. 

The notion of completeness fundamental to the theory of measure algebras is briefly reviewed next.

\begin{definition}  If $P$ is a partially ordered set then;
	
	(a) $P$ is {\bf conditionally complete} , or
	{\bf order-complete} or {\bf Dedekind complete} when every
	non-empty subset of $P$ with an upper bound has a least upper bound.
	
	(b) $P$ is {\bf $\sigma$-order-complete}, or {\bf Dedekind $\sigma$-complete} when (i) every countable non-empty subset
	of $P$ with an upper bound has a least upper bound (ii) every countable
	non-empty subset of $P$ with a lower bound has a greatest lower bound.
	
	In the special case of Boolean algebras, one half of this definition implies the other; more precisely we have,
	for any Boolean algebra $\frak B$,
	
	$$\eqalign{\frak B&\text{ is (Dedekind) }\sigma\text{-complete}\cr
		&\iff\text{ every non-empty countable subset of }\frak B
		\text{ has a least upper bound}\cr
		&\iff\text{ every non-empty countable subset of }\frak B
		\text{ has a greatest lower bound}.\cr}$$
\end{definition} 

Note that $\Sigma_X$ is Dedekind $\sigma$-complete, because if
$\sequencen{E_n}$ is any sequence in $\Sigma_X$ then $\bigcup_{n\in\Bbb
	N}E_n = \bigvee \{E_n:n\in\Bbb N\}$ is the least upper bound of $\{E_n:n\in\Bbb N\}$ in $\Sigma_X$.

\begin{definition} 
	\label{def:meas_alg}
	A {\bf measure algebra} is a pair $(\frak B,\bar\mu)$, where 
	
	\quad(i) $\frak B$ is a Dedekind $\sigma$-complete Boolean algebra;
	
	\quad(ii) $\bar\mu:\frak B\to[0,\infty]$ is a function such that 
	
	\qquad($a$) $\bar\mu 0=0$; 
	
	\qquad($b$) whenever $\sequencen{a_n}$ is a disjoint sequence in $\frak B$, 
	$\bar\mu(\sup_{n\in\Bbb N}a_n)=\sum_{n=0}^{\infty}\bar\mu a_n$; 
	
	\qquad($c$) $\bar\mu a>0$ whenever $a\in\frak B$ and $a\ne 0$.
\end{definition}

\begin{proposition} 
	\label{quotient_ba_com}
	Let $\frak B$ be a Dedekind $\sigma$-complete
	Boolean algebra and $I$ a $\sigma$-ideal of $\frak B$.   Then the
	quotient Boolean algebra $\frak B/I$ is Dedekind $\sigma$-complete.
\end{proposition}

We give a sketch of proof. Let $A\subseteq\frak B/I$ be a
non-empty countable set.   For each $u \in A$, choose a $b_u\in \frak B$
such that $u=b_u^{\ssbullet}$.   Then $c=\sup_{u\in A}b_u$ is
surely defined in $\frak B$;  consider $v=c^{\ssbullet}$ in $\frak B/I$.
Now it can be shown that the map $b\mapsto \phi(b)=b^{\ssbullet}$ is sequentially
order-continuous. This means $\phi(c) = \sup_{u \in A} \phi(b_u)$ implying  
$v=c^{\ssbullet} = \sup_{u \in A} b_u^{\ssbullet} =\sup A$. But $A$ is arbitrary, 
so $\frak B/I$ is Dedekind $\sigma$-complete.

\begin{corollary}
	\label{cor:quotient_ba_com}
	Let $X$ be a set, $\Sigma_X$ a $\sigma$-algebra
	of subsets of $X$, and $\Cal I$ a $\sigma$-ideal of subsets of $X$.
	Then $\Sigma_X\cap\Cal I$ is a $\sigma$-ideal of the Boolean algebra
	$(\Sigma_X,\symmdiff,\cap)$, and $\Sigma_X/\Sigma_X\cap \Cal I$ is Dedekind $\sigma$-complete.
\end{corollary}

Ofcourse if we just prove that $\Sigma\cap\Cal I$ is a $\sigma$-ideal then result
automatically follows from the Proposition~\ref{quotient_ba_com}. Now $\Sigma_X\cap\Cal I$ is a family
of subsets of $X$ each having a form $F \cap E$ where $F\in\Sigma_X$ and $E\in\Cal I$.

\quad(i) We have $\emptyset \cap \emptyset \in\Sigma_X\cap\Cal I$ satisfying first clause of $\sigma$-ideal definition.

\quad(ii) Let $A\subseteq (F \cap E)\in\Sigma_X\cap\Cal I$,
then $A$ could be expressed in a form $(F \cap G)$ where $G\subseteq E\in\Cal I$ but since $G\in\Cal I$ hence
$A\in\Sigma_X\cap\Cal I$. satisfying second clause of $\sigma$-ideal definition

\quad(iii) Let $\langle A_n\rangle_{n\in\Bbb N}$ be any sequence in $\Sigma_X\cap\Cal I$, where $A_n = F_n \cap E_n$.
Let $F = (\bigcup_{n\in\Bbb N}F_n) \in \Sigma_X$, then $\bigcup_{n\in\Bbb N}A_n \subseteq  (\bigcup_{n\in\Bbb N}F_n)\cap(\bigcup_{n\in\Bbb N}E_n) \in\Sigma_X\cap\Cal I$. Hence by second clause above $\bigcup_{n\in\Bbb N}A_n \in\Sigma_X\cap\Cal I$ satisfying third clause.

Indeed $\Sigma_X\cap\Cal I$ is a $\sigma$-ideal of the Boolean algebra  $\Sigma_X$, and $\Sigma_X/\Sigma_X\cap \Cal I$ is Dedekind $\sigma$-complete.

\subsection{Atomic measure spaces and dual atomic Boolean algebras.}
\label{sec:atomic}
\begin{definition}
	\label{def:pointsupp}
	Suppose $X$ is any set, and $h:X\to[0,\infty]$ is any function. For every $E\subseteq X$ if we
	declare $\mu E=\sum_{x\in E}h(x)$ taking $\sum_{x\in\emptyset}h(x)=0$ then $(X,\Cal PX,\mu)$ becomes a measure space.
	Such measures $\mu$ are termed as {\bf point-supported}.
\end{definition}

Note that for infinite sets $E$ one can take
$\sum_{x\in E}h(x)=\sup\{\sum_{x\in I}h(x):I\subseteq E$ is finite$\}$,
because every $h(x)$ is non-negative. As an example, if $X=\Bbb N$, then
for countably infinite $E$ it reduces to
$\sum_{n\in E}h(n)=\lim_{n\to\infty}\sum_{m\in E,m\le n}h(m)$.

This special case of $h(x)=1$ for every $x$, is called {\bf counting measure} on $X$. Here $\mu E$
is just the number of points of $E$ given $E$ is finite, else is $\infty$ if $E$ is infinite.  

\begin{definition}
	Let $(X,\Sigma,\mu)$ be a measure space;
	
	\quad(i) A set $A\in\Sigma$ is an {\bf atom} for $\mu$ if $\mu A>0$ and
	whenever $E\in\Sigma$, $E\subseteq A$ either $E$ or $A\setminus E$ is
	negligible.
	
	\quad(ii) Measure $\mu$, or space $(X,\Sigma,\mu)$, is {\bf purely atomic} or {\bf discrete} if whenever
	$E\in\Sigma$ and $E$ is not negligible there is an atom for $\mu$ included in $E$.
	
	\quad(iii)Measure $\mu$, or space $(X,\Sigma,\mu)$, is {\bf atomless} or {\bf diffused} or {\bf continuous} if 
	there are no atoms for $\mu$.
	
\end{definition}

Every counting measure and in general point-supported measure is purely atomic because 
$\{x\}$ must be an atom whenever $\mu\{x\}>0$.

Dually in Boolean algebras the equivalent definitions are reviewed now.

\begin{definition}
	Let $\frak B$ be a Boolean algebra;
	
	\quad(i) An {\bf atom} in $\frak B$ is a non-zero $a\in\frak B$ such
	that the only elements included in $a$ are $0$ and $a$.
	
	\quad(ii) $\frak B$ is {\bf atomless} if it has no atoms.
	
	\quad(iii) $\frak B$ is {\bf purely atomic} if every non-zero element includes an atom.
\end{definition}

\medskip

A counting measure on a set $X$ is always complete, strictly localizable and purely atomic. It is $\sigma$-finite iff $X$ is countable, totally finite iff $X$ is finite, a probability measure iff $X$ is a singleton, and atomless iff $X$ is empty.


\subsection{Appendix for Base Arrows}

A basic result regarding a function between sets states that if $X$ and $Y$ are some sets and $\Tau$ is a
$\sigma$-algebra of subsets of $Y$; then for a function $\phi:X\to Y$,
$\{\phi^{-1}[F]:F\in\Tau\}$ is a $\sigma$-algebra of subsets of $X$.

Again we sketch the basic steps for proof.
1. $A = \{\phi^{-1}[F]:F\in\Tau\}$ 2. $Y\in\Tau, \phi^{-1}[Y] = X \in A$ 
3. $E\in\Tau, \phi^{-1}[E] \in A, (Y\setminus E) \in\Tau, \phi^{-1}[Y\setminus E] = 
X \setminus \phi^{-1}[E] \in A $ 4. $X \setminus \phi^{-1}[Y] = \emptyset \in A $
5. for every sequence $\langle E_n\rangle_{n\in\Bbb N} \in \Tau $ its union $\bigcup_{n\in\Bbb N}E_n \in
\Tau $,$\phi^{-1}[\bigcup_{n\in\Bbb N}E_n] = \bigcup_{n\in\Bbb N}\phi^{-1}[E_n] \in A$.

If we associate $\Sigma_X$, some $\sigma$-algebra of subsets of $X$ along with $X$ then a function $\phi:X\to Y$
is termed as $\Sigma_X$-measurable if $\{\phi^{-1}[F]:F\in\Tau\} \subseteq \Sigma_X$. Hence for a general real-valued
function with Borel $\sigma$-algebra $\Sigma_{\Cal B}$ on the codomain $\Bbb R$ the following constitutes as one of the basic definitions of measure theory. 

\begin{definition}
	\label{def:subspace_sigma_meas_func}
	Let $X$ be any set, $\Sigma_X$ a $\sigma$-algebra
	of subsets of $X$, and $E$ a subset of $X$.   A function $f:E\to\Bbb R$
	is called {\bf $\Sigma_X$-measurable} (or {\bf measurable}) if it
	satisfies any, or equivalently all, of following conditions 
	
	\quad(i) $\{x:f(x)<a\}\in\Sigma_E$ for every $a\in\Bbb R$;
	
	\quad(ii) $\{x:f(x)\le a\}\in\Sigma_E$ for every $a\in\Bbb R$;
	
	\quad(iii) $\{x:f(x)>a\}\in\Sigma_E$ for every $a\in\Bbb R$;
	
	\quad(iv) $\{x:f(x)\ge a\}\in\Sigma_E$ for every $a\in\Bbb R$.
	
	where $\Sigma_E$ is the subspace $\sigma$-algebra of subsets of $E \subseteq X$.
\end{definition}

Ofcourse if $X$ is $\Bbb R^r$, and $\Cal B$ is its Borel
$\sigma$-algebra, a $\Sigma$-measurable function is
called {\bf Borel measurable}. If $X$ is $\Bbb R^r$, and 
$\Sigma_{\Cal L}$ is the $\sigma$-algebra of Lebesgue measurable sets,
then a $\Sigma_X$-measurable function is called {\bf Lebesgue measurable}.

\begin{remark}
	\label{rem:partial_total_measurable_1}
	It is important to note that Definition~\ref{def:subspace_sigma_meas_func} given in~\cite{fremlinmt2} is generally defined on any subset $E \subseteq X$ or in other words partial measurable functions. From the category theory perspective unless we are working with partial categories, the measurable functions will be taken as defined on total domain $X$. Following ~\cite{fremlinmt1},~\cite{fremlinmt2} and~\cite{fremlinmt3} we have retained this generality for proving theorems and propositions here since we can directly use these when dealing with partial categories in measure theory. Moreover special the case of (total) measurable functions on total $X$ is much simpler once we are used to the general case. We caution the reader whenever this distinction is crucial in given context. A particular one is the difference between $\mathfrak{L^0}_{X}$, the space of all measurable functions from $X$ to $\Bbb R$ and general $\mathfrak{L^0}$, the space of all partial functions $f$ from $E$ to $\Bbb R$ where $E \subseteq X$ is conegligible and $f\restr F$ is measurable for some conegligible set $F \subseteq X$. See Definitions~\ref{def:virtually_measurable_fun_space} and~\ref{def:measurable_fun_space} and Remark~\ref{rem:partial_total_measurable_2}.
\end{remark}   

As mentioned earlier by the phrase {\bf localizable measurable space} we mean the triple $(X,\Sigma_{X},\Cal N)$ where $\Cal N$ is the null ideal of the measure $\mu$. Hence for this structured object the appropriate morphism intuitively is a {$\Sigma_X$-measurable function} roughly also preserving the null structure.

\begin{definition}
	Let $(X,\Sigma_X,\mu)$,$(Y,\Sigma_Y,\nu)$ be measure spaces and $(X,\Sigma_X,\Cal M)$,$(Y,\Sigma_X,\Cal N)$
	be the corresponding measurable spaces where $\Cal M$,$\Cal N$ are the {\bf null ideals} of measures $\mu$,$\nu$ respectively. A function $\phi:X\to Y$ such that $\phi^{-1}[F]\in\Sigma_X$ for every $F\in\Sigma_Y$ and $\mu\phi^{-1}[F]=0$ whenever $\nu F=0$ is called a {\bf non-singular measurable} or {\bf measure-zero-reflecting} function.
\end{definition}

Note that every element in the null ideal need not be measurable yet because we ensure that pre-image of every $F\in\Sigma_Y$ whenever $\nu F=0$ is also measurable with measure zero, consequently $\{\phi^{-1}[N]:N\in\Cal N\} \subseteq \Cal M$. In other words the pre-image of every element of $\Cal N$ ($\sigma$-ideal of $\nu$-negligible sets family) is an element of $\Cal M$ ($\sigma$-ideal of all $\mu$-negligible sets family). For measure spaces the appropriate morphisms roughly also preserve actual measures as defined next.  

\begin{definition}
	If $(X,\Sigma_X,\mu)$ and $(Y,\Sigma_Y,\nu)$ are measure
	spaces, a function $\phi:X\to Y$ is {\bf inverse-measure-preserving} if
	$\phi^{-1}[F]\in\Sigma$ and $\mu(\phi^{-1}[F])=\nu F$ for every
	$F\in\Tau$.
\end{definition}

\begin{definition} Let $(\frak B,\bar\mu)$ and
	$(\frak A,\bar\nu)$ be
	measure algebras.   A Boolean homomorphism $\phi:\frak B\to\frak A$ is
	{\bf measure-preserving} if $\bar\nu(\phi b)=\bar\mu b$ for every
	$a\in\frak B$.
\end{definition}

\begin{theorem} 
	\label{thm:seq_order_con}
	Let $\frak B$ and $\frak A$ be Boolean algebras
	and $\phi:\frak B\to\frak A$ a Boolean homomorphism (a ring homomorphism)
	with kernel $I$.
	
	\quad(i) If $\phi$ is sequentially order-continuous then $I$ is
	a $\sigma$-ideal.
	
	\quad(ii) If $\phi[\frak B]$ is regularly embedded in $\frak A$ and $I$
	is a $\sigma$-ideal then $\phi$ is sequentially order-continuous.
\end{theorem}

\begin{proof}
	First note that the phrase {\bf Boolean
		homomorphism} is simply a function $\phi:\frak B\to\frak A$ which is a
	ring homomorphism. Hence by definition,
	$\phi(a\symmdiff b)=\phi a\symmdiff\phi b$,
	$\phi(a\cap b)=\phi a\cap\phi b$ for all $a$, $b\in\frak B$ and $\phi(1_{\frak B})=1_{\frak A}$.
	Of-course the kernel of a (ring) homomorphism is the set of elements mapped to $0$ and it is
	always an ideal in $\frak B$ as easily verified.
	
	{{\bf(i)} 
		To prove $I$ is a $\sigma$-ideal we just have to show it is sequentially order-closed.
		If $\sequencen{a_n}\subseteq I$ is a non-decreasing sequence and
		has a supremum $c\in \frak B$, then $\phi$ being sequentially order-continuous, $\phi c=\phi(\sup_{n\in\Bbb N}a_n) =\sup_{n\in\Bbb N}\phi(a_n)=0$, so $c\in I$ proving that $I$ is a $\sigma$-ideal.
		
		\medskip
		
		\quad{\bf (ii)} The term regular embedding refers to an {\bf injective order-continuous}
		Boolean homomorphism. The sub-algebra $\phi[\frak B]$ is regularly embedded in $\frak A$ means the
		identity map from $\phi[\frak B]$ to $\frak A$ is order-continuous, implying that whenever
		$c\in\phi[\frak B]$ is the supremum (in $\phi[\frak B]$) of $A\subseteq\phi[\frak B]$, then $c$ is also the supremum in
		$\frak A$ of $A$ and similarly for infima.
		Hence it will be enough to show that
		$\phi$ is sequentially order-continuous when considered as a map from $\frak B$ to
		$\phi[\frak B]$.   Suppose that $\sequencen{a_n}\subseteq\frak B$ is non-increasing sequence and that $\inf_{n\in\Bbb N}a_n=0$. But Suppose, if possible, that $0$ is not the infimum of $\phi[\sequencen{a_n}]$ in $\phi[\frak B]$.
		This means there is $c\in\frak B$ such that $0\ne\phi c\subseteq\phi a_n$ for every $a_n \in \sequencen{a_n}$.
		Now {$\phi(c\setminus a_n)=\phi c\setminus\phi a_n=0$} for every $a_n \in \sequencen{a_n}$, so $c\setminus a_n\in I$ for every
		$a_n \in \sequencen{a_n}$.   The set $C=\{c\setminus a_n:a_n \in \sequencen{a_n}\}$ is non-decreasing and
		has supremum $c$;  since $I$ is a $\sigma$-ideal, $c=\sup C\in I$, and
		$\phi c=0$, contradicting $c \ne 0$.  Thus
		$\inf_{n\in\Bbb N}\phi(a_n)=0$ in either $\phi[\frak B]$ or $\frak A$. But $\sequencen{a_n}$ is
		arbitrary, so $\phi$ is order-continuous.
	}
\end{proof}

Further, following corollary automatically follows since { $\phi[\frak B]=\frak B/I$ is always regularly embedded in $\frak B/I$.
	
	\begin{corollary} 
		\label{cor:seq_order_con}
		Let $\frak B$ be a Boolean algebra and $I$ an
		ideal of $\frak B$;  denote $\phi$ for the canonical map from $\frak B$ to
		$\frak B/I$ then $\phi$ is sequentially order-continuous iff $I$ is a $\sigma$-ideal.
	\end{corollary}

	\begin{theorem} Let $(X,\Sigma,\mu)$ and $(Y,\Tau,\nu)$
		be measure spaces, and $(\frak A,\bar\mu)$, $(\frak B,\bar\nu)$ their
		measure algebras.   Write $\hat\Sigma$ for the domain of the completion
		$\hat\mu$ of $\mu$.   Let $D\subseteq X$ be a set of full outer
		measure, and let $\hat\Sigma_D$ be the
		subspace $\sigma$-algebra on $D$ induced by $\hat\Sigma$.  Let
		$\phi:D\to Y$ be a function such that
		$\phi^{-1}[F]\in\hat\Sigma_D$ for every $F\in\Tau$ and
		$\phi^{-1}[F]$ is $\mu$-negligible whenever $\nu F=0$.   Then
		there is a sequentially order-continuous Boolean homomorphism 
		$\pi:\frak B\to\frak A$ defined by the formula
		
		\centerline{$\pi F^{\ssbullet}=E^{\ssbullet}$ whenever $F\in\Tau$,
			$E\in\Sigma$ and $(E\cap D)\symmdiff\phi^{-1}[F]$ is negligible.}
	\end{theorem}
	
	\begin{proof}
		{ Let $F\in\Tau$.  Then there is an $H\in\hat\Sigma$ such that
			$H\cap D=\phi^{-1}[F]$;  now there is an $E\in\Sigma$ such that
			$E\symmdiff H$ is negligible, so that $(E\cap D)\symmdiff\phi^{-1}[F]$
			is negligible.   If $E_1$ is another member of $\Sigma$ such that
			$(E_1\cap D)\symmdiff\phi^{-1}[F]$ is negligible, then $(E\symmdiff
			E_1)\cap D$ is negligible, so is included in a negligible member $G$ of
			$\Sigma$.   Since $(E\symmdiff E_1)\setminus G$ belongs to $\Sigma$ and
			is disjoint from $D$, it is negligible;  accordingly $E\symmdiff E_1$ is
			negligible and $E^{\ssbullet}=E_1^{\ssbullet}$ in $\frak A$.
			
			What this means is that the formula offered defines a map $\pi:\frak
			B\to\frak A$.   It is now easy to check that $\pi$ is a Boolean
			homomorphism, because if
			
			\centerline{$(E\cap D)\symmdiff\phi^{-1}[F]$, \quad$(E'\cap
				D)\symmdiff\phi^{-1}[F']$}
			
			\noindent are negligible, so are
			
			\centerline{$((X\setminus E)\cap
				D)\symmdiff\phi^{-1}[Y\setminus F]$,
				\quad$((E\cup E')\cap D)\symmdiff\phi^{-1}[F\cup F']$.}
			
			To see that $\pi$ is sequentially order-continuous, let
			$\sequencen{b_n}$ be a sequence in $\frak B$.  For each $n$ we may
			choose an $F_n\in\Tau$ such that $F_n^{\ssbullet}=b_n$, and
			$E_n\in\Sigma$ such that $(E_n\cap D)\symmdiff\phi^{-1}[F_n]$ is
			negligible;  now, setting $F=\bigcup_{n\in\Bbb N}F_n$,
			$E=\bigcup_{n\in\Bbb N}E_n$,
			
			\centerline{$(E\cap D)\symmdiff\phi^{-1}[F]
				\subseteq\bigcup_{n\in\Bbb N}(E_n\cap D)\symmdiff\phi^{-1}[F_n]$}
			
			\noindent is negligible, so
			
			\centerline{$\pi(\sup_{n\in\Bbb N}b_n)=\pi(F^{\ssbullet})
				=E^{\ssbullet}=\sup_{n\in\Bbb N}E_n^{\ssbullet}
				=\sup_{n\in\Bbb N}\pi b_n$.}
			
			\noindent (Recall that the maps $E\mapsto E^{\ssbullet}$, 
			$F\mapsto F^{\ssbullet}$ are sequentially order-continuous, by 321H.)   
			So $\pi$ is sequentially order-continuous (313L(c-iii)).
		}
	\end{proof}
	
	Now the objects considered earlier could be regarded as measurable spaces with some added structure such as null ideal or actual measure or dually as appropriate Boolean algebras. Thus the proper structure preserving maps on these objects are measurable functions which preserve this additional structure such as measurable non-singular (measure zero-reflecting) morphisms and Inverse-measure-preserving maps and dually the corresponding Boolean homomorphisms which were reviewed in this section.
	
	\begin{definition}~\cite{fremlinmt2}
		\label{def:directsum_measure_space}
		If $\langle(X_i,\Sigma_i,\mu_i)\rangle_{i\in I}$ is any indexed family of
		measure spaces, then by setting $X=\bigcup_{i\in I}(X_i\times\{i\})$;  for $E\subseteq X$, $i\in I$
		$E_i=\{x:(x,i)\in E\}$; we have $\Sigma_X
		=\{E:E\subseteq X,\,E_i\in\Sigma_i$ for every $i \in I\}$,
		
		$\mu E=\sum_{i\in I}\mu_iE_i$ for every $E\in\Sigma$.
		
		\noindent Then $(X,\Sigma_X,\mu)$ is a measure
		space and {\bf direct sum} of the family
		$\familyiI{(X_i,\Sigma_i,\mu_i)}$, denoted as $\bigoplus_{i\in I}(X_i,\Sigma_{X_i},\mu_{X_i})$
	\end{definition}
	
	The following property of direct sum suggests intuitively separating a global  measurable function (or its equivalence class under $\eae$) on $(X,\Sigma_X,\mu)$ into local measurable functions (or classes) on subspaces $(X_i,\Sigma_i,\mu_i)$  in the partition of global domain.
	
	\begin{proposition}~\cite{fremlinmt2}
		Let $\familyiI{(X_i,\Sigma_i,\mu_i)}$ be a
		family of measure spaces, with direct sum $(X,\Sigma_X,\mu)$. If $f$ is
		a real-valued function defined on a subset of $X$ and for each $i\in I$,
		if we set $f_i(x)=f(x,i)$ whenever $(x,i)\in\dom f$; the $f$ is measurable 
		iff $f_i$ is measurable for every $i\in I$.
	\end{proposition}
	
	The next result from~\cite{fremlinmt2} based on direct sum makes it precise that a global measurable function $f$ or its equivalence class $f^{\ssbullet}$ can be identified with the local measurable functions as $(f\phi_1,f\phi_2,...)$ or classes $(f^{\ssbullet}\phi_1,f^{\ssbullet}\phi_2,...)$ . 
	
	If $\langle(X_i,\Sigma_i,\mu_i)\rangle_{i\in I}$ is a
	family of measure spaces, with direct sum $(X,\Sigma_X,\mu)$.
	(i) Writing $\phi_i:X_i\to X$ for the canonical maps, $\phi_i(x)=(x,i)$ for $x\in X_i$, it can be shown that
	$f\mapsto\langle f\phi_i\rangle_{i\in I}$ is a bijection between
	$\mathfrak{L^0}(\mu)$ and $\prod_{i\in I}\mathfrak{L^0}(\mu_i)$.   (ii) It
	corresponds to a bijection or a canonical isomorphism between $L^0(\mu)$ and
	$\prod_{i\in I}L^0(\mu_i)$. (iii) This also induces an
	isomorphism between $L^p(\mu)$ and the subspace {$\{u:u\in\prod_{i\in I}L^p(\mu_i),\,
		\|u\|=\bigl(\sum_{i\in I}\|u(i)\|_p^p)^{1/p}<\infty\}$}
	of $\prod_{i\in I}L^p(\mu_i)$, for a $p\in [1,\infty)$.
	
	\subsection{Appendix for Appendix for codomain categories}
	
	\begin{definition}~\cite{fremlinmt2} A {\bf partially ordered linear space} is a linear space $(U,+,\cdot)$ over
		$\Bbb R$ together with a partial order $\le$ on $U$ such that
		
		\centerline{$u\le v\Longrightarrow u+w\le v+w$,}
		
		\centerline{$u\ge 0$, $\alpha\ge 0\Longrightarrow \alpha u\ge 0$}
		
		\noindent for $u$, $v$, $w\in U$ and $\alpha\in\Bbb R$.
		
		{\centerline{$u\le v\iff 0\le v-u\iff -v\le-u$.}}
		{\centerline{$u\le v\Longrightarrow 0=u+(-u)\le v+(-u)=v-u
				\Longrightarrow u=0+u\le v-u+u=v$,}
			
			\centerline{$u\le v\Longrightarrow -v=u+(-v-u)\le v+(-v-u)=-u$.}}
	\end{definition}
	
	\begin{definition}~\cite{fremlinmt2}
		A {\bf Riesz space} or {\bf vector lattice}
		is a partially ordered linear space which is a lattice.
	\end{definition}
	
	\begin{definition}~\cite{fremlinmt2}
		A {\bf Riesz
			homomorphism} from $M$ to $N$ is a linear operator $T:M\to N$ such that
		whenever $B\subseteq M$ is a finite non-empty set and $\inf A=0$ in $M$,
		then $\inf T[B]=0$ in $N$.
	\end{definition}

	\subsection{Appendix for Partial Categories}
	\label{app:parcat}
	
	\begin{definition}
		Restriction Category \cite{cockettlack}. A restriction structure on a category $\mathbf{R}$  consists of an operator $\overline{(\cdot)}$ on morphisms which maps each $f:X \rightarrow Y$ to $\bar{f}: X \rightarrow X$ (termed restriction idempotent of $f$) such that
		\begin{enumerate}
			\item $f \circ \bar{f} = f$ for all $f:X \rightarrow Y$,
			\item $\bar{f} \circ \bar{g} = \bar{g} \circ \bar{f}$ whenever $dom(f) = dom(g)$,
			\item $\overline{f \circ \bar{g}} = \bar{f} \circ \bar{g}$ whenever $dom(f) = dom(g)$,
			\item $\bar{h} \circ \bar{f} = \overline{h \circ f} \circ f$ whenever $dom(f) = dom(g)$,
		\end{enumerate}
		Such a category is termed as restriction category. 
	\end{definition}
	
	In Par$(\mathbf{LocMeas},\mathcal{M})$ the restriction idempotent $\bar{f}:\mathbb{X} \rightarrow \mathbb{X}$ for a partial measurable function $f:\mathbb{X} \rightarrow \mathbb{Y}$ is given by the partial identity function $\bar{f}(x)=x$ wherever $f$ is defined on the measurable space $X$ and undefined otherwise.
	\begin{definition}
		\textbf{Restriction functor}\cite{cockettlack}: A functor $F: \mathbf{R} \rightarrow \mathbf{R'}$ between restriction categories is a restriction (preserving) functor if $\overline{F(f)}=F(\bar{f})$ for every $f \in \mathbf{R}$.
	\end{definition}
	The notion of restriction functor as the term suggests, captures the notion of structure (restriction) preserving morphisms inside the category $\mathbf{rCat}$ of restriction categories and restriction functors.
	Observe that if $\bar{f}=\mathbf{1}_\mathbb{X}$ then it is precisely the total map. Total[Par$(\mathbf{LocMeas},\mathcal{M})$] which is the total subcategory of the argument restriction category embeds via a faithful restriction functor inside Par$(\mathbf{LocMeas},\mathcal{M})$; and is same as the category $\mathbf{LocMeas}$.
	
	The notion of restriction category is a useful notion since the related concepts of inverse and dagger categories are easily understood through this notion.
	
	\begin{definition}\cite{cockettlack}
		Partial Isomorphism: A morphism $f:X \rightarrow Y$ is called partial isomorphism in a restriction category if there is is a unique morphism $f^{\circ}: Y \rightarrow X$ (the partial inverse of $f$) such that $f \circ f^{\circ} = \overline{f^{\circ}}$ and $f^{\circ} \circ f = \overline{f}$.
	\end{definition}
	\begin{definition}\cite{cockettlack}
		Inverse category: A restriction category with partial isomorphisms as the only morphisms is termed inverse category. 
	\end{definition}
	In Par$((\mathbf{LocMeas},\mathcal{M}),\mathcal{M})$, every map is a partial isomorphism making it an example of inverse category.
	Intuitively inverse categories are `groupoids with partiality' when compared with restriction categories which are `categories with partiality'.
	
	\begin{definition}\cite{Heunen2013}
		Dagger category: A category $\mathbf{C}$ is a dagger category if it is equipped with a contravariant endofunctor $\dagger : \mathbf{C} \rightarrow \mathbf{C}^{op}$ which is identity on objects and an involution $\dagger \circ \dagger = \mathbf{id_C}$. In other words, $\dagger(1_X) = {1_X}^{\dagger} = {1_X}$ for all objects $X$ and $\dagger \circ \dagger(f) = {f}^{\dagger \dagger} = {f}$ for all morphisms $f$ in $\mathbf{C}$. 
	\end{definition}

	\subsection{Appendix for Functors and Signal Spaces}
	
	\begin{lemma}
		\label{thm:composite_meas_func}
		Let $\Sigma_X$, $\Sigma_Y$ be 
		$\sigma$-algebras of subsets of $X$ and $Y$ respectively. Let 
		$D\subseteq X$ and $\phi:D\to Y$ be a function such that 
		$\phi^{-1}[F]\in\Sigma_D$, for every 
		$F\in\Tau$, where $\Sigma_D$ is the subspace $\sigma$-algebra of $\Sigma_X$. 
		For every $[-\infty,\infty]$-real valued 
		$\Sigma_Y$-measurable function $g$ defined on $C \subseteq Y$,
		the composite function $g\phi$ is $\Sigma_X$-measurable.  	
	\end{lemma} 
	\begin{proof}
		Let $A=\dom g\phi=\phi^{-1}[C]$ and $a\in\Bbb R$. Since $g$ is $\Sigma_Y$-measurable,
		by Definition~\ref{def:subspace_sigma_meas_func} there exists an $F\in\Sigma_Y$ such that $\{y:g(y)\le a\}=F\cap C$. On the other hand, 
		there exists an $E\in\Sigma_X$ such that $\phi^{-1}[F]=E\cap D$. Hence $\{x:g\phi(x)\le a\}=A\cap E\in\Sigma_A$. Since $a$ is arbitrary, $g\phi$ is $\Sigma_X$-measurable.
	\end{proof} 
	
	\begin{theorem}~\cite{fremlinmt1} 
		\label{thm:prop_meas_func}
		Let $X$ be any set and $\Sigma_X$ a $\sigma$-algebra
		of subsets of $X$.   Let $f$ and $g$ be real-valued functions defined on
		domains $\dom f$, $\dom g\subseteq X$.
		
		(a) If $f$ is constant then it is measurable.
		
		(b) If $f$ and $g$ are measurable, then $f+g$ is measurable, where
		$(f+g)(x)=f(x)+g(x)$ for $x\in \dom f\cap\dom g$.
		
		(c) If $f$ is measurable and $c\in\Bbb R$ a scalar, then $cf$ is measurable,
		where $(cf)(x)=c\cdot f(x)$ for $x\in \dom f$.
		
		(d) If $f$ and $g$ are measurable, then $f\times g$ is measurable, where $(f\times
		g)(x)=f(x)\times g(x)$ for $x\in \dom f\cap\dom g$.
		
		(e) Let $\langle f_n\rangle_{n\in\Bbb N}$ is a sequence of
		$\Sigma_X$-measurable real-valued functions with domains included in $X$.
		Let {$(\sup_{n\in\Bbb N}f_n)(x)=\sup_{n\in\Bbb N}f_n(x)$} for all those 
		$x\in\bigcap_{n\in\Bbb N}\dom f_n$ for which the
		supremum exists in $\Bbb R$.   Then $\sup_{n\in\Bbb N}f_n$ is
		measurable.
		
		(f) If $f$ is measurable and $h$ is  a Borel measurable function from a
		subset $\dom h$
		of $\Bbb R$ to $\Bbb R$, then $hf$ is measurable, where
		$(hf)(x)=h(f(x))$ for $x\in\dom(hf)=\{y:y\in\dom f,\,f(y)\in\dom h\}$.
		
		(g) If $f$ is measurable and $E\subseteq\Bbb R$ is a Borel set, then
		there is an
		$F\in\Sigma$ such that $f^{-1}[E]=\{x:f(x)\in E\}$ is equal to $F\cap
		\dom f$.
		
		(h) If $f$ is measurable and $A$ is any set, then $f\restr A$ is
		measurable, where $\dom(f\restr A)=A\cap\dom f$ and
		$(f\restr A)(x)=f(x)$ for $x\in A\cap\dom f$.
	\end{theorem}
	\begin{proof}
		Since this is a very basic and central theorem of measure theory we sketch the technique of proof 
		for only {(a)} and {(b)}. These are found in almost every text on measure theory.
		Let $\Sigma_E$ denote the subspace
		$\sigma$-algebra of subsets of $E\subseteq X$.
		
		\medskip
		
		\quad(a) Let $f(x)=c$, a constant for every $x\in \dom f$, then
		$\{x:f(x)<a\}=\dom f$ if $c<a$, else $\emptyset$. Since both belong to $\Sigma_{\dom f}$, $f$
		is $\Sigma_X$-measurable.

		\medskip
		
		\quad(b) Consider $\{x:f(x)+g(x)\le a\}$, but $f(x)+g(x)\le a$ iff $f(x)\le a - g(x)$ iff there exists a rational number $q$ such that $f(x) \le q \le a-g(x)$ hence
		$\{x : f(x) + g(x) \le a\} = \bigcup_{q\in\Bbb Q} [f^{-1}((-\infty, q)) \cap g^{-1}((-\infty, a-q))]$. Since countable union of measurable sets is measurable, hence $f + g$ is $\Sigma_X$-measurable where $(f + g)(x) = f(x) + g(x) $ for $x \in \dom f\cap\dom g$

		Now the rest of cases may be proved using similar arguments. Refer~\cite{fremlinmt1}.
	\end{proof}

	A set $A\subseteq X$ is {\bf conegligible} if
	$X\setminus A$
	is negligible (;  that is, there is a measurable set
	$E\subseteq A$
	such that $\mu(X\setminus E)=0$).   Note that (i) $X$ is conegligible
	(ii) if $A\subseteq B\subseteq X$ and $A$ is conegligible then $B$ is
	conegligible (iii) if $\langle A_n\rangle_{n\in\Bbb N}$ is a sequence of
	conegligible sets, then $\bigcap_{n\in\Bbb N}A_n$ is conegligible.
	
	\begin{proposition}~\cite{fremlinmt4}
		\label{prop:equivalence_coneg_func}
		Let $(X,\Sigma,\mu)$ be a measure space, and $\Cal F$
		the set of
		real-valued functions whose domains are conegligible subsets of $X$.
		(i) Show that $\{(f,g):f,\,g\in\Cal F,\,f\leae g\}$ and
		$\{(f,g):f,\,g\in\Cal F,\,f\geae g\}$ are reflexive transitive relations
		on $\Cal F$, each the inverse of the other.   (ii) Show that
		$\{(f,g):f,\,g\in\Cal F,\,f\eae g\}$ is their intersection, and is an
		equivalence relation on $\Cal F$.
	\end{proposition}
	\begin{proof}
		First using the definition of conegligible set recall that when $f$ and $g$ are real-valued functions defined on
		conegligible subsets of a measure space, then we write $f\eae g$,
		$f\leae g$ or $f\geae g$ to denote, respectively,
		
		\centerline{$f=g$ a.e., which means,
			$\{x:x\in\dom(f)\cap\dom(g)$, $f(x)=g(x)\}$ is conegligible,}
		
		\centerline{$f\le g$ a.e., which means,
			$\{x:x\in\dom(f)\cap\dom(g)$, $f(x)\le g(x)\}$ is conegligible,}
		
		\centerline{$f\ge g$ a.e., which means,
			$\{x:x\in\dom(f)\cap\dom(g)$, $f(x)\ge g(x)\}$ is conegligible.}
		
		Now consider $\{(f,g):f,\,g\in\Cal F,\,f\leae g\}$, then it includes $(f,f)$
		as {$f\le f$ a.e., since, $\{x:x\in\dom(f)\cap\dom(f)$, $f(x)\le f(x)\}$ is surely conegligible,} proving
		reflexivity. Also if $(g,h)$ is included then it means {$g\le h$ a.e., which means,
			$\{x:x\in\dom(g)\cap\dom(h)$, $g(x)\le h(x)\}$ is conegligible,}. But this means 
		{$\{x:x\in\dom(f)\cap\dom(g)\cap\dom(h)$, $f(x)\le g(x)\le h(x)\}$ is conegligible,}. Now using the clause
		if $A\subseteq B\subseteq X$ and $A$ is conegligible then $B$ is conegligible, we have 
		{$\{x:x\in\dom(f)\cap\dom(h)$, $f(x)\le h(x)\}$ is conegligible,} implying $f\le h$ a.e. proving transitivity.
		Similarly once can prove the same for other relation. It is easy to observe that these are inverse of each other
		since {$\{(f,g):f,\,g\in\Cal F,\,f\leae g\} \,\iff\,\{(g,f):f,\,g\in\Cal F,\,g\geae f\}$}. Now their intersection
		consists of all real-valued functions defined on conegligible subsets such that both $f\leae g$ and $f\geae g$ is
		true which certainly means {$\{(f,g):f,\,g\in\Cal F,\,f\eae g\}$}. Reflexivity and transitivity are verified as in 
		earlier cases while symmetry newly holds now since {$\{(f,g):f,\,g\in\Cal F,\,f\eae g\} \,\iff\,\{(g,f):f,\,g\in\Cal F,\,g\eae
			f\}$} proving that it is an equivalence relation. 
	\end{proof}
	
	A real-valued function $f$ for which there is a conegligible set $E$
	such that the restriction of $f$ to $E$, $f\restr E$  is measurable, is called {\bf
		virtually measurable}. Recall $D\sim X$ if $D\symmdiff X \in \Cal N$ where $\Cal N$ is null-ideal of subsets of $X$ and
	$\sim$ is an equivalence relation on $\Sigma_X$. But since $\mu(X\symmdiff D)= \mu(X\setminus D) = 0$ for a measurable set $D \subseteq E$ where $E$ is conegligible, hence a virtually measurable function is a real-valued
	function whose domain is some member of the equivalence class $X^{\ssbullet}$.

	\begin{definition}
		\label{def:virtually_measurable_fun_space}
		Let $(X,\Sigma,\mu)$ be a measure space,
		$\mathfrak{L}^0$ or $\mathfrak{L}^0(\mu)$ is the space of
		real-valued functions $f$ defined on conegligible subsets of $X$ which are virtually
		measurable or $f\restr E$ is measurable for some conegligible set $E\subseteq X$.
	\end{definition}

	\begin{definition}
		\label{def:measurable_fun_space}
		Let $(X,\Sigma,\mu)$ be a measure space,
		$\mathfrak{L}^0_{X}$ is the space of all
		real-valued functions $f: X \rightarrow \Bbb R$ defined on $X$ which
		are measurable.
	\end{definition}
	
	\begin{remark}
		\label{rem:partial_total_measurable_2}
		Most standard references use the same symbol for a function in $\mathfrak{L}^0$ or $\mathfrak{L}^0_{X}$ and for its equivalence class in $L^0$, however it should be noted that the space $\mathfrak{L}^0$ is larger than the space $\mathfrak{L}^0_{X}$ and for every $f\in\mathfrak{L}^0$ there is a $g\in\mathfrak{L}^0_{X}$ such that $g\eae f$. Also whereas $\mathfrak{L}^0_{X}$ is a
		Dedekind $\sigma$-complete Riesz space, $\mathfrak{L}^0$ is not even a linear space since its members are not defined at every point of the underlying space therefore not quite measurable. More precisely their restrictions to some conegligible subsets are measurable which emphasized by using the term {\bf virtually measurable}.
		
		If we write $\mathfrak{N}$ for the subspace of $\mathfrak{L}^0_{X}$ consisting of measurable functions that are
		zero almost everywhere ({$f=\mathbf{0}$ a.e., means, $\{x:x\in X$, $f(x)=0\}$ is conegligible, where $\mathbf{0}$ is constant zero function on $X$}) then the quotient space $\mathfrak{L}^0_{X}/\mathfrak{N}$ is identical to the Dedekind $\sigma$-complete Riesz space $L^0(\mu)$, as ordered linear space. We shall evade the distinction between $\mathfrak{L}^0_{X}$ and $L^0(\mu)$ in most arguments of this thesis since the Dedekind $\sigma$-complete Riesz space $\mathfrak{L}^0_{X}$ parallels the Dedekind $\sigma$-complete Riesz space $L^0(\mu)$ very closely and most propositions such as Proposition~\ref{prop:L0_func_space} involving only countably many members of these spaces hold for both of them. Thus we shall deal with $L^0$ primarily and almost all of propositions and other properties involving $L^0$ are also valid for $\mathfrak{L}^0_{X}$.
	\end{remark}
	
	\begin{proposition}~\cite{fremlinmt2}
		\label{prop:properties_coneg_func}
		If $(X,\Sigma,\mu)$ is a measure space, then we have the following basic properties, corresponding to 
		Theorem~\ref{thm:prop_meas_func}:
		
		\medskip
		
		{\bf (a)} A constant (real-valued) function defined almost everywhere in $X$ belongs to $\mathfrak{L^0}$.
		
		{\bf (b)} $f+g\in\mathfrak{L^0}$ for all $f$, $g\in\mathfrak{L^0}${ (since if
			$f\restr F$ and $g\restr G$ are measurable, then
			$(f+g)\restr(F\cap G)=(f\restr F)+(g\restr G)$ is also measurable)}.
		
		{\bf (c)} $cf\in\mathfrak{L^0}$ for all $f\in\mathfrak{L^0}$ and scalar
		$c\in\Bbb R$.
		
		{\bf (d)} $f\times g\in\mathfrak{L^0}$ for all $f$, $g\in\mathfrak{L^0}$.
		
		{\bf (e)} If $\sequencen{f_n}$ is a sequence in $\mathfrak{L^0}$ and
		$f=\sup_{n\in\Bbb N}f_n$ is defined almost everywhere in $X$, then $f\in\mathfrak{L^0}$.
		
		{\bf (f)} If $f\in\mathfrak{L^0}$ and $h:\Bbb R\to\Bbb R$ is Borel measurable, then
		$hf\in\mathfrak{L^0}$

		{\bf (g)} $\mathfrak{L^0}$ is simply the set of real-valued functions, defined on
		subsets of $X$, which are equal almost everywhere
		to some $\Sigma_X$-measurable function from $X$ to $\Bbb R$.
		Hint: (i) If $g:X\to\Bbb R$ is $\Sigma_X$-measurable and
		$f\eae g$, then $F=\{x:x\in\dom f,\,f(x)=g(x)\}$ is conegligible therefore
		$f\restr F=g\restr F$ is measurable (ii)
		If $f\in\mathfrak{L^0}$, let $E\subseteq X$ be a conegligible set such that
		$f\restr E$ is measurable.   Then $A=E\cap\dom f$ is conegligible and
		$f\restr A$ is measurable, so there is a measurable $g:X\to\Bbb R$
		agreeing with $f$ on $A$ and $g\eae f$.
	\end{proposition}
	
	We leave the proof for these cases which can be found in~\cite{fremlinmt2} 241B.
	
	\begin{definition}
		\label{def:L0_space}
		For a measure space $(X,\Sigma_X,\mu)$, $\eae$ is an equivalence relation on $\mathfrak{L^0}(\mu)$.
		and $L^0(\mu)$ is defined as the set of equivalence classes in $\mathfrak{L^0}(\mu)$ under
		$\eae$. Corresponding to $f \in \mathfrak{L^0}(\mu)$, its equivalence class is denoted as $f^{\ssbullet} \in \mathfrak{L^0}(\mu)$.
	\end{definition}

	\begin{definition}
		\label{def:Lp_space}
		For a measure space $(X,\Sigma_X,\mu)$, and $p\in(1,\infty)$,
		$\mathfrak{L}^p=\mathfrak{L}^p(\mu)$ is defined as the set of functions
		$f\in\mathfrak{L}^0=\mathfrak{L}^0(\mu)$ such that $|f|^p$ is integrable,
		and $L^p=L^p(\mu)$ is defined as the set of functions
		$\{f^{\ssbullet}:f\in\mathfrak{L}^p\}\subseteq L^0=L^0(\mu)$.
	\end{definition}
\end{document}